\documentclass{CSML}

\pdfoutput=1

\usepackage{lastpage}
\def\dOi{13(4:28)2017}
\lmcsheading%
{\dOi}
{1--\pageref{LastPage}}
{}
{}
{Jan.~13, 2014}
{Dec.~20, 2017}
{}

\usepackage{xspace}
\usepackage{color}
\usepackage{listings}
\usepackage{textcomp} % helps listings
\usepackage[T1]{fontenc} % so listings double quotes are straight
\usepackage[textsize=tiny,textwidth=18ex]{todonotes}

\usepackage{url}
\usepackage{amssymb}
\usepackage[hidelinks]{hyperref}

\usepackage{tikz}

\makeatletter
\DeclareRobustCommand*\cal{\@fontswitch\relax\mathcal}
\makeatother

\theoremstyle{plain}
\theoremstyle{plain}
\theoremstyle{plain}
\theoremstyle{plain}
\theoremstyle{plain}
\theoremstyle{plain}

\newcommand{\mycomment}[1]{}

\newcommand{\Var}{\textit{Var}}
\newcommand{\Cfg}{{\it Cfg}}
\newcommand{\ra}{\rightarrow}
\newcommand{\lra}{\leftrightarrow}
\newcommand{\Ra}{\mathbin{\Rightarrow}}
\newcommand{\andx}{\mathbin{\wedge}}

\newcommand{\Prop}[0]{\textit{Prop}}
\newcommand{\Pred}{\textit{Pred}}
\newcommand{\Int}{{\it Int}}
\newcommand{\Real}{{\it Real}}
\newcommand{\Dom}{{\it Dom}}
\newcommand{\ttrue}{{\it true}}
\newcommand{\ffalse}{{\it false}}

\newcommand{\Pattern}{\textsc{Pattern}\xspace}
\newcommand{\Nat}{{\it Nat}}
\newcommand{\PNat}{{\it PNat}}
\newcommand{\Bool}{{\it Bool}}
\newcommand{\Seq}{{\it Seq}}
\newcommand{\MSet}{{\it MultiSet}}
\newcommand{\Map}{{\it Map}}
\newcommand{\Set}{{\it Set}}
\newcommand{\reverse}{{\it rev}}
\newcommand{\upto}{{\it upto}}

\newcommand{\llist}{{\it list}}
\newcommand{\listf}{{\it lseg}}
\newcommand{\tree}{{\it tree}}
\newcommand{\node}{{\it node}}
\newcommand{\Tree}{{\it Tree}}

\newcommand{\Cells}{\it Cells}
\newcommand{\Env}{\it Env}
\newcommand{\Heap}{\it Heap}

\newcommand{\SLemp}{{\it emp}}
\newcommand{\SLstar}{*}
\newcommand{\SLempM}{\SLemp_{M}}
\newcommand{\SLstarM}{\SLstar_{M}}
\newcommand{\mapstoM}{\mapsto_{M}}
\newcommand{\emptytree}{{\it none}}
\newcommand{\magicwand}{\mbox{$-\!*\,$}}

\newcommand{\FV}{{\it FV}}
\newcommand{\PL}{{\it PL}}
\newcommand{\ML}{{\it ML}}
\newcommand{\FOL}{{\it FOL}}
\newcommand{\PLmodels}{\models_\PL^=}
\newcommand{\PLvdash}{\vdash_\PL^=}

\newcommand{\K}{\ensuremath{\mathbb{K}}\xspace}
\newcommand{\ellipses}{\mathrel{\cdot\kern-2pt\cdot\kern-2pt\cdot}}
\newcommand{\mall}[2]{\langle{#2}\rangle_{\sf{#1}}}
\newcommand{\mprefix}[2]{\mall{#1}{{#2}\ \ellipses}}
\newcommand{\mmiddle}[2]{\mall{#1}{\ellipses\ {#2}\ \ellipses}}
\newcommand{\msuffix}[2]{\mall{#1}{\ellipses\ {#2}}}
\newcommand{\kall}{\mall}
\newcommand{\kprefix}{\mprefix}
\newcommand{\kmiddle}{\mmiddle}
\newcommand{\ksuffix}{\msuffix}
\newcommand{\constant}[1]{\ensuremath{\mathsf{#1}}}

\newcommand{\sequent}[4]{{#2}\ \mathrel{\vdash}_{#1} {#3} \Ra {#4}}

\newcommand{\orx}{\mathbin{\vee}}
\newcommand{\rl}[1]{
\hspace*{-2.5ex}
$
\begin{array}{rl}
\textrm{rule}\hspace*{-2ex} & {#1}
\end{array}
$
}

\newcommand{\ReflexivityRule}[0]{\textsc{Reflexivity}}

\newcommand{\AxiomRule}[0]{\textsc{Axiom}}

\newcommand{\TransitivityRule}[0]{\textsc{Transitivity}}
\newcommand{\CaseAnalysisRule}[0]{\textsc{Case Analysis}}

\newcommand{\ConsequenceRule}[0]{\textsc{Consequence}}
\newcommand{\AbstractionRule}[0]{\textsc{Abstraction}}
\newcommand{\CircularityRule}[0]{\textsc{Circularity}}

\newlength{\lstspace} % leave empty

\lstdefinelanguage{verbatim}{
	basicstyle=\ttfamily,
	aboveskip=\lstspace,
	belowskip=\lstspace,
}
\newcommand{\lstc}[1]{\text{\lstinline{#1}}}

\definecolor{Grey}{rgb}{0.8,0.8,0.8}
\definecolor{backgroundGray}{gray}{.7}

\newcommand{\mdot}{\cdot}
\newcommand{\kdot}{\mdot}

\newcommand{\inv}[1]{
\textrm{inv} \ \ ${#1}$
}
\newcommand{\graybox}[2]{
\begingroup
\setlength{\fboxsep}{1pt}
\hspace*{-1ex}
\colorbox{backgroundGray}{
\begin{minipage}{#1}
\begingroup
\vspace*{.5ex}
\setlength{\fboxsep}{1pt}
#2
\vspace*{1.3ex}
\endgroup
\end{minipage}
}\hfill
\endgroup
}

\newcommand{\grigore}[1]{}

\begin{document}

% Author macros::begin %%%%%%%%%%%%%%%%%%%%%%%%%%%%%%%%%%%%%%%%%%%%%%%%
\title[Mtching Logic]{Matching Logic\rsuper*}

\author[G.~Ro\c{s}u]{Grigore Ro\c{s}u}
\address{University of Illinois at Urbana-Champaign, USA}
\email{grosu@illinois.edu}
\thanks{
The work presented in this paper was supported in part by
% list one grant per line (separate them with comma and end with ", and <last>.")
the Boeing grant on "Formal Analysis Tools for Cyber Security" 2014-2017,
the NSF grants CCF-1218605, CCF-1318191 and CCF-1421575, and
the DARPA grant under agreement number FA8750-12-C-0284.
}

% mandatory: Please provide 1-5 keywords
% mandatory: Please choose ACM 1998 classifications from http://www.acm.org/about/class/ccs98-html . E.g., cite as "F.1.1 Models of Computation". 

\titlecomment{{\lsuper*}Extended version of an invited paper at the 26$^{\rm th}$
International Conference on Rewriting Techniques and Applications (RTA'15),
June 29 to July 1, 2015, Warsaw, Poland.}
% Author macros::end %%%%%%%%%%%%%%%%%%%%%%%%%%%%%%%%%%%%%%%%%%%%%%%%%

\begin{abstract}
This paper presents {\em matching logic}, a first-order logic (FOL)
variant for specifying and reasoning about structure by means of patterns and
pattern matching.
Its sentences, the {\em patterns}, are constructed using
{\em variables}, {\em symbols}, {\em connectives} and {\em quantifiers},
but no difference is made between function and predicate symbols.
In models, a pattern evaluates into a power-set domain (the set of values
that {\em match} it), in contrast to FOL where functions and predicates
map into a regular domain.
%Matching logic therefore makes no distinction between function and
%predicate symbols, allowing them to be uniformly used both as ``function''
%symbols in order to build terms, and as ``predicate'' symbols in order to
%impose constraints on them.
Matching logic uniformly generalizes several logical frameworks important
for program analysis, such as: propositional logic, algebraic specification,
FOL with equality, modal logic, and separation logic.
Patterns can specify separation requirements at any level
in any program configuration, not only in the heaps or stores, without
any special logical constructs for that: the very nature of pattern
matching is that if two structures are matched as part of a pattern, then
they can only be spatially separated.
Like FOL, matching logic can also be translated into pure predicate logic
with equality,
at the same time admitting its own sound and complete proof system.
A practical aspect of matching logic is that FOL reasoning with equality remains sound,
so off-the-shelf provers and SMT solvers can be used for matching logic
reasoning.
Matching logic is particularly well-suited for reasoning about programs
in programming languages that have an operational semantics, but
it is not limited to this.
\end{abstract}

\maketitle

\tableofcontents

\section{Introduction}
\label{sec:introduction}

In their simplest form, as term templates with variables,
patterns abound in mathematics and computer science.
They match a concrete, or ground, term if and only if there is
some substitution applied to the pattern's variables that makes it equal
to the concrete term, possibly via domain reasoning.
This means, intuitively, that the concrete term obeys the structure
specified by the pattern.
We show that when combined with logical connectives and variable constraints
and quantifiers, patterns provide a powerful means to specify and reason about
the structure of states, or configurations, of a programming language.

Matching logic was inspired from the domain of programming language semantics,
specifically from attempting to use operational semantics directly
for program verification.
Recently, operational semantics of several real languages have been proposed,
e.g., of C~\cite{ellison-rosu-2012-popl,hathhorn-ellison-rosu-2015-pldi},
Java~\cite{bogdanas-rosu-2015-popl},
JavaScript~\cite{bodin-etal-javascript,park-stefanescu-rosu-2015-pldi},
Python~\cite{guth-2013-thesis,pltredex-python},
PHP~\cite{k-php}, CAML~\cite{caml-ott},
thanks to the development of semantics engineering frameworks like
PLT-Redex~\cite{plt-redex}, Ott~\cite{ott-icfp},
\K~\cite{rosu-serbanuta-2010-jlap,rosu-serbanuta-2013-k}, 
etc., which make defining
an operational semantics for a programming language almost as easy as
implementing an interpreter, if not easier.
Operational semantics are comparatively easy to define and
understand,
require little formal training, scale up well, and, being
executable, can be tested.
Indeed, the language semantics above have more than 1,000
(some even more than 3,000) semantic rules and have been tested on
benchmarks/test-suites that language
implementations use to test their conformance, where
available.
Thus, operational semantics are typically used as trusted reference models
for the defined languages.
We would like to use such operational semantics of languages, {\em unchanged},
for program verification.

\begin{figure}[!t]
\centering
\includegraphics[width=3.4in]{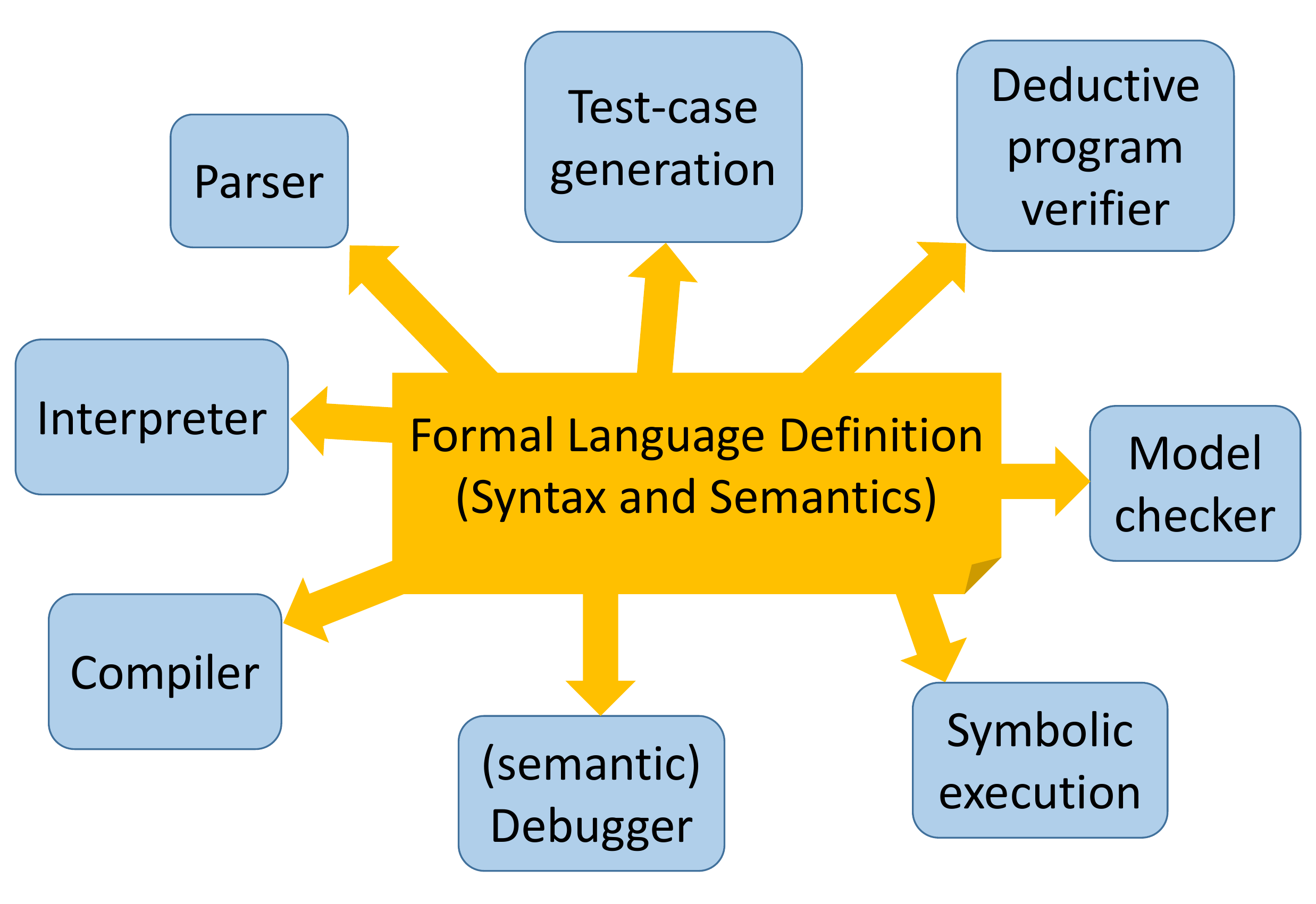}
\caption{Architecture of the \K framework, powered by matching logic}
\label{fig:dream}
\end{figure}

Despite their advantages, operational semantics are rarely used directly for
program verification, because the general belief is that proofs tend to be
low-level,
%and tedious,
as they work directly with the corresponding transition system.
%Such proofs tend to be low-level and tedious, as they involve formalizing and
%working directly with the corresponding transition system.
%Other approaches, such as
Hoare \cite{hoare-69} or dynamic \cite{Harel84dynamiclogic} logics are typically
used, because they allow higher level reasoning.
However, these come at the cost of (re)defining the language semantics as a set
of abstract proof rules, which are harder to understand and trust.
The state-of-the-art in mechanical program verification is to develop and prove
such language-specific proof systems sound w.r.t.\ a trusted operational
semantics~\cite{Nipkow-FAC98,DBLP:journals/jlp/Jacobs04,DBLP:conf/esop/Appel11},
but that needs to be done for each language separately and is labor intensive.

Defining even one complete semantics for a real language like C or Java is
already a huge effort.
Defining multiple semantics, each good for a different purpose, is at best
uneconomical, with or without proofs of soundness w.r.t.\ the reference
semantics.
It is therefore not surprising that many practical program verifiers forgo
defining a semantics altogether, and instead they implement ad-hoc
verification condition (VC) generation, sometimes via (unverified) translations
to intermediate verification languages like 
Boogie~\cite{Barnett06boogie:a} or
Why3~\cite{DBLP:conf/esop/FilliatreP13}.
For example, program verifiers for C like
VCC~\cite{DBLP:conf/tphol/CohenDHLMSST09} and
Frama-C~\cite{DBLP:conf/esop/FilliatreP13}, and for
Java like jStar~\cite{1449782} take this approach.
Also, {\em none} of the 35 verifiers that participated
in the 2016 software verification competition (SV-COMP)
\cite{DBLP:conf/tacas/Beyer16} appear to be based on a formal semantics of
any kind.
The consequence is that such tools cannot be trusted.
We would like program verifiers, ideally, to produce
proof certificates whose trust base is only an operational semantics of the
target language, same as mechanical verifiers based on Coq~\cite{Coq:manual}
or Isabelle~\cite{Nipkow:2002:IPA:1791547} do,
but without the effort to define any other semantics of the same language,
either directly as a separate proof system or indirectly by extending
the operational semantics with language-specific lemmas.
We would like program verifiers, ideally, to take an operational
semantics of a language as input and to yield, as output, a verifier
for that language which is as easy to use and as efficient as verifiers
specifically developed for that language.

%To trust a non-executable semantics of a desired language,
%an equivalence to an
%executable semantics is typically proved.
\mycomment{
Unfortunately, the current state-of-the-art in program verification is to
define yet another semantics for these languages, amenable for reasoning
about programs, such as an axiomatic or a dynamic logic semantics, because the
general belief is that operational semantics are too low level for program
verification.
Moreover, when the correctness of the verifier itself is a concern, tedious proofs
of equivalence between the operational and the alternative semantics are
produced.
That is because operational semantics are comparatively much easier to define
and at the same time are executable (and thus also testable), so they are often
considered as reference models of the corresponding languages, while the
alternative semantics devised for verification purposes tend to be more
mathematically involved and are not executable so they may hide tricky errors.
Defining even one semantics for a real language like C or Java is already a
huge effort.
Defining more semantics, each good for a different purpose, is at best very
uneconomical, with or without proofs of equivalence with the reference
semantics.
}

Matching logic was born from our belief that programming languages must
have formal definitions, and that tools for a
given language, such as interpreters, compilers, state-space
explorers, model checkers, deductive program verifiers, etc., can be
derived from just {\em one} reference formal
definition of the language, which is executable.
No other semantics for the same language should be needed.
This belief is reflected in the design of the \K framework
~\cite{rosu-serbanuta-2010-jlap,rosu-serbanuta-2013-k}
(\url{http://kframework.org}), illustrated in Figure~\ref{fig:dream}.
This is the ideal scenario and there is enough evidence that
it is within our reach in the short term.
For example, \cite{stefanescu-park-yuwen-li-rosu-2016-oopsla} presents
a program verification module of \K, based on matching logic, which takes the
respective operational semantics of
C~\cite{hathhorn-ellison-rosu-2015-pldi},
Java~\cite{bogdanas-rosu-2015-popl}, and
JavaScript~\cite{park-stefanescu-rosu-2015-pldi}
as input and yields automated program verifiers for these languages,
capable of verifying challenging heap-manipulating
programs at performance comparable to that of state-of-the-art
verifiers specifically crafted for those languages.
A precursor of this verifier,
MatchC~\cite{rosu-stefanescu-2012-oopsla},
has an online interface at \url{http://matching-logic.org}
where one can verify dozens of predefined programs or new ones;
e.g., the program in Figure~\ref{fig:matchC-example}
is under the {\tt io} folder and it takes about 150ms to verify.

\begin{figure}
\begin{footnotesize}
\begin{lstlisting}
struct listNode {  int val;  struct listNode *next;  };

void list_read_write(int n) {
\end{lstlisting}
%\vspace*{-1ex}
\graybox{335pt}{
\vspace*{1ex}
\rl{
\kprefix{code}{\constant{\$} \Rightarrow \lstc{return;}} \
\kprefix{in}{\constant{A}\Rightarrow \kdot} \ 
\ksuffix{out}{\kdot\Rightarrow\constant{rev}(\constant{A})}
\ \andx \ 
\lstc{n} = \constant{len}(\constant{A})}
}
\begin{lstlisting}
  int i=0;
  struct listNode *x=0;
\end{lstlisting}
\graybox{330pt}{
\vspace*{1ex}
\inv{
\kprefix{in}{\constant{\beta} \ \andx \ \constant{len}(\constant{\beta}) =
\lstc{n} - \lstc{i} \ \andx \ \lstc{i} \leq \lstc{n}} \
\kprefix{heap}{\constant{list}(\lstc{x}, \constant{\alpha})} \
\andx \ \constant{A} = \constant{rev}(\constant{\alpha}) @ \constant{\beta}
}}
\begin{lstlisting}
  while (i < n) {
    struct listNode *y = x;
    x = (struct listNode*) malloc(sizeof(struct listNode));
    scanf("%d", &(x->val));
    x->next = y;
    i += 1; }
\end{lstlisting}
\graybox{230pt}{
\vspace*{1ex}
\inv{
\ksuffix{out}{\constant{\alpha}} \
\kprefix{heap}{\constant{list}(\lstc{x},\constant{\beta})} \
\andx \ \constant{rev}(\constant{A}) = \constant{\alpha} @ \constant{\beta}
}
}
\begin{lstlisting}
  while (x) {
    struct listNode *y;
    y = x->next;
    printf("%d ",x->val);
    free(x);
    x = y; }
}
\end{lstlisting}
\end{footnotesize}
\vspace*{-2ex}
\caption{Reading, storing, and reverse writing a sequence of integers}
%\vspace*{-3ex}
\label{fig:matchC-example}
\end{figure}

\mycomment{
The main idea is that semantic rules match and apply on program
{\em configurations}, which are algebraic data types defined as
terms constrained by equations capturing the needed mathematical domains,
such as lists (e.g., for input/output buffers, function stacks, etc.), sets
(e.g., for concurrent threads or processes, for resources held, etc.),
maps (e.g., for environments, heaps, etc), and so on.
}

To reason about programs we need to be able to reason about program
{\em configurations}.
Specifically, we need to define configuration abstractions and reason with
them.
Consider, for example, the program in Figure~\ref{fig:matchC-example} which
shows a C function that reads \texttt{n} elements from the standard input
and prints them to the standard output in reversed order (for now, we can
ignore the specifications, which are grayed).
While doing so, it allocates a singly linked list storing the elements as
they are read, and then deallocates the list as the elements are printed.
In the end, the heap stays unchanged.
To state the specification of this program, we need to match an abstract
sequence of \texttt{n} elements in the input buffer, and then to match its
reverse at the end of the output buffer when the function terminates.
Further, to state the invariants of the two loops we need to identify a
singly linked pattern in the heap, which is a partial map.
Many such sequence or map patterns, as well as operations on them, can be
defined using conventional algebraic data types (ADTs).
But some of them cannot.

A major limitation of ADTs and of first-order logic (FOL) is that operation
symbols are interpreted as functions in models, which sometimes is insufficient.
E.g., a two-element linked list in the heap
(we regard heaps as maps from natural number locations to values)
starting with location
7 and holding values 9 and 5, written as $\llist(7,9@5)$,
can allow infinitely many heap values, one for each location where the value
5 may be stored.
So we cannot define $\llist$ as an operation symbol $\Int \times \Seq \ra \Map$.
The FOL alternative is to define $\llist$ as a predicate
$\Int \times \Seq \times \Map$, but mentioning the map all the time as an
argument makes specifications verbose and hard to read, use and reason about.
An alternative, proposed by separation
logic~\cite{reynolds-02}, is to fix and
move the map domain from explicit in models to implicit in the
logic, so that $\llist(7,9@5)$ is interpreted as a predicate
but the non-deterministic map choices are implicit in the logic.
We then may need custom separation logics for different languages
that require different variations of map models or different configurations
making use of different kinds of resources.
This may also require specialized separation logic provers needed
for each, or otherwise encodings that need to be proved correct.
Finally, since the map domain is not available as data, one cannot
use FOL variables to range over maps and thus proof rules like
``heap framing'' need to be added to the logic explicitly.

Matching logic avoids the limitations of both approaches above,
by interpreting its terms/formulae as {\em sets} of values.
Matching logic's formulae, called {\em patterns}, are built using
variables, symbols from a signature, and FOL connectives and quantifiers.
%
%We only treat the many-sorted first-order case here, but the same
%ideas can be extended to order-sorted or higher-order contexts.
%Specifically, if $(S,\Sigma)$ is a many-sorted signature and $\Var$ an $S$-sorted
%set of variables, then a pattern $\varphi$ of sort $s\in S$ can inductively be a
%variable in $\Var_s$, or have the form $\sigma(\varphi_1,\ldots,\varphi_n)$ where
%$\sigma\in\Sigma$ is a symbol of result $s$ (and arguments of any sorts) and
%$\varphi_1$, \ldots, $\varphi_n$ are patterns of appropriate sorts, or
%$\neg\varphi'$ or $\varphi' \wedge \varphi''$ or
%$\exists x . \varphi'$ where $\varphi'$ and $\varphi''$ are patterns of sort $s$ and
%$x \in \Var$ (of any sort).
%Derived constructs $\vee$, $\forall$, $\ra$, $\lra$, $\top$, $\bot$ can be
%defined as usual.
We can think of matching logic as collapsing the function and predicate
symbols of FOL, allowing patterns to be simultaneously regarded {\em both} as
terms and as predicates.
When regarded as terms they build structure, when regarded as predicates
they express constraints.
Semantically, the matching logic models are similar to the FOL models,
except that the symbols in the signature are interpreted as functions
returning sets of values instead of single values.
Patterns are then also interpreted as sets of values, where conjunction is
interpreted as intersection, negation as complement, and the existential
quantifier as union over all compatible valuations.
The name ``matching logic'' was inspired from the case when the model
is that of terms, common in the context of language semantics, where
terms represent (fragments of) program configurations.
There, a pattern is interpreted as the set of terms that match it.

The (grayed) specifications in Figure~\ref{fig:matchC-example} show 
examples of matching logic patterns, over the signature used to
define the semantics of C \cite{hathhorn-ellison-rosu-2015-pldi}.
The signature includes symbols corresponding to the syntax of the
language, to semantic constructs such as $\kall{code}{\_}$ holding the
remaining code fragment, $\kall{heap}{\_}$ holding the current
heap as a map, and $\kall{in}{\_}$ and $\kall{out}{\_}$ holding
the current input and resp.\ output buffers as sequences,
among many others.
Let us discuss the invariant pattern of the first loop (second grayed area).
It says that the pattern
$\constant{list}(\constant{x},\constant{\alpha})$ is matched somewhere
in the heap, and that the sequence $\constant{\beta}$ of size
$\constant{n}-\constant{i}$ is available at the beginning of the input
buffer such that
$\constant{A}$ is the reverse of the sequence that $\constant{x}$ points
to, $\constant{rev}(\constant{\alpha})$, concatenated with $\constant{\beta}$.
The ellipses ``$\cdots$'' are syntactic sugar for existentially quantified
variables, which we call ``structural frame variables''.
Note how symbols from the signature are mixed with logical constructs,
and how variables can range over any data stored in configurations,
including over heap fragments.
In addition to the implicit existential quantifiers for ``$\cdots$'',
the sequence $\beta$ under the $\kall{in}{\_}$ symbol is conjuncted with
logical constraints about its length; also, the pair consisting of the
$\kall{in}{\_}$ and $\kall{heap}{\_}$ patterns at the top, which is itself
a configuration pattern, is conjuncted with the equality constraint
$\constant{A} = \constant{rev}(\constant{\alpha}) @ \constant{\beta}$.
While such mixes of symbols and logical connectives are disallowed in other
logics, such as FOL or separation logic, they are not only well-formed but
also strongly encouraged to be used in matching logic;
besides succinctness of specifications, they also allow for local reasoning.
This is discussed in detail shortly, in the example in
Section~\ref{sec:example}.

Matching logic is particularly well-suited for reasoning about programs
when their language has an operational semantics.
That is because its patterns give us full access to all the details
in a program configuration, at the same time allowing us to hide
irrelevant detail using existential quantification (e.g., the ``...''
framing variables in Figure~\ref{fig:matchC-example}) or separately defined
abstractions (e.g., the $\constant{list}(\constant{x},\alpha)$ pattern
in Figure~\ref{fig:matchC-example}).
Also, both the operational semantics of a language and its reachability
properties can be encoded as rules $\varphi \Ra \varphi'$ between patterns,
called {\em reachability rules} in \cite{stefanescu-park-yuwen-li-rosu-2016-oopsla,stefanescu-ciobaca-mereuta-moore-serbanuta-rosu-2014-rta,rosu-stefanescu-ciobaca-moore-2013-lics,rosu-stefanescu-2012-oopsla}, and one generic,
language-independent proof system can be used both for executing programs
and for proving them correct.
In both cases, the operational semantics rules are used to advance
the computation.
When executing programs the pattern to reduce is ground and the
application of the semantic steps becomes conventional term rewriting.
When verifying reachability properties, the pattern to reduce is symbolic
and typically contains constraints and abstractions, so matching logic
reasoning is used in-between semantic rewrite rule applications
to re-arrange the configuration so that semantic rules match or assertions
can be proved.
We refer the interested reader to
\cite{stefanescu-park-yuwen-li-rosu-2016-oopsla} for full details on
our recommended verification approach using matching logic.

Although we favor the verification approach above, which led to the development
of matching logic, there is nothing to limit the use of matching logic with
other verification approaches, as an intuitive and succinct notation for
encoding state properties.
For example, Proposition~\ref{prop:SL} tells us that any separation logic
formula is a matching logic pattern {\em as is}.
%, without any translation or encoding.
So one can, for example, take an existing separation logic semantics of a
language, regard it as a matching logic semantics and then extend
it to also consider structures in the configuration that separation logic was
not meant to directly reason about, such as function/exception/break-continue
stacks,
%(e.g., one can prove that
%function $f$ is always only called from within function $g$, directly or
%indirectly through other function calls or function pointers, and when
%that happens there is only one thread created and nothing
%allocated in the heap),
input/output buffers, etc.
For this reason, we here present matching logic as a stand-alone
logic, without favoring any particular use of it.

This paper is an extended version of the RTA'15 conference
paper~\cite{rosu-2015-rta}, which was the first to allow
the unrestricted mix of symbols and logical quantifiers in patterns.
A much simpler variant of matching logic was introduced in
2010 in \cite{rosu-ellison-schulte-2010-amast} as a state specification logic,
and has been used since then in several verification efforts 
\cite{rosu-stefanescu-2011-nier-icse,rosu-stefanescu-2012-oopsla,rosu-stefanescu-2012-fm,rosu-stefanescu-ciobaca-moore-2013-lics,stefanescu-ciobaca-mereuta-moore-serbanuta-rosu-2014-rta,stefanescu-park-yuwen-li-rosu-2016-oopsla},
and implemented in MatchC~\cite{rosu-stefanescu-2012-oopsla} by reduction to
Maude~\cite{maude-book} (for matching) and to Z3~\cite{z3-solver}
(for domain reasoning).
However, that matching logic variant shares only the basic intuition of
``terms as formulae'' with the logic presented in this paper, and was only
syntactic sugar for first-order logic (FOL) with equality in a fixed model,
essentially allowing only term patterns $t$ and regarding them as syntactic
sugar for equalities $\square=t$ (see Section~\ref{sec:related-work}).

%The rest of the paper is organized as follows.
Section~\ref{sec:matching-logic} introduces the syntax and semantics
of matching logic, as well as some basic properties.
Sections~\ref{sec:propositional-calculus} and \ref{sec:predicate-logic}
show how propositional calculus and, respectively, pure predicate logic
fall as instances of matching logic.
Section~\ref{sec:useful-symbols} shows how several important mathematical
concepts can be defined in matching logic, such as definedness, equality,
membership, and functions.
Using these, Sections~\ref{sec:alg-spec}, \ref{sec:FOL}, \ref{sec:modal-logic}
and \ref{sec:sep-logic} then show how
algebraic specifications,
first-order logic,
modal logic and, respectively, separation logic also fall as instances of
matching logic.
Section~\ref{sec:PL-reduction} shows that, like FOL,
matching logic also reduces to pure predicate logic with equality.
Section~\ref{sec:deduction} introduces our sound and complete
proof system for matching logic.
Section~\ref{sec:related-work} discusses related work and
Section~\ref{sec:conclusion} concludes.

%\newpage

\section{Matching Logic: Basic Notions}
\label{sec:matching-logic}

%Here we introduce the basic definitions and show that, like first-order
%logic~(FOL), matching logic can also be reduced to pure predicate logic.
We assume the reader is familiar with many-sorted sets, functions, and
first-order logic (FOL).
For any given set of sorts $S$, we assume $\Var$ is an $S$-sorted set
of variables, sortwise infinite and disjoint.
We may write $x:s$ instead of $x\in\Var_s$, and when the sort of
$x$ is irrelevant we just write $x \in \Var$.
We let ${\cal P}(M)$ denote the powerset of a many-sorted set $M$, which is
itself many-sorted.
We only treat the many-sorted case here, but we see no inherent limitations
in extending the constructions and results in this paper to the
order-sorted case.

\subsection{Patterns}
We start by defining the syntax of patterns.

\begin{defi}
\label{dfn:ML-patterns}
Let $(S,\Sigma)$ be a many-sorted signature of \textbf{symbols}.
%, and assume a
%sort-wise infinite $S$-sorted set of variables $\Var$.
Matching logic $(S,\Sigma)$-\textbf{formulae}, also called
$(S,\Sigma)$-\textbf{patterns}, or just (matching logic)
\textbf{formulae} or \textbf{patterns} when $(S,\Sigma)$ is
understood from context,
are inductively defined as follows for all sorts $s\in S$:
$$
\begin{array}{rrl@{\ \ \ //\ \ }l}
\varphi_s & ::= & x\in\Var_s & \mbox{Variable}\\
& \mid & \sigma(\varphi_{s_1},...,\varphi_{s_n}) \ \mbox{ with }\sigma\in\Sigma_{s_1...s_n,s} \textrm{ \ \ (written $\Sigma_{\lambda,s}$ when $n=0$)} & \mbox{Structure}\\
& \mid & \neg \varphi_s & \mbox{Complement} \\
& \mid & \varphi_s \wedge \varphi_s & \mbox{Intersection} \\
& \mid & \exists x \,.\, \varphi_s \mbox{ with } x \in\! \Var \textrm{ \ \ (of any sort)} & \mbox{Binding}
\end{array}
$$
Let \Pattern be the $S$-sorted set of patterns.
By abuse of language, we refer to the symbols in $\Sigma$ also as
patterns: think of $\sigma\in\Sigma_{s_1 \ldots s_n,s}$ as
the pattern $\sigma(x_1\!:\!s_1,\ldots,x_n\!:\!s_n)$.
\end{defi}

We argue that the syntax of patterns above is necessary in order to express
meaningful patterns, and at the same time it is minimal.
Indeed, variable patterns allow us to extract the matched elements or
structure and possibly use them in other places in more complex patterns.
Forming new patterns from existing patterns by adding more structure/symbols
to them is standard and the very basic operation used to construct terms,
which are the simplest patterns.
Complementing and intersecting patterns allows us to reason with patterns
the same way we reason with logical propositions and formulae.
Finally, the existential binder serves a dual role.
On the one hand, it allows us to abstract away irrelevant parts of the
matched structure, which is particularly useful when defining and reasoning
about program invariants or structural framing.
On the other hand, it allows us to define complex patterns with binders in them,
such as $\lambda$-, $\mu$-, or $\nu$-bound terms/patterns
(to be presented elsewhere).

To ease notation, $\varphi\in\Pattern$ means
$\varphi$ is a pattern, while $\varphi_s\in\Pattern$ or
$\varphi\in\Pattern_s$ that it has sort $s$.
We adopt the following derived constructs (``syntactic sugar''): %\vspace*{-1ex}
$$\begin{array}{rcl@{\hspace*{10ex}}rcl}
\top_{\!\!s} & \equiv & \exists x\!:\!s \,.\, x &
\varphi_1 \ra \varphi_2 & \equiv & \neg\varphi_1 \vee \varphi_2 \\
\bot_{s} & \equiv & \neg \top_{\!\!s} &
\varphi_1 \lra \varphi_2 & \equiv & (\varphi_1 \ra \varphi_2) \wedge
 (\varphi_2 \ra \varphi_1) \\
\varphi_1 \vee \varphi_2 & \equiv & \neg(\neg\varphi_1\wedge\neg\varphi_2) &
\forall x . \varphi & \equiv & \neg(\exists x . \neg\varphi) \\
%\varphi_1 = \varphi_2 & \equiv & \varphi_1 \lra \varphi_2
\end{array}
%\vspace*{-1ex}
$$
Intuitively, $\top$ is a pattern that is matched by all elements, $\bot$ is matched
by no elements, $\varphi_1 \vee \varphi_2$ is matched by all elements matching
$\varphi_1$ or $\varphi_2$, and so on.
We will shortly formalize this intuition.
We assume the usual precedence of the FOL-like constructs, with $\neg$ binding
tighter than $\wedge$ tighter than $\vee$ tighter than $\ra$ tighter than
$\lra$ tighter than the quantifiers.

We adapt from first-order logic the notions of
free variable, (variable capture free) substitution, and
variable renaming, briefly recalled below.
Let $\FV(\varphi)$ denote the {\em free variables} of $\varphi$, defined
as follows: $\FV(x) = \{x\}$,
$\FV(\sigma(\varphi_{s_1},...,\varphi_{s_n})) = 
\FV(\varphi_{s_1}) \cup \cdots \cup \FV(\varphi_{s_n})$, 
$\FV(\neg \varphi) = \FV(\varphi)$,
$\FV(\varphi_1 \andx \varphi_2) = \FV(\varphi_1) \cup \FV(\varphi_2)$,
and
$\FV(\exists x . \varphi) = \FV(\varphi) \ \backslash\ \{x\}$.
Similarly, the usual variable capture free {\em substitution}:
$x[\varphi/x]=\varphi$ and $y[\varphi/x]=y$ when variable $y$ is different
from $x$,
$\sigma(\varphi_{s_1},...,\varphi_{s_n})[\varphi/x] =
\sigma(\varphi_{s_1}[\varphi/x],...,\varphi_{s_n}[\varphi/x])$,
$(\varphi_1 \wedge \varphi_2)[\varphi/x] =
\varphi_1[\varphi/x] \wedge \varphi_2[\varphi/x]$,
$(\neg\varphi')[\varphi/x]=\neg(\varphi'[\varphi/x])$,
and
$(\exists x.\varphi')[\varphi/x] = \exists x.\varphi'$
and
$(\exists y.\varphi')[\varphi/x] = \exists y.(\varphi'[\varphi/x])$
when variable $y$ is different from $x$ and $y\not\in\FV(\varphi)$
(to avoid variable capture).
And {\em variable renaming},
$\exists x . \varphi \equiv \exists y . (\varphi[y/x])$, which can be
used to avoid variable capture.

\subsection{Example}
\label{sec:example}

There are many examples of patterns throughout the paper resulting from
formulae in various other logics that are captured by matching logic, such as 
propositional logic (Section~\ref{sec:propositional-calculus}),
predicate logic (Section~\ref{sec:predicate-logic}),
algebraic specifications (Section~\ref{sec:alg-spec})
first-order logic (Section~\ref{sec:FOL})
modal logic (Section~\ref{sec:modal-logic}),
and separation logic (Section~\ref{sec:sep-logic}).
We will discuss them in their respective sections, showing that
formulae in these logics can be regarded as matching logic patterns.
Here, instead, we discuss an example inspired from programming language
semantics, which is the area that motivated the development of matching logic.

Consider the operational semantics of a real language like C, whose configuration
has more than 100 semantic components
\cite{ellison-rosu-2012-popl,hathhorn-ellison-rosu-2015-pldi,stefanescu-park-yuwen-li-rosu-2016-oopsla}.
The semantic components, here called ``cells'' and written using symbols
$\kall{cell}{...}$, can be nested
and their grouping (symbol) is governed by associativity and commutativity
axioms.
There is a top cell $\kall{cfg}{...}$ holding subcells
$\kall{code}{...}$,
$\kall{heap}{...}$, $\kall{in}{...}$, $\kall{out}{...}$ among many others,
holding the current code fragment, heap, input buffer, output buffer,
respectively.
We cannot show the signature of all the symbols defining the configuration of 
a language like C for space reasons, but encourage the interested reader
to check the aforementioned papers.
We only show a small subset of symbols that is sufficient to write interesting
patterns for illustration purposes, mentioning that nothing changes in the
subsequent developments of matching logic as the signature grows or changes.
That is, we do not have a matching logic for C, another for Java, another
for JavaScript, etc.; all these languages have their respective signatures
and patterns, and the same matching logic machinery applies to all of them
in the same way.

To motivate certain patterns below, we will refer to results that are introduced
later in the paper.
The purpose of this example, however, is to illustrate and discuss various kinds
of patterns,
and especially to show that it is useful to mix symbols with logical
connectives.
The hasty reader can only skim the patterns and their descriptions below
for now, and revisit the example later as other results back-reference it.

\newcommand{\Id}{\textit{Id}}
\newcommand{\Exp}{\textit{Exp}}
\newcommand{\Stmt}{\textit{Stmt}}
\newcommand{\Pgm}{\textit{Pgm}}
\newcommand{\CfgCell}{\textit{CfgCell}}
\newcommand{\CodeCell}{\textit{CodeCell}}
\newcommand{\EnvCell}{\textit{EnvCell}}
\newcommand{\HeapCell}{\textit{HeapCell}}
\newcommand{\InCell}{\textit{InCell}}
\newcommand{\OutCell}{\textit{OutCell}}

\begin{figure}
$$
\begin{array}{rcl@{\hspace*{-6ex}}r}
\multicolumn{4}{l}{
\begin{array}{rcl@{\hspace*{8ex}}r@{}}
S & = & \{\ \ \Id,\Exp, \Stmt, ... & \textrm{// code synactic categories} \\
  &   & \ \ \ \Bool, \Nat, \Int, ... & \textrm{// basic domains} \\ 
  &   & \ \ \ \Seq_\Int, \Map_{\Id,\Nat}, \Map_{\Nat,\Int}, ... & \textrm{// more domains} \\ 
  &   & \ \ \ \CfgCell, \Cfg, & \textrm{// top cell and contents} \\ 
  &   & \ \ \ \CodeCell, \EnvCell, \HeapCell, \InCell, \OutCell, & \textrm{// config cells} \\
  &   & ... \ \ \ \}
\end{array}
} \\
\Sigma_{\Id\,\Exp,\, \Stmt} & = & \{\_\texttt{=}\_\texttt{;},\ ...\} & \textrm{// assignment, ...} \\
\Sigma_{\Exp\,\Stmt\,\Stmt,\, \Stmt} & = & \{\texttt{if(}\_\texttt{)\{}\_\texttt{\}else\{}\_\texttt{\}},\ ...\} & \textrm{// conditional, ...} \\
\Sigma_{\Exp\,\Stmt,\, \Stmt} & = & \{\texttt{while(}\_\texttt{)\{}\_\texttt{\}},\ ...\} & \textrm{// while loop, ...} \\
\Sigma_{\Stmt\,\Stmt,\,\Stmt} & = & \{\ \_\_ \ \} & \textrm{// sequential composition} \\
\multicolumn{4}{l}{... \ \ \ \ \ \textrm{// other syntactic language constructs} } \\
\Sigma_{\Stmt,\,\Cfg} & = & \{\ \kall{code}{\_}\ \} & \textrm{// cell holding the code} \\
\Sigma_{\lambda,\,\Map_{\Id,\Nat}} & = & \{\ . \ \} & \textrm{// empty environment map} \\
\Sigma_{\Id\,\Nat,\Map_{\Id,\Nat}} & = & \{\ \_\mapsto\_ \ \} & \textrm{// one-binding environment map} \\
\Sigma_{\Map_{\Id,\Nat}\,\Map_{\Id,\Nat},\,\Map_{\Id,\Nat}} & = & \{\ \_,\_ \ \} & \textrm{// environment map merge} \\
\Sigma_{\Map_{\Id,\Nat},\,\Cfg} & = & \{\ \kall{env}{\_}\ \} & \textrm{// cell holding the environment map} \\
\multicolumn{4}{l}{ ... \ \ \ \ \ \Map_{\Nat,\Int} \textrm{ symbols defined similarly to } \Map_{\Id,\Nat} } \\
\Sigma_{\Map_{\Nat,\Int},\,\Cfg} & = & \{\ \kall{heap}{\_}\ \} & \textrm{// cell holding the heap map} \\
\Sigma_{\lambda,\,\Seq_\Int} & = & \{\ \epsilon \ \} & \textrm{// empty sequence} \\
\Sigma_{\Int,\,\Seq_\Int} & = & \{\ \_\ \} & \textrm{// one-integer sequence} \\
\Sigma_{\Seq_\Int\,\Seq_\Int,\,\Seq_\Int} & = & \{\ \_@\_\ \} & \textrm{// sequence concatenation} \\
\Sigma_{\Seq_\Int,\,\Nat} & = & \{\ \constant{len} \ \} & \textrm{// sequence length} \\
\Sigma_{\Seq_\Int,\,\Seq_\Int} & = & \{\ \constant{rev} \ \} & \textrm{// sequence reverse} \\
\Sigma_{\Seq_\Int,\,\Cfg} & = & \{\ \kall{in}{\_}\ \} & \textrm{// input buffer} \\
\Sigma_{\Seq_\Int,\,\Cfg} & = & \{\ \kall{out}{\_}\ \} & \textrm{// output buffer} \\
\Sigma_{\lambda,\,\Cfg} & = & \{\ . \ \} & \textrm{// empty configuration contents} \\
\Sigma_{\Cfg,\,\Cfg} & = & \{\ \_\_\ \} & \textrm{// merging configuration contents} \\
\Sigma_{\Cfg,\,\CfgCell} & = & \{\ \kall{cfg}{\_}\ \} & \textrm{// top configuration cell}
\end{array}
$$
\caption{\label{fig:signature-C}Signature for building program configurations in a C-like language}
\end{figure}

Consider the signature $(S,\Sigma)$ in Figure~\ref{fig:signature-C},
consisting of symbols needed to construct semantic configurations for a C-like
language.
Usual terms are already patterns, in particular the
first while loop in the program in Figure~\ref{fig:matchC-example}, say
\texttt{LOOP}.
So are terms with variables, e.g.:
$$
%\varphi \ \ \equiv \ \ 
\kall{cfg}{
\kall{code}{
\texttt{LOOP}\,k}
\
\kall{env}{
\texttt{x}\mapsto x,\,\texttt{n}\mapsto n,\,\texttt{i}\mapsto i,\,e}
\
\kall{heap}{
x\mapsto a,\,x+1\mapsto y,\,h
}
\
\kall{in}{
\beta}
\
\kall{out}{
\epsilon}
}
$$
The intuition for this pattern is that it matches all the configurations whose
code starts with \texttt{LOOP} ($k$, the ``code frame'', matches the rest of
the code), whose environment binds program identifiers \texttt{n} and
\texttt{i} to values $n$ and $i$, respectively, and \texttt{x} to location
$x$ ($e$, the ``environment frame'', matches the rest of the environment map)
such that both $x$ and $x+1$ are allocated and bound to some values in the
heap ($h$, the ``heap frame'', matches the rest of the heap), whose input
buffer contains some sequence ($\beta$) and whose output buffer contains the
empty sequence.
This intuition will be formalized shortly.
Also, we will show how the various symbols can be constrained or defined
axiomatically, like in algebraic (Section~\ref{sec:alg-spec}) or FOL
(Section~\ref{sec:FOL}) specifications; for example, sequences are associative and
have $\epsilon$ as unit, maps and $\Cfg$ are both associative and commutative
with ``$.$'' as unit, $\constant{len}(i@\alpha)=1+\constant{len}(\alpha)$, etc.

The interesting patterns are those combining symbols and logical connectives.
For example, suppose that we want to restrict the pattern above to only match 
configurations where $i \leq n$.
As discussed later in the paper (Section~\ref{sec:equality}), equality can be
axiomatized in matching logic and used in any sort context.
Also, due to their ubiquity, Boolean expressions are allowed to be used in
any sort context unchanged, with the meaning that they equal $\ttrue$; that is,
we write just $b$ instead of $b=\ttrue$ (Section~\ref{sec:builtins}).
With these, we can restrict the pattern above as follows (note the top-level
conjuction):
$$
%\varphi \ \ \equiv \ \ 
\kall{cfg}{
\kall{code}{
\texttt{LOOP}\,k}
\
\kall{env}{
\texttt{x}\mapsto x,\,\texttt{n}\mapsto n,\,\texttt{i}\mapsto i,\,e}
\
\kall{heap}{
x\mapsto a,\,x+1\mapsto y,\,h
}
\
\kall{in}{
\beta}
\
\kall{out}{
\epsilon}
}
\ \andx\ i \leq n
$$
Quantifiers can be used, for example, to abstract away irrelevant parts of the
pattern.
Suppose, for example, that we work in a context where the code and the output
cells are irrelevant, and so are the frames of the environment and heap cells.
Then we can ``hide'' them to the context as follows:
$$
(\exists c\,.\,\exists e\,.\,\exists h\,.\,
\kall{cfg}{
\kall{env}{
\texttt{x}\mapsto x,\,\texttt{n}\mapsto n,\,\texttt{i}\mapsto i,\,e}
\
\kall{heap}{
x\mapsto a,\,x+1\mapsto y,\,h
}
\
\kall{in}{
\beta}
\
c
})
\ \andx\ i \leq n
$$
Following a notational convention proposed and implemented in \K
(\url{http://kframework.org}
\cite{rosu-serbanuta-2010-jlap,rosu-serbanuta-2013-k}), we use 
``...'' as syntactic sugar for such
existential quantifiers used for framing:
$$
\kprefix{cfg}{
\kprefix{env}{
\texttt{x}\mapsto x,\,\texttt{n}\mapsto n,\,\texttt{i}\mapsto i}
\
\kprefix{heap}{
x\mapsto a,\,x+1\mapsto y
}
\
\kall{in}{
\beta}
}
\ \andx\ i \leq n
$$
It is often the case that program identifiers are bound to default
mathematical variables (their symbolic values) in the environment,
and then the mathematical variables are used in many other parts of
the configuration pattern to state additional logical or structural
constraints.
For that reason, we typically want to match the program identifiers
to their (symbolic) values once and for all with a separate, default
(sub)pattern, which is then not mentioned anymore in subsequent
patterns:
$$
\begin{array}{l}
\kprefix{cfg}{
\kprefix{env}{
\texttt{x}\mapsto x,\,\texttt{n}\mapsto n,\,\texttt{i}\mapsto i}
} \ \ \ \ \ \ \ \ \ \ \textrm{// assumed by default below}

\\
\ \andx \ 
\kprefix{cfg}{
\kprefix{heap}{
x\mapsto a,\,x+1\mapsto y
}
\
\kall{in}{
\beta}
}
\ \andx\ i \leq n
\end{array}$$
Note that the pattern above contains two (top-level) $\kall{cfg}{...}$
sub-pattern constraints and one logical constraint, $i \leq n$.
This pattern will be matched by precisely those configurations that
match both sub-patterns and satisfy the constraint, which are the
same configurations that match the previous pattern.
Therefore, the last two patterns are equal
(pattern equality is formalized in Section~\ref{sec:equality}; see
Notation~\ref{notation:equality} and Proposition~\ref{prop:equality})).

Now suppose that we want to state that location $x$ in the heap points
to a linked list over the list data-structure in the program in
Figure~\ref{fig:matchC-example}, which comprises a mathematical sequence
of integers $\alpha$.
The precise locations of the various nodes in the list are irrelevant.
Such a linked-list pattern can be defined by adding a symbol representing it
to the signature, say $\llist\in\Sigma_{\Nat\,\Seq_\Int,\Map_{\Nat,\Int}}$,
together with two axioms
(similar to those in separation logic, Section~\ref{sec:sep-logic}); it is
shown in Section~\ref{sec:map-patterns} that the pattern $\llist(x,\alpha)$
is matched by precisely all the (infinitely many) linked lists starting with
location $x$ and containing the sequence of elements $\alpha$.
This shows why we want pattern symbols to be interpreted into power-set
domains, so they can evaluate to sets of elements (all those that match them)
instead of just elements.
Matching logic also allows us to axiomatically state that a symbol is to be
interpreted as a function (Section~\ref{sec:functions}); in fact, in this simple
example we assume all the symbols of our signature $\Sigma$ above to be
constrained to be functions, except for those of $\Map$ results.
We can now refine the pattern above as follows:
$$
\begin{array}{l}
\kprefix{cfg}{
\kprefix{heap}{
\llist(x,\alpha)
}
\
\kall{in}{
\beta}
}
\ \andx\ i \leq n
\end{array}$$
Inspired from the invariant of the first loop in
Figure~\ref{fig:matchC-example}, let us add some more constraints:
$$
\begin{array}{l}
\kprefix{cfg}{
\kprefix{heap}{
\llist(x,\alpha)
}
\
\kprefix{in}{
\beta}
}
\ \andx \ \constant{len}(\beta) = n - i
\ \andx\ i \leq n
\ \andx\ A = \constant{rev}(\alpha)@\beta
\end{array}$$
The pattern above is additionally stating that the
$\kall{in}{...}$ cell starts with a prefix of size equal to $n-i$
which appended to the reverse of the sequence that $x$ points to in the heap
equals the original input sequence $A$.
We can arrange the pattern to better localize the logical constraints to the
sub-patterns for which they are relevant.
For example, the first two constraints are relevant for the sequence $\beta$, so
we can move them to their place:
$$
\begin{array}{l}
\kprefix{cfg}{
\kprefix{heap}{
\llist(x,\alpha)
}
\
\kprefix{in}{
\beta
\ \andx \ \constant{len}(\beta) = n - i
\ \andx\ i \leq n
}
}
\ \andx\ A = \constant{rev}(\alpha)@\beta
\end{array}$$
The above transformation is indeed correct, thanks to
Proposition~\ref{prop:constraint-propagation} (constraint propagation).
Similarly, the remaining constraint can be localized to the two
cells that need it.
Using also the fact that cell concatenation is commutative, we rewrite the pattern
into:
$$
\begin{array}{l}
\kprefix{cfg}{
\kprefix{in}{
\beta
\ \andx \ \constant{len}(\beta) = n - i
\ \andx\ i \leq n
}
\
\kprefix{heap}{
\llist(x,\alpha)
}
\ \andx\ A = \constant{rev}(\alpha)@\beta
}
\end{array}$$
The pattern above is very similar to the first invariant in
Figure~\ref{fig:matchC-example}; the latter does not mention
the top $\kall{cfg}{...}$ cell because our implementation
adds it automatically.
The top cell is not necessary anyway, we added it mostly for
uniformity in our notation for configurations.

So constraints can be propagated up and down a pattern to where they are
needed.
But how are the constraints generated?
One way to generate
constraints is through reasoning using language semantic rules, such as the
case analysis and consequence rules in \cite{stefanescu-park-yuwen-li-rosu-2016-oopsla}.
Another way to generate constraints is by local reasoning about patterns.
For example, using the axioms of $\llist$ in Section~\ref{sec:map-patterns},
we can infer $\llist(x,\alpha) \ra \varphi_1 \orx \varphi_2$, where
$\varphi_1$ is $\cdot \andx (\alpha = \epsilon)$ (the empty map ``$\cdot$'' with constraint
``$\alpha$ is the empty sequence'') and $\varphi_2$ is
$\exists a\,.\,\exists\gamma\,.\,\alpha=a@\gamma \andx \exists y\,.\,(x\mapsto a,\,x+1\mapsto y,\,\llist(y,\gamma))$.
By Proposition~\ref{prop:structural-framing} (structural framing) and propositional
reasoning we can then infer the following pattern:
$$
\begin{array}{rl}
 & \kprefix{cfg}{
\kprefix{in}{
\beta
\ \andx \ \constant{len}(\beta) = n - i
\ \andx\ i \leq n
}
\
\kprefix{heap}{
\llist(x,\alpha)
}
\ \andx\ A = \constant{rev}(\alpha)@\beta
} \\
\ra & 
\kprefix{cfg}{
\kprefix{in}{
\beta
\ \andx \ \constant{len}(\beta) = n - i
\ \andx\ i \leq n
}
\
\kprefix{heap}{
\varphi_1 \orx \varphi_2
}
\ \andx\ A = \constant{rev}(\alpha)@\beta
}
\end{array}$$
Since symbol application distributes over $\orx$ (Proposition~\ref{prop:symbol-distributivity}),
the pattern to the right of $\ra$ above becomes (again via propositional reasoning):
$$
\begin{array}{rl}
 & \kprefix{cfg}{
\kprefix{in}{
\beta
\ \andx \ \constant{len}(\beta) = n - i
\ \andx\ i \leq n
}
\
\kprefix{heap}{
\varphi_1
}
\ \andx\ A = \constant{rev}(\alpha)@\beta
}
 \\
\orx & 
\kprefix{cfg}{
\kprefix{in}{
\beta
\ \andx \ \constant{len}(\beta) = n - i
\ \andx\ i \leq n
}
\
\kprefix{heap}{
\varphi_2
}
\ \andx\ A = \constant{rev}(\alpha)@\beta
}
\end{array}$$
We can now propagate the constraints of each of $\varphi_1$ and $\varphi_2$ up into
their respective disjunct above, to be used in combination with the other constrains
on sequences.

We stop here with our example.
Note that we made no effort above to construct a signature
that does not allow junk configurations (for example, there is nothing to stop us from
adding two or more heaps in a configuration); such junk configurations can be dismissed
either by adding stronger sorting or by well-formedness predicates/patterns.
Also, our syntax for empty maps (``$\cdot$'') and for map merging (``$\_,\_$'') above is
different from that in Section~\ref{sec:map-patterns}.
The syntax above is close to the one we use in our \K implementation, while the syntax
in Section~\ref{sec:map-patterns} was specifically chosen to be similar to that
of separation logic in order to support the subsequent results in
Section~\ref{sec:SLasML}.

\subsection{Semantics}

In their simplest form, as terms with variables, patterns are usually matched
by other terms that have more structure, possibly by ground terms.
However, sometimes we may need to do the matching modulo some background
theories or modulo some existing domains, for example integers where
addition is commutative or $2+3 = 1+4$, etc.
For maximum generality, we prefer to impose no theoretical restrictions on the
models in which patterns are interpreted, or matched, leaving such restrictions
to be dealt with in implementations (for example, one may limit to free models,
or to ones for which decision procedures exist, etc.).
This has the additional benefit that it yields complete deduction
(Section~\ref{sec:deduction}).

\begin{defi}
\label{dfn:ML-model}
A \textbf{matching logic $(S,\Sigma)$-model} $M$, or just a
\textbf{$\Sigma$-model} when $S$ is understood, or simply a \textbf{model}
when both $S$ and $\Sigma$ are understood, consists of:
\begin{enumerate}
\item An $S$-sorted set $\{M_s\}_{s\in S}$, where each set $M_s$,
called the \textbf{carrier of sort $s$ of $M$}, is assumed
non-empty; and
\item A function $\sigma_M:M_{s_1}\times\cdots\times M_{s_n}
\rightarrow {\cal P}(M_s)$ for each symbol
$\sigma\in\Sigma_{s_1\ldots s_n,s}$, called the \textbf{interpretation}
of $\sigma$ in $M$.
\end{enumerate}
\end{defi}

Note that symbols are interpreted as relations, and that the usual
$(S,\Sigma)$-algebra models are a special case of matching logic
models, where $|\sigma_M(m_1,\ldots,m_n)|=1$ for any $m_1\in M_{s_1}$,
\dots, $m_n \in M_{s_n}$.
Similarly, partial $(S,\Sigma)$-algebra models also fall as special case,
where $|\sigma_M(m_1,\ldots,m_n)|\leq 1$, since we can capture the
undefinedness of $\sigma_M$ on $m_1$, \dots, $m_n$ with
$\sigma_M(m_1,\ldots,m_n) = \emptyset$.
We tacitly use the same notation $\sigma_M$ for its extension
to argument sets,
${\cal P}(M_{s_1})\times\cdots\times {\cal P}(M_{s_n}) \ra {\cal P}(M_s)$, that is,
$$\sigma_M(A_1,\ldots,A_n) =
 \textstyle\bigcup\{\sigma_M(a_1,\ldots,a_n) \mid a_1 \in A_1,\ldots, a_n \in A_n\}$$
where $A_1 \subseteq M_{s_1}, \ldots, A_n \subseteq M_{s_n}$.
%Also, when we write ``$M$'' using a calligraphic font, we tacitly
%assume the carriers and interpretations written as above,
%using non-calligraphic font, and write
%${M}=(\{M_s\}_{s\in S},\{\sigma_M\}_{\sigma\in\Sigma})$.

\begin{defi}
\label{def:rho-bar}
Given a model $M$ and a map
$\rho:\Var\rightarrow M$, called an \textbf{$M$-valuation}, let its extension
$\overline{\rho}:\Pattern\ra{\cal P}(M)$
be inductively defined as follows:
\begin{itemize}
\item $\overline{\rho}(x) = \{\rho(x)\}$, for all $x\in\Var_s$
\item $\overline{\rho}(\sigma(\varphi_{1},\ldots,\varphi_{n}))=
\sigma_M(\overline{\rho}(\varphi_1),\ldots \overline{\rho}(\varphi_n))$ for all
$\sigma\in\Sigma_{s_1...s_n,s}$ and appropriate $\varphi_1,...,\varphi_n$
%\item $\overline(\rho)(\top_{\!\!s}) = M_s$
\item $\overline{\rho}(\neg\varphi) = M_s \ \backslash\ \overline{\rho}(\varphi)$ for all
$\varphi\in\Pattern_s$
\item $\overline{\rho}(\varphi_1 \wedge \varphi_2) =
\overline{\rho}(\varphi_1) \cap \overline{\rho}(\varphi_2)$
for all $\varphi_1, \varphi_2$ patterns of the same sort
\item $\overline{\rho}(\exists x . \varphi) =
\bigcup\{\overline{\rho'}(\varphi) \mid \rho':\Var\rightarrow M,\ 
\rho'\!\!\upharpoonright_{\Var\backslash\{x\}} =
\rho\!\!\upharpoonright_{\Var\backslash\{x\}}\}
= \bigcup_{a\in M} \overline{\rho[a/x]}(\varphi)
$
\end{itemize}
where `` $\backslash$'' is set difference,
``$\rho\!\!\upharpoonright_V$'' is
$\rho$ restricted to $V \subseteq \Var$,
and ``$\rho[a/x]$'' is map $\rho'$ with $\rho'(x)=a$ and $\rho'(y)=\rho(y)$ if
$y\neq x$.
If $a\in \overline{\rho}(\varphi)$ then we say $a$ \textbf{matches}
$\varphi$ (with witness $\rho$).
\end{defi}

It is easy to see that the usual notion of term matching is an instance of the above;
indeed, if $\varphi$ is a term with variables and $M$ is the ground term model, then
a ground term $a$ matches $\varphi$ iff there is some substitution $\rho$ such that
$\rho(\varphi)=a$.
It may be insightful to note that patterns can also be regarded as predicates, when
we think of ``$a$ matches pattern $\varphi$'' as ``predicate $\varphi$ holds in $a$''.
But matching logic allows more complex patterns than terms or predicates, and models
which are not necessarily conventional (term) algebras.

The extension of $\rho$ works as expected with the derived constructs:
\begin{itemize}
\item $\overline{\rho}(\top_{\!\!s}) = M_s$ and $\overline{\rho}(\bot_s) = \emptyset$
\item $\overline{\rho}(\varphi_1 \vee \varphi_2) =
\overline{\rho}(\varphi_1) \cup \overline{\rho}(\varphi_2)$
\item $\overline{\rho}(\varphi_1 \ra \varphi_2) =
\{m \in M_s \mid m\in\overline{\rho}(\varphi_1) \mbox{ implies }
m\in\overline{\rho}(\varphi_2)\} = 
M_s \ \backslash\ (\overline{\rho}(\varphi_1) \ \backslash\ \overline{\rho}(\varphi_2))
$
\item $\overline{\rho}(\varphi_1 \lra \varphi_2) =
\{m \in M_s \mid m\in\overline{\rho}(\varphi_1) \mbox{ iff }
m\in\overline{\rho}(\varphi_2)\} = 
M_s \ \backslash\ (\overline{\rho}(\varphi_1) \ \Delta\ \overline{\rho}(\varphi_2))
$ \\ \hfill (``$\Delta$'' is the set symmetric difference operation)
\item $\overline{\rho}(\forall x . \varphi) =
\bigcap\{\overline{\rho'}(\varphi) \mid \rho':\Var\rightarrow M,\ 
\rho'\!\!\upharpoonright_{\Var\backslash\{x\}} =
\rho\!\!\upharpoonright_{\Var\backslash\{x\}}
\} = \bigcap_{a\in M} \overline{\rho[a/x]}(\varphi)$
\end{itemize}

\begin{figure}
\begin{tikzpicture}
\filldraw[fill=lightgray] (-2,-2) rectangle (3,2);
\scope % A \ B
\fill[white] (0,0) circle (1);
\fill[lightgray] (1,0) circle (1);
\endscope
% outline
\draw (0,0) circle (1) node{$\hspace*{-5ex}\varphi_1$}
      (1,0) circle (1) node{$\hspace*{5ex}\varphi_2$};
\draw (0.5,-1.5) node{\mbox{gray area matches } $\varphi_1 \ra \varphi_2$};
\end{tikzpicture}
\hspace*{10ex}
\begin{tikzpicture}[fill=white]
\filldraw[fill=lightgray] (-2,-2) rectangle (3,2);
% left hand
\scope
\clip (-2,-2) rectangle (2,2)
      (1,0) circle (1);
\fill (0,0) circle (1);
\endscope
% right hand
\scope
\clip (-2,-2) rectangle (2,2)
      (0,0) circle (1);
\fill (1,0) circle (1);
\endscope
% outline
\draw (0,0) circle (1) node{$\hspace*{-5ex}\varphi_1$}
      (1,0) circle (1) node{$\hspace*{5ex}\varphi_2$};
\draw (0.5,-1.5) node{\mbox{gray area matches } $\varphi_1 \lra \varphi_2$};
\end{tikzpicture}
\caption{\label{fig:diagram}
Matching logic semantics of pattern implication and equivalence}
\end{figure}

Interpreting formulae as sets of elements in models is reminiscent of
modal logic, where they are interpreted as the ``worlds'' in which they hold,
and of separation logic, where they are interpreted as the ``heaps'' they match.
We discuss the relationship between matching logic and these logics in depth
in Sections~\ref{sec:modal-logic} and, respectively, \ref{sec:sep-logic}.

Therefore, the matching logic interpretation of the logical connectives is not
two-valued like in classical logics.
In particular, the interpretation of $\varphi_1 \ra \varphi_2$ is the set of
all the elements that if matched by $\varphi_1$ then are also matched by
$\varphi_2$.
One should be careful when reasoning with such non-classical logics, as
basic intuitions may deceive.
For example, the interpretation of $\varphi_1 \ra \varphi_2$ is the total set
(i.e., same as $\top$) iff all elements matching $\varphi_1$ also match
$\varphi_2$, but it is the empty set iff $\varphi_2$ is matched by no elements
(same as $\bot$) while $\varphi_1$ is matched by all elements
(same as $\top$).
If in doubt, thanks to the set-theoretical interpretation of the matching logic
connectives, we can always draw diagrams to enhance our intuition; for
example, Figure~\ref{fig:diagram} depicts the semantics of pattern implication
and of pattern equivalence.

When doing logical reasoning with patterns, we sometimes want to think of a
pattern exclusively as a ``predicate'', that is, as something which is either
true or false.
To avoid using quotes in such situations, we introduce the following:

\begin{defi}
\label{dfn:predicates}
Pattern $\varphi_s$ is an \textbf{$M$-predicate}, or a
\textbf{predicate in $M$}, iff for any
\textbf{$M$-valuation} $\rho:\Var\rightarrow M$, it is the case that
$\overline{\rho}(\varphi_s)$
is either $M_s$ (it holds) or $\emptyset$ (it does not hold).
Pattern $\varphi_s$ is a \textbf{predicate} iff it is a predicate in all
models $M$.
\end{defi}

Note that $\top_s$ and $\bot_s$ are predicates, and if $\varphi$, $\varphi_1$
and $\varphi_2$ are predicates then so are $\neg\varphi$,
$\varphi_1 \wedge \varphi_2$, and $\exists x\,.\,\varphi$.
That is, the logical connectives of matching logic preserve the predicate
nature of patterns.
Section~\ref{sec:useful-symbols} will introduce several useful predicate
constructs.

\begin{defi}
\label{dfn:satisfaction}
$M$ \textbf{satisfies} $\varphi_s$, written
${M}\models \varphi_s$, iff
$\overline{\rho}(\varphi_s) = M_s$ for all $\rho:\Var\rightarrow M$.
\end{defi}

\begin{prop}
\label{prop:simple}
Unless otherwise stated, assume the default pattern sort to be $s$.
Then:
\begin{enumerate}
\item If $\rho_1,\rho_2:\Var\rightarrow M$,
$\rho_1\!\!\upharpoonright_{\FV(\varphi)}=\rho_2\!\!\upharpoonright_{\FV(\varphi)}$
then $\overline{\rho_1}(\varphi)=\overline{\rho_2}(\varphi)$
\item If $x\in\Var_s$ then ${M} \models x$ iff $|M_s| = 1$
\item If $\sigma\in\Sigma_{s_1\ldots s_n,s}$ and $\varphi_1,\ldots,\varphi_n$
are patterns of sorts $s_1,\ldots,s_n$, respectively, then we have
${M} \models \sigma(\varphi_1,\ldots,\varphi_n)$
iff 
$\sigma_M(\overline{\rho}(\varphi_1),\ldots, \overline{\rho}(\varphi_n)) = M_s$
for any $\rho:\Var\rightarrow M$
\item ${M}\models \neg \varphi$ iff $\overline{\rho}(\varphi)=\emptyset$
for any $\rho:\Var\rightarrow M$
\item ${M} \models \varphi_1 \wedge \varphi_2$ iff 
${M} \models \varphi_1$ and ${M} \models \varphi_2$
\item If $\exists x . \varphi_s$ closed,
${M} \models \exists x . \varphi_s$ iff
$\bigcup\{\overline{\rho}(\varphi_s) \mid \rho:\Var\rightarrow M\} = M_s$;
hence, ${M} \models \exists x . x$
\item ${M} \models \varphi_1 \ra \varphi_2$ iff 
$\overline{\rho}(\varphi_1) \subseteq \overline{\rho}(\varphi_2)$
for all $\rho:\Var\rightarrow M$
\item ${M} \models \varphi_1 \lra \varphi_2$ iff 
$\overline{\rho}(\varphi_1) = \overline{\rho}(\varphi_2)$
for all $\rho:\Var\rightarrow M$
%\item ${M} \models \varphi_1 = \varphi_2$ iff 
%$\overline{\rho}(\varphi_1) = \overline{\rho}(\varphi_2)$
%for all $\rho:\Var\rightarrow M$
\item ${M} \models \forall x . \varphi$ iff
${M} \models \varphi$
\end{enumerate}
\end{prop}
\begin{proof}
The proof of each of the properties is below:
\begin{enumerate}

\item
Structural induction on $\varphi$.
The only interesting case is when $\varphi$ has the form $\exists x.\varphi'$,
so $\FV(\varphi)=\FV(\varphi')\ \backslash \ \{x\}$.
Then
$$
\begin{array}{@{}rcl}
\overline{\rho_1}(\exists x . \varphi') & = &
\bigcup\{\overline{\rho_1'}(\varphi') \mid \rho_1':\Var\rightarrow M,\ 
\rho_1'\!\!\upharpoonright_{\Var\backslash\{x\}} =
\rho_1\!\!\upharpoonright_{\Var\backslash\{x\}}
\} \\
&& \mbox{(by Definition~\ref{def:rho-bar})} \\
& = &
\bigcup\{\overline{\rho_1'}(\varphi') \mid \rho_1':\Var\rightarrow M,\ 
\rho_1'\!\!\upharpoonright_{\FV(\varphi)} =
\rho_1\!\!\upharpoonright_{\FV(\varphi)}
\} \\
&& \mbox{(by the induction hypothesis)} \\
& = &
\bigcup\{\overline{\rho_2'}(\varphi') \mid \rho_2':\Var\rightarrow M,\ 
\rho_2'\!\!\upharpoonright_{\FV(\varphi)} =
\rho_2\!\!\upharpoonright_{\FV(\varphi)}
\} \\
&& \mbox{(since $\rho_1\!\!\upharpoonright_{\FV(\varphi)}=\rho_2\!\!\upharpoonright_{\FV(\varphi)}$)} \\
& = &
\bigcup\{\overline{\rho_2'}(\varphi') \mid \rho_2':\Var\rightarrow M,\ 
\rho_2'\!\!\upharpoonright_{\Var\backslash\{x\}} =
\rho_2\!\!\upharpoonright_{\Var\backslash\{x\}}
\} \\
&& \mbox{(by the induction hypothesis)} \\
& = &
\overline{\rho_2}(\exists x . \varphi')
\end{array}
$$

\item
${M} \models x$ iff
$\overline{\rho}(x)=M_s$ for all $\rho:\Var\ra M$, 
iff
$\{\rho(x)\}=M_s$ for all $\rho:\Var\ra M$,
iff $M_s$ has only one element.

\item
${M} \models \sigma(\varphi_1,\ldots,\varphi_n)$
iff
$\overline{\rho}(\sigma(\varphi_1,\ldots,\varphi_n))=M_s$ for all valuations
$\rho:\Var\ra M$,
iff
$\sigma_M(\overline{\rho}(\varphi_1),\ldots \overline{\rho}(\varphi_n)) = M_s$
for any $\rho:\Var\rightarrow M$.

\item
${M}\models \neg \varphi$ iff $\overline{\rho}(\neg\varphi)=M_s$
for any $\rho:\Var\rightarrow M$,
iff $M_s \ \backslash\ \overline{\rho}(\varphi)=M_s$
for any $\rho:\Var\rightarrow M$,
iff $\overline{\rho}(\varphi)=\emptyset$
for any $\rho:\Var\rightarrow M$.

\item
${M} \models \varphi_1 \wedge \varphi_2$
iff $\overline{\rho}(\varphi_1 \wedge \varphi_2)=M_s$
for any $\rho:\Var\rightarrow M$,
iff $\overline{\rho}(\varphi_1) \cap \overline{\rho}(\varphi_2)=M_s$
for any $\rho:\Var\rightarrow M$,
iff $\overline{\rho}(\varphi_1)=M_s$ and $\overline{\rho}(\varphi_2)=M_s$
for any $\rho:\Var\rightarrow M$,
iff ${M} \models \varphi_1$ and ${M} \models \varphi_2$.

\item
${M} \models \exists x . \varphi_s$
iff $\overline{\rho}(\exists x . \varphi_s) = M_s$
for any $\rho:\Var\rightarrow M$,
iff 
$
\bigcup\{\overline{\rho'}(\varphi_s) \mid \rho':\Var\rightarrow M,\ 
\rho'\!\!\upharpoonright_{\Var\backslash\{x\}} =
\rho\!\!\upharpoonright_{\Var\backslash\{x\}}
\}=M_s$
for any $\rho:\Var\rightarrow M$,
iff (by the first property in this proposition, since $\FV(\varphi_s)\subseteq\{x\}$)
$
\bigcup\{\overline{\rho'}(\varphi_s) \mid \rho':\Var\rightarrow M\}=M_s$
for any $\rho:\Var\rightarrow M$,
iff
$\bigcup\{\overline{\rho}(\varphi_s) \mid \rho:\Var\rightarrow M\}=M_s$.
In particular, if $\varphi_s=x$ then
$\bigcup\{\overline{\rho}(x) \mid \rho:\Var\rightarrow M\}=M_s$,
so ${M} \models \exists x . x$.

\item
${M} \models \varphi_1 \ra \varphi_2$ iff 
$\overline{\rho}(\varphi_1 \ra \varphi_2) = M$
for all $\rho:\Var\rightarrow M$,
iff $\overline{\rho}(\neg(\varphi_1 \wedge \neg \varphi_2)) = M$
for all $\rho:\Var\rightarrow M$,
iff $\overline{\rho}(\varphi_1 \wedge \neg \varphi_2) = \emptyset$
for all $\rho:\Var\rightarrow M$,
iff $\overline{\rho}(\varphi_1) \cap (M\ \backslash\ \overline{\rho}(\varphi_2)) = \emptyset$
for all $\rho:\Var\rightarrow M$,
iff 
$\overline{\rho}(\varphi_1) \subseteq \overline{\rho}(\varphi_2)$
for all $\rho:\Var\rightarrow M$.

\item
Follows from the previous similar properties for $\wedge$ and $\ra$.

\item 
${M} \models \forall x . \varphi$ iff
$\overline{\rho}(\forall x . \varphi) =
\bigcap\{\overline{\rho'}(\varphi) \mid \rho':\Var\rightarrow M,\ 
\rho'\!\!\upharpoonright_{\Var\backslash\{x\}} =
\rho\!\!\upharpoonright_{\Var\backslash\{x\}}
\}=M$
for all $\rho:\Var\rightarrow M$,
iff $\overline{\rho'}(\varphi) = M$ for all
$\rho, \rho':\Var\rightarrow M$ with
$\rho'\!\!\upharpoonright_{\Var\backslash\{x\}} =
\rho\!\!\upharpoonright_{\Var\backslash\{x\}}$,
iff $\overline{\rho}(\varphi) = M$ for all
$\rho:\Var\rightarrow M$,  iff
${M} \models \varphi$.
\end{enumerate}
Therefore, all properties hold.
\end{proof}
Since $\exists x . x$ is satisfied by all models (by (6) above), we could have
also defined $\top$ as $\exists x . x$ instead of as $x \vee \neg x$.
Properties (9) and (2) in Proposition~\ref{prop:simple} imply that the pattern
$\forall x . x$ is satisfied precisely by the models whose carrier of the
sort of $x$ contains only one element.

Note that property
``if $\varphi$ closed then ${M}\models \neg\varphi$
iff ${M} \not \models \varphi$'', which holds in classical
logics like FOL, does not hold in matching logic.
This is because $M \models \neg\varphi$ means $\neg\varphi$ is
matched by all elements, i.e., $\varphi$ is matched by no element,
while $M \not\models \varphi$ means $\varphi$ is not matched by some elements.
These two notions are different when patterns can have more than two
interpretations, which happens when $M$ can have more than one element.
%For example, if $\varphi$ is a constant symbol, say $0$,
%then ${M}\models \neg 0$ is equivalent to $0_M=\emptyset$,
%while ${M} \not \models 0$ is equivalent to $0_M \neq M_s$.

\begin{defi}
\label{dfn:validity}
Pattern $\varphi$ is \textbf{valid}, written $\models \varphi$,
iff ${M} \models \varphi$ for all ${M}$.
If $F\subseteq\Pattern$ then ${M} \models F$ iff
${M} \models \varphi$ for all $\varphi\in F$.
$F$ \textbf{entails} $\varphi$, written $F \models \varphi$,
iff for each ${M}$, ${M} \models F$ implies ${M} \models \varphi$.
A \textbf{matching logic specification} is a triple $(S,\Sigma,F)$
with $F\subseteq\Pattern$.
\end{defi}

\subsection{Basic Properties}

A natural question is how to formally reason about patterns.
Although they can be inductively built with symbols, like terms are,
the following result says that pure predicate
logic reasoning is sound for matching logic when we regard
patterns as predicates.
By {\em pure} predicate logic we mean predicate logic with
just predicate symbols, without constants or function symbols.
As shown in Section~\ref{sec:deduction}, the Substitution axiom
of non-pure predicate logics
($\forall x\,.\,\varphi) \ra \varphi[t/x]$ is not sound when $t$
is an arbitrary matching logic pattern
(it needs to be modified to only allow patterns which interpret to
singletons).

\begin{prop}
\label{prop:PL-validity}
The following properties hold for patterns of any sort $s \in S$,
so the Hilbert-style axioms and proof rules that are sound and complete for pure
predicate logic \cite{Godel1930}, are also sound for matching logic, for any sort
(more axioms and proof rules are needed for completeness, as shown in
Section~\ref{sec:deduction}):
\begin{enumerate}
\item $\models \varphi$, where $\varphi$ is a propositional tautology over patterns of sort $s$.
\item Modus ponens: $\models \varphi_1$ and $\models \varphi_1 \ra \varphi_2$
imply $\models\varphi_2$.
\item $\models (\forall x\,.\,\varphi_1\ra\varphi_2) \ra (\varphi_1 \ra \forall x\,.\,\varphi_2)$
when $x\not\in\FV(\varphi_1)$.
\item Universal generalization: $\models\varphi$ implies $\models \forall x\,.\,\varphi$.
\item Substitution: $\models (\forall x\,.\,\varphi) \ra \varphi[y/x]$, with variable
$y\not\in\FV(\forall x\,.\,\varphi)$ of same sort as $x$.
\end{enumerate}
\end{prop}
\begin{proof}
%\todo[inline]{check the \Pred\ sort}
Indeed,
\begin{enumerate}
\item 
Let $\psi$ be a propositional tautology over propositional variables
$p_1$, ..., $p_n$, such that $\varphi$ is obtained from $\psi$ by
substituting patterns $\varphi_1$, ..., $\varphi_n$ of sort $s$ for
propositional variables $p_1$, ..., $p_n$, respectively.
Let $M$ be any matching logic model, whose carrier of sort $s$ is
$M_s$, and let $\rho$ be any $M$-valuation.
It is well-known that power-sets are Boolean algebras, in our case
$({\cal P}(M_s),\neg, \cap)$ with $\neg$ the complement w.r.t.\ $M_s$,
and that all Boolean algebras are models of propositional calculus.
Therefore, no matter how we interpret the variables $p_1$, ..., $p_n$
as subsets of $M_s$, in particular as $\overline{\rho}(\varphi_1)$, ...,
$\overline{\rho}(\varphi_n)$, respectively, the interpretation of
$\psi$ is the entire set $M_s$.
Hence, $\models \varphi$.
\item
If $\rho:\Var\ra M$ is a matching logic model valuation such that
$\overline{\rho}(\varphi_1) = M_s$ and
$\overline{\rho}(\varphi_1)\subseteq\overline{\rho}(\varphi_2)$, then
it must be that $\overline{\rho}(\varphi_2) = M_s$.
\item
By (7) in Proposition~\ref{prop:simple}, it suffices to show that
$\overline{\rho}(\forall x\,.\,\varphi_1\ra\varphi_2) \subseteq
\overline{\rho}(\varphi_1 \ra \forall x\,.\,\varphi_2)$
for any valuation $\rho:\Var\ra M$.
A stronger result (equality) holds, as expected:
$$
\begin{array}{@{}rcl}
\overline{\rho}(\forall x\,.\,\varphi_1\ra\varphi_2) & = &
\bigcap\{\overline{\rho'}(\varphi_1 \ra \varphi_2) \mid \rho':\Var\ra M,\ 
\rho'\!\!\upharpoonright_{\Var\backslash\{x\}} =
\rho\!\!\upharpoonright_{\Var\backslash\{x\}}
\} \\
%&& \mbox{(by observation following Definition~\ref{def:rho-bar})} \\
& = &
\bigcap\{
M_s \ \backslash\ (\overline{\rho'}(\varphi_1) \ \backslash\ \overline{\rho'}(\varphi_2))
\mid \rho':\Var\ra M,\ 
\rho'\!\!\upharpoonright_{\Var\backslash\{x\}} =
\rho\!\!\upharpoonright_{\Var\backslash\{x\}}
\} \\
%&& \mbox{(by observation following Definition~\ref{def:rho-bar})} \\
& = &
\bigcap\{
M_s \ \backslash\ (\overline{\rho}(\varphi_1) \ \backslash\ \overline{\rho'}(\varphi_2))
\mid \rho':\Var\ra M,\ 
\rho'\!\!\upharpoonright_{\Var\backslash\{x\}} =
\rho\!\!\upharpoonright_{\Var\backslash\{x\}}
\} \\
&& \mbox{(by (1) in Proposition~\ref{prop:simple}, because $x\not\in\FV(\varphi_1)$)} \\
& = &
M_s \ \backslash\ (\overline{\rho}(\varphi_1) \ \backslash\ \bigcap\{
\overline{\rho'}(\varphi_2))
\mid \rho':\Var\ra M,\ 
\rho'\!\!\upharpoonright_{\Var\backslash\{x\}} =
\rho\!\!\upharpoonright_{\Var\backslash\{x\}}
\} \\
&& \mbox{(set theory properties of relative complements)} \\
& = &
M_s \ \backslash\ (\overline{\rho}(\varphi_1) \ \backslash\
\overline{\rho}(\forall x\,.\,\varphi_2) \\
%&& \mbox{(by observation following Definition~\ref{def:rho-bar})} \\
& = &
\overline{\rho}(\varphi_1 \ra \forall x\,.\,\varphi_2)
\end{array}
$$
\item
Immediate by (9) in Proposition~\ref{prop:simple}.
\item
Follows by (7) and (1) in Proposition~\ref{prop:simple}, because for any
valuation $\rho:\Var\ra M$, we have
$\overline{\rho}(\varphi[y/x])=\overline{\rho'}(\varphi)$ where
$\rho':\Var\ra M$ is defined as
$\rho'\!\!\upharpoonright_{\Var\backslash\{x\}} =
\rho\!\!\upharpoonright_{\Var\backslash\{x\}}$
and $\rho'(x) = \rho(y)$ (this holds also when $y=x$), while
$\overline{\rho}(\forall x\,.\,\varphi)$ is the intersection of all
$\overline{\rho''}(\varphi)$ for all
$\rho'':\Var\ra M$ with
$\rho''\!\!\upharpoonright_{\Var\backslash\{x\}} =
\rho\!\!\upharpoonright_{\Var\backslash\{x\}}$
and $\rho'$ is one of these $\rho''$.
\end{enumerate}
Therefore, pure predicate logic reasoning can also be used to
reason about patterns.
\end{proof}

%\grigore{Check if the other implication in (4) above can be derived or not, which
%is equivalent to deriving $(\forall x\,.\,\varphi) \ra \varphi$.}

Proposition \ref{prop:PL-validity} tells us that the proof
system of pure predicate logic is actually sound for matching logic, {\em unchanged}.
That is, we do not need to attempt to translate patterns to predicate logic
formulae in order to reason about them, we can simply regard them as
predicates the way they are.
Section~\ref{sec:deduction} shows that a few additional proof rules yield
a sound and complete proof system for matching logic, similarly to how (term)
Substitution together with the
other four proof rules of pure predicate logic brings complete deduction to FOL.

Sometimes we can show that patterns are two-valued:
\begin{defi}
\label{dfn:spec-predicates}
Pattern $\varphi$ is called a \textbf{predicate} in
$(S,\Sigma,F)$, or simply a \textbf{predicate} when $(S,\Sigma,F)$ is
understood, iff it is an $M$-predicate (Definition~\ref{dfn:predicates})
in all models $M$ with $M \models F$.
\end{defi}
However, note that Proposition~\ref{prop:PL-validity} applies to any patterns,
not only to predicates.
Moreover, there are also interesting properties that appear to be very specific
to patterns and their dual logical-structural nature, and not to predicates,
such as the following:
\begin{prop}
\label{prop:structural-framing}
\textbf{(Structural Framing)}
If $\sigma\in\Sigma_{s_1\ldots s_n,s}$ and
$\varphi_i,\varphi'_i\in \Pattern_{s_i}$ such that
$\models \varphi_i \ra \varphi'_i$ for all $i\in 1\ldots n$, then
$\models \sigma(\varphi_1,\ldots,\varphi_n) \ra
 \sigma(\varphi'_1,\ldots,\varphi'_n)$.
\end{prop}
\begin{proof}
Immediate by (7) in Proposition~\ref{prop:simple}, because for any model
$M$, the extension of $\sigma_M$ as a function
${\cal P}(M_{s_1})\times\cdots\times {\cal P}(M_{s_n}) \ra {\cal P}(M_s)$
is monotone.
\end{proof}
This structural framing property generalizes to {\em positive},
or {\em monotone} contexts: if $\models \varphi \ra \varphi'$ then
$\models C[\varphi] \ra C[\varphi']$ for any positive context $C$.
By a positive/monotone context we mean a context with no negation
on the path to the placeholder.\footnote{In the context of programming
language semantics, reasoning typically happens in semantic cells in the
program configuration and the program configuration is typically a term
with variables, possibly domain-constrained, so requiring a context to be
positive is not a strong requirement.}
Indeed, except for $\neg$, the matching logic constructs are interpreted
as monotone functions over powerset domains.
Structural framing is crucial for localizing reasoning.
Consider, for example, the property 
$$\models \ (1 \mapsto 5 \SLstar 2 \mapsto 0 \SLstar 7 \mapsto 9 \SLstar 8 \mapsto 1)
\ra \llist(7,9\cdot 5)$$
proved in Section~\ref{sec:deduction} for the matching logic
specifications of maps
(which captures separation logic: Section~\ref{sec:sep-logic}).
Taking $\sigma$ as the map/heap merge operation $\_\SLstar\_$,
Proposition~\ref{prop:structural-framing} implies
$$\models \ (1 \mapsto 5 \SLstar 2 \mapsto 0 \SLstar 7 \mapsto 9 \SLstar 8 \mapsto 1 \SLstar h)
\ra \llist(7,9\cdot 5) \SLstar h$$
where $h$ is a free map/heap variable.
So the property we ``locally'' proved can be ``framed'' within any map/heap.
Of course, one can go further and ``globalize'' the property in any positive
context.
For example, consider the operational semantics of a real language like C,
whose configuration was partly discussed in the example in Section~\ref{sec:example}.
Recall from Section~\ref{sec:example} that semantic cells, written using symbols
$\kall{cell}{...}$, can be nested and their grouping (symbol) is governed by
associativity and commutativity axioms.
Also, there is a top cell $\kall{cfg}{...}$ holding a subcell
$\kall{heap}{...}$ among many others.
Proposition~\ref{prop:simple} then implies
$$
\models \ \ \kall{cfg}{\kall{heap}{1 \mapsto 5 \SLstar 2 \mapsto 0 \SLstar 7 \mapsto 9 \SLstar 8 \mapsto 1\SLstar h}\ c}
\ \ra \ 
\kall{cfg}{\kall{heap}{\llist(7,9\cdot 5)\SLstar h}\ c}
$$
where $h$ and $c$ are free variables (the ``heap'' and,
respectively, ``configuration'' frames).

As discussed in the example in Section~\ref{sec:example}, sometimes it is
useful to move the logical connectives from inside terms to the top level,
or viceversa.
While disjunction and existential quantification can be propagated both
ways through symbol applications ($\lra$), conjunction and universal
quantification weaken the pattern as they are propagated from the inside
to the outside of a symbol application ($\ra$), and negation appears to
not be movable at all:

\begin{prop}
\label{prop:symbol-distributivity}
\textbf{(Distributivity of symbol application)}
Let $\sigma\in\Sigma_{s_1\ldots s_n,s}$ and
$\varphi_i\in \Pattern_{s_i}$ for all $1\leq i \leq n$.
Pick a particular $1\leq i \leq n$.
Let $\varphi'_i\in \Pattern_{s_i}$ be another pattern of sort $s_i$ and let
$C_{\sigma,i}[\square]$ be the context
$\sigma(\varphi_1,\ldots,\varphi_{i-1},\square,\varphi_{i+1},\ldots\varphi_n)$
(a context $C[\square]$ is a pattern with one occurrence of a free variable,
``$\square$'', and $C[\varphi]$ is $C[\varphi/\square]$).
Then:
\begin{enumerate}
\item
$\models C_{\sigma,i}[\varphi_i \orx \varphi'_i] \lra
C_{\sigma,i}[\varphi_i] \orx C_{\sigma,i}[\varphi'_i]$
\item
$\models C_{\sigma,i}[\exists x\,.\,\varphi_i] \lra
\exists x\,.\,C_{\sigma,i}[\varphi_i]$, where $x\not\in\FV(C_{\sigma,i}[\square])$
\item 
$\models C_{\sigma,i}[\varphi_i \andx \varphi'_i] \ra
C_{\sigma,i}[\varphi_i] \andx C_{\sigma,i}[\varphi'_i]$
\item
$\models C_{\sigma,i}[\forall x\,.\,\varphi_i] \ra
\forall x\,.\,C_{\sigma,i}[\varphi_i]$, where $x\not\in\FV(C_{\sigma,i}[\square])$
\end{enumerate}
\end{prop}
\begin{proof}
Trivial, using the basic set properties that for any function
$f : X \ra {\cal P}(Y)$ 
(i.e., relation in $X \times Y$), if $\{A_i\}_{i\in {\cal I}}$ is a family
of subsets of $X$, i.e., $A_i\subseteq X$ for all $i \in {\cal I}$, then
$f(\bigcup\{A_i \mid i \in {\cal I}\}) = \bigcup\{f(A_i) \mid i \in {\cal I}\}$
and $f(\bigcap\{A_i \mid i \in {\cal I}\}) \subseteq \bigcap\{f(A_i) \mid i \in {\cal I}\}$,
where $f(A) = \bigcup\{f(a) \mid a \in A\}$.
Note the inclusion for intersection, as opposed to equality for disjunction.
The inclusion for intersection becomes equality when $f$ is injective as a
relation, that is, when $f(a)\cap f(a') \neq \emptyset$ implies $a=a'$.
\end{proof}

The other implications in (3) and (4) above in
Proposition~\ref{prop:symbol-distributivity} do not hold in general.
%We show a counter-example to the former and leave it as an exercise
%to the reader to find a counter-example to the latter.
Consider a signature $\Sigma$ containing only one sort, two constants
$a$ and $b$, and a binary symbol $f$.
Consider also a model $M$ containing only two elements, $a_M$ and $b_M$,
with constants $a$ and $b$ interpreted as $\{a_M\}$ and $\{b_M\}$,
respectively, and with $f$ interpreted as the injective function
$f_M(a_M,a_M) = \{a_M\}$,
$f_M(b_M,a_M) = \{b_M\}$,
$f_M(a_M,b_M) = \{b_M\}$,
$f_M(b_M,b_M) = \{a_M\}$.
Let $C_{f,2}[\square]$ be the context $f(a \vee b, \square)$ and let
$\varphi_2$ and $\varphi_2'$ be $a$ and $b$, respectively.
Then the pattern $C_{f,2}[a] \wedge C_{f,2}[b]$, that is
$f(a\vee b, a) \wedge f(a \vee b, b)$, is interpreted by any valuation
to $M$ as the (total) set $\{a_M,b_M\}$, while
$C_{f,2}[a \wedge b]$, that is $f(a \vee b, a \wedge b)$, as the empty set
(because $a \wedge b$ is interpreted as the empty set).
Therefore,
$\not\models C_{f,2}[a] \wedge C_{f,2}[b] \ra C_{f,2}[a \wedge b]$.
Similarly,
$\not\models \forall x\,.\,C_{f,2}[x] \ra C_{f,2}[\forall x\,.\,x]$
because $\forall x\,.\,C_{f,2}[x]$ and $C_{f,2}[\forall x\,.\,x]$
are interpreted as $\{a_M,b_M\}$ and $\emptyset$, respectively,
by any valuation to $M$.

The reason for which the counter-examples above worked was that the context
$C_{f,2}[\square]$, that is $f(a \vee b, \square)$, did not yield
an injective relation in $M$: indeed, it was not the case that the
interpretations of $f(a \vee b, x)$ and $f(a \vee b, y)$ were disjoint
whenever $x$ and $y$ were interpreted as distinct elements.
We can define a general notion of injectivity, for any context
$C_{\sigma,i}[\square]$, which generalizes the usual notion of injectivity
of a function or relation:

\begin{defi}
\label{dfn:injectivity}
With the notation in Proposition~\ref{prop:symbol-distributivity},
$C_{\sigma,i}[\square]$ is \textbf{injective} in specification
$(S,\Sigma,F)$ iff
$F \models C_{\sigma,i}[x] \wedge C_{\sigma,i}[y] \ra C_{\sigma,i}[x \wedge y]$,
where $x,y\in\Var_{s_i}$ are distinct variables which do not occur in
$C_{\sigma,i}[\square]$.
We drop $(S,\Sigma,F)$ when understood.
Symbol $\sigma$ is \textbf{injective on position $i$} iff 
$C_{\sigma,i}[\square]$ is injective with $\varphi_1$, ..., $\varphi_{i-1}$,
$\varphi_{i+1}$,...,$\varphi_n$ chosen as distinct variables.
\end{defi}

It is easy to check that $\sigma$ is injective on position $i$ iff
for any model $M$ with $M \models F$,
$\sigma_M$ is injective on position $i$ as a relation in $M$.
Recall that functions are particular relations, and that injectivity is
a property of relations in general:
$R\subseteq M_{s_1}\times M_{s_n}\times M_s$ is injective
on position $1\leq i \leq$ iff
$(a_1,...,a_{i-1},a_i,a_{i+1},...,a_n,b) \in R$
and
$(a_1,...,a_{i-1},a_i',a_{i+1},...,a_n,b) \in R$
implies $a_i = a_i'$.
Regarding $\sigma_M$ as such a relation, its injectivity on position
$i$ means that
$\sigma_M(a_1,...,a_{i-1},a_i,a_{i+1},...,a_n) \cap
\sigma_M(a_1,...,a_{i-1},a_i',a_{i+1},...,a_n) \neq \emptyset$
implies $a_i = a_i'$.

\begin{prop}
\label{prop:injective-symbol-distributivity}
\textbf{(Distributivity of injective symbol application)}
With the notation in Definition~\ref{dfn:injectivity}, if
$C_{\sigma,i}[\square]$ is injective in $(S,\Sigma,F)$ and
$\varphi_i,\varphi_i' \in\Pattern_{s_i}$ then:
\begin{enumerate}
\item 
$F \models C_{\sigma,i}[\varphi_i] \andx C_{\sigma,i}[\varphi'_i] \ra
C_{\sigma,i}[\varphi_i \andx \varphi'_i]$
\item
$F \models \forall x\,.\,C_{\sigma,i}[\varphi_i] \ra
C_{\sigma,i}[\forall x\,.\,\varphi_i]$, where $x\not\in\FV(C_{\sigma,i}[\square])$
\end{enumerate}
Together with Proposition~\ref{prop:symbol-distributivity}, this implies
the full distributivity of injective contexts w.r.t. the matching logic
constructs $\wedge$, $\vee$, $\forall$, $\exists$ (but not $\neg$).
\end{prop}
\begin{proof}
Let $M$ be a model with $M \models F$ and let $\rho:\Var \ra M$ be a
valuation.

To prove the first property, let
$b \in \overline{\rho}(C_{\sigma,i}[\varphi_i] \andx C_{\sigma,i}[\varphi'_i])$,
that is,
$b \in \overline{\rho}(C_{\sigma,i}[\varphi_i])$ and
$b \in \overline{\rho}( C_{\sigma,i}[\varphi'_i])$.
Then there are $a,a' \in M_{s_i}$ such that
$a \in \overline{\rho}(\varphi_i)$ and
$b\in \overline{\rho[a/\square]}(C_{\sigma,i}[\square])$, and
$a' \in \overline{\rho}(\varphi_i')$ and
$b\in \overline{\rho[a'/\square]}(C_{\sigma,i}[\square])$.
Let $x,y\in\Var_{s_i}$ be two distinct variables that do not occur in
$C_{\sigma,i}[\square]$ and let $\rho'$ be the valuation
$\rho[a/x][a'/y]$.
Then we have $b \in \overline{\rho'}(C_{\sigma,i}[x])$ and
$b \in \overline{\rho'}(C_{\sigma,i}[y])$, that is,
$b \in \overline{\rho'}(C_{\sigma,i}[x] \wedge C_{\sigma,i}[y])$.
The injectivity hypothesis then implies
$b \in \overline{\rho'}(C_{\sigma,i}[x \wedge y])$.
Therefore, $\overline{\rho'}(x \wedge y)$ is non-empty, that is,
$\{a\} \cap \{a'\}$ is non-empty, that is, $a=a'$.
Since $a \in \overline{\rho}(\varphi_i)$ and
$a' \in \overline{\rho}(\varphi_i')$,
it follows that
$a \in \overline{\rho}(\varphi_i \wedge \varphi_i')$.
Since $b\in \overline{\rho[a/\square]}(C_{\sigma,i}[\square])$,
it follows that
$b\in \overline{\rho}(C_{\sigma,i}[\varphi_i \wedge \varphi_i'])$.

For the second, let
$b\in\overline{\rho}(\forall x\,.\,C_{\sigma,i}[\varphi_i])$, that is,
$b\in\overline{\rho[v/x]}(C_{\sigma,i}[\varphi_i])$ for all
$v\in M_{{\it sort}(x)}$, that is,
for any $v\in M_{{\it sort}(x)}$
there is some $a_v\in\overline{\rho[v/x]}(\varphi_i)$
such that
$b\in \overline{\rho[a_v/\square]}(C_{\sigma,i}[\square])$
(because $x\not\in\FV(C_{\sigma,i}[\square])$,
so 
$\overline{\rho[v/x][a_v/\square]}(C_{\sigma,i}[\square])
=
\overline{\rho[a_v/\square]}(C_{\sigma,i}[\square])$).
The injectivity of $C_{\sigma,i}[\square]$ implies that 
all such $a_v$ elements are equal.
Indeed,
let $v,v'\in M_{{\it sort}(x)}$ and
$a_v\in\overline{\rho[v/x]}(\varphi_i)$
and
$a_{v'}\in\overline{\rho[v'/x]}(\varphi_i)$
such that
$b\in \overline{\rho[a_v/\square]}(C_{\sigma,i}[\square])$
and 
$b\in \overline{\rho[a_{v'}/\square]}(C_{\sigma,i}[\square])$.
Let $z,y\in\Var_{s_i}$ be two distinct variables that do not occur
in $C_{\sigma,i}[\square]$, like in the Definition~\ref{dfn:injectivity}
of injectivity (but with $z$ instead of $x$ to avoid name collision),
and note that the above implies
$b\in\overline{\rho[a_v/z][a_{v'}/y]}(C_{\sigma,i}[z] \wedge C_{\sigma,i}[y])$.
Then Definition~\ref{dfn:injectivity} implies
$b\in\overline{\rho[a_v/z][a_{v'}/y]}(C_{\sigma,i}[z \wedge y])$.
Therefore $\overline{\rho[a_v/z][a_{v'}/y]}(z \wedge y) \neq \emptyset$,
that is, $a_v = a_{v'}$.
Since all the elements $a_v\in\overline{\rho[v/x]}(\varphi_i)$ for all
$v \in M_{{\it sort}(x)}$ are equal, it follows that there is some element
$a\in\overline{\rho}(\forall x\,.\,\varphi_i)$ such that
$a_v = a$ for all $a_v$ as above.
Moreover, 
$b\in \overline{\rho[a/\square]}(C_{\sigma,i}[\square])$,
that is,
$b\in \overline{\rho}(C_{\sigma,i}[\forall x\,.\,\varphi_i])$.
\end{proof}

The notion of context injectivity in Definition~\ref{dfn:injectivity}
is the weakest theoretical condition we were able to find in order
for the (bidirectional) distributivity of conjunction and
universal quantification to hold.
In practice, stronger conditions are met.
For example, Section~\ref{sec:constructors} discusses constructors,
which are symbols whose interpretations are injective in all their
arguments at the same time
(i.e., $\sigma_M(a_1,...,a_n) \cap \sigma_M(a_1',...,a_n')\neq\emptyset$ implies
$a_1=a_1'$, ..., $a_n=a_n'$).
Contexts corresponding to constructors are injective in the sense of
Definition~\ref{dfn:injectivity}.

We next demonstrate the usefulness of matching logic
by a series of other examples.

\section{Instance: Propositional Calculus}
\label{sec:propositional-calculus}

In Section~\ref{sec:matching-logic}, (1) in Proposition~\ref{prop:PL-validity},
we showed that propositional reasoning is sound for matching logic.
Here we go one step further and show that we can can instantiate matching
logic to become precisely propositional calculus, without any translation
needed in any direction.
The idea is to add a special sort for propositions, say $\Prop$, then to use
the already existing syntax of matching logic to build propositions as we know
them, and then to show that the existing semantics of matching logic, given by
$\models$, yields the expected semantics of propositions as we know it in
propositional calculus (let us refer to it as $\models_\Prop$).

We build a matching logic signature as follows:
$S$ contains only one sort, $\Prop$, and
$\Sigma$ is empty.
Let us also drop the existential quantifier, so that the resulting syntax of
patterns becomes exactly that of propositional calculus:
$$
\begin{array}{rrl}
\varphi & ::= & \Var_\Prop \\
& \mid & \neg\varphi \\
& \mid & \varphi \wedge \varphi
\end{array}
$$
Then the default matching logic semantics endows the resulting syntax of
propositions with the desired propositional calculus semantics:
\begin{prop}
\label{prop:propositional}
For any proposition $\varphi$, the following holds:
$\models_{\it Prop}\varphi$ iff $\models \varphi$.
\end{prop}
\begin{proof}
The implication ``$\models_{\it Prop}\varphi$ implies $\models \varphi$''
follows by (1) in Proposition~\ref{prop:PL-validity}.
For the other implication, let us suppose that $\models \varphi$ and
let $\theta : \Var_\Prop \ra \{\ttrue,\ffalse\}$ be an arbitrary propositional
valuation (it is often called a ``model'' in the literature, but we refrain
from using that terminology to avoid confusion with our notion of model).
All we have to do is show that $\theta(\varphi) = \ttrue$.
Let $M$ be the matching logic model with $M_\Prop=\{\ttrue,\ffalse\}$
and let $\rho_\theta : \Var \ra M$ be the matching logic valuation where
$\rho_\theta(x)=\theta(x)$ for each $x \in \Var_\Prop$.

Note that, unlike in propositional calculus where propositions $\psi$
evaluate to precisely one of $\ttrue$ or $\ffalse$ for any given valuation $\theta$,
in matching logic $\overline{\rho_\theta}(\psi)$ can be any of the four subsets
of $\{\ttrue,\ffalse\}$.
For example, if $x$ and $y$ are variables such that $\theta(x)=\ttrue$ 
and $\theta(y)=\ffalse$, then
$\overline{\rho_\theta}(x)=\{\ttrue\}$,
$\overline{\rho_\theta}(y)=\{\ffalse\}$,
$\overline{\rho_\theta}(\neg x)=\{\ffalse\}$,
$\overline{\rho_\theta}(\neg y)=\{\ttrue\}$,
$\overline{\rho_\theta}(x \wedge y)=\emptyset$,
$\overline{\rho_\theta}(x \vee y)=\{\ttrue,\ffalse\}$.
Nevertheless, we can inductively show that the propositional validity of a
proposition $\psi$ is dictated by the membership of $\ttrue$ to its matching logic
evaluation as a set: $\theta(\psi)=\ttrue$ iff $\ttrue\in\overline{\rho_\theta}(\psi)$.
Indeed:
if $\psi$ is a variable $x$ then $\overline{\rho_\theta}(x)=\{\theta(x)\}$,
so the property holds; 
if $\psi$ is $\neg\psi'$ then
$\overline{\rho_\theta}(\neg\psi')=\{\ttrue,\ffalse\}\backslash\overline{\rho_\theta}(\psi')$,
so $\ttrue\in\overline{\rho_\theta}(\neg\psi')$ iff
$\ttrue\not\in\overline{\rho_\theta}(\psi')$,
iff (by the induction hypothesis)
$\theta(\psi')\not=\ttrue$,
iff (by the two-valued semantics of propositional calculus)
$\theta(\neg\psi')=\ttrue$;
finally, if $\psi$ is $\psi_1 \wedge \psi_2$ then
$\ttrue\in\overline{\rho_\theta}(\psi_1 \wedge \psi_2)$ iff
$\ttrue\in\overline{\rho_\theta}(\psi_1)$ and
$\ttrue\in\overline{\rho_\theta}(\psi_2)$,
iff (by the induction hypothesis)
$\theta(\psi_1)=\ttrue$ and $\theta(\psi_2)=\ttrue$,
iff $\theta(\psi_1 \wedge \psi_2)=\ttrue$.

Now $\models \varphi$ implies $\overline{\rho_\theta}(\varphi)=\{\ttrue,\ffalse\}$,
so $\ttrue \in \overline{\rho_\theta}(\varphi)$.
By the result proved inductively above we conclude that $\theta(\varphi)=\ttrue$.
\end{proof}

An alternative way to capture propositional logic is to add a constant symbol
(i.e., a symbol without any arguments) to $\Sigma$ for each propositional
variable, like we do for modal logic in Section~\ref{sec:modal-logic}.
This is similar to how predicate logic captures propositional calculus,
namely by associating a predicate without arguments to each propositional
variable.
We leave the details as an exercise to the interested reader.

\section{Instance: (Pure) Predicate Logic}
\label{sec:predicate-logic}

Recall from Section~\ref{sec:matching-logic}, Proposition~\ref{prop:PL-validity}
and the discussion preceding it, that by {\em pure} predicate logic in this paper
we mean predicate logic or first-order logic (FOL) with only predicate symbols
(no function and no constant symbols).
Note that some works call the fragment of FOL with only constant
(i.e., zero-argument function) symbols ``predicate logic'',
others use ``predicate logic'' as a synonym for FOL.
We do not discuss the fragment of FOL with only constant symbols
in this paper, so from here on we take the liberty to refer to
``pure predicate logic'' as just ``predicate logic''.
Proposition~\ref{prop:PL-validity} showed that predicate logic reasoning is
sound for matching logic.
Similarly to propositional calculus in Section~\ref{sec:propositional-calculus},
here we go one step further and show that we can can instantiate matching
logic to become precisely predicate logic; the FOL case will be discussed in
Section~\ref{sec:FOL}.
We follow the same approach like for propositional calculus: add a special
sort for predicates, say $\Pred$, then use the already existing syntax of
matching logic to build formulae as we know them in predicate logic, and
then show that the existing semantics of matching logic, given by $\models$,
yields the expected semantics of pure predicate logic.
We let $\models_\PL$ denote the predicate logic satisfaction.

Recall that predicate logic is the fragment of first-order logic with
just predicate symbols, that is, with no function (including no constant)
and no equality symbols.
We consider only the many-sorted case here.
Formally, if $S$ is a sort set and $\Pi$ is a set of predicate symbols,
the syntax of pure predicate logic formulae is
$$
\begin{array}{rrl}
\varphi & ::= & \pi(x_1,\ldots,x_n) \mbox{ with }\pi\in\Pi_{s_1 \ldots s_n},\ x_1\in\Var_{s_1},\ ...,\ x_n\in\Var_{s_n} \\ 
& \mid & \neg \varphi \\
& \mid & \varphi \wedge \varphi \\
& \mid & \exists x \,.\, \varphi
\end{array}
$$

Without loss of generality, suppose that we can pick a fresh sort name, $\Pred$;
that is, $\Pred\not\in S$.
Let us now construct the matching logic signature
$(S\mathrel{\cup}\{{\Pred}\},\Sigma)$, where
$\Sigma_{s_1\ldots s_n,{\Pred}}=\Pi_{s_1\ldots s_n}$
are the only symbols in $\Sigma$; that is, $\Sigma$ contains precisely the
predicate symbols of the predicate logic signature, but regarded as pattern
symbols of result sort $\Pred$.
Suppose also that we disallow any variables of sort $\Pred$ in patterns.
Then the matching logic patterns of sort $\Pred$ are precisely the predicate
logic formulae, without any translation in any direction.
Moreover, the following result shows that the default matching logic semantics
endows these patterns with their desired predicate logic semantics:

\begin{prop}
\label{prop:pure-predicate}
For any predicate logic formula $\varphi$, the following holds:
$\models_\PL\varphi$ iff $\models \varphi$.
\end{prop}
\begin{proof}
That $\models_\PL\varphi$ implies $\models \varphi$
follows by Proposition~\ref{prop:PL-validity}: each
of the proof rules of the complete proof system of
(pure) predicate logic \cite{Godel1930} is sound for matching logic.
For the other implication, note that we can associate to any predicate
logic model
$M^{\it PL}=(\{M_s^{\it PL}\}_{s\in S},\{\pi_{M^{\it PL}}\}_{\pi\in\Pi})$
a matching logic model
$M^{\it ML}=(\{M_s^{\it ML}\}_{s\in S\cup\{\Pred\}},\{\pi_{M^{\it ML}}\}_{\pi\in\Sigma})$,
where $M^{\it ML}_s=M^{\it PL}_s$ for all $s\in S$ and
$M_\Pred^{\it ML}=\{\star\}$ (with $\star$ some arbitrary but fixed element) and
$\pi_{M^{\it ML}}(a_1,\ldots,a_n)=\{\star\}$ iff $\pi_{M^{\it PL}}(a_1,\ldots,a_n)$ holds,
and $\pi_{M^{\it ML}}(a_1,\ldots,a_n)=\emptyset$ otherwise.
Furthermore, we can show that for any PL formula $\varphi$, we have
$M^{\it PL}\models_\PL\varphi$ iff $M^{\it ML}\models_\ML\varphi$.
Since $\varphi$ does not contain any variables of sort $\Pred$, by (1) in
Proposition~\ref{prop:simple} it suffices to show that for any
$\rho : \Var \ra M^{\it PL}$, it is the case that
$M^{\it PL},\rho\models_\PL\varphi$ iff $\overline{\rho}(\varphi)=\{\star\}$.
We can easily show this property by structural induction on $\varphi$.
The only relatively non-trivial case is the complement construct, which shows
why it was important for $M_\Pred^{\it ML}$ to contain precisely one element:
$M^{\it PL},\rho\models_\PL\neg\varphi$ iff
$M^{\it PL},\rho\not\models_\PL\varphi$ iff
(by the induction hypothesis)
$\overline{\rho}(\varphi)\neq\{\star\}$
iff 
$\overline{\rho}(\varphi)=\emptyset$
iff
$\overline{\rho}(\neg\varphi)=\{\star\}$.

Therefore, $M^{\it PL}\models_\PL\varphi$ iff $M^{\it ML}\models_\ML\varphi$.
Since the predicate logic model $M^{\it PL}$ was chosen arbitrarily,
it follows that $\models \varphi$ implies $\models_\PL\varphi$.
\end{proof}

\section{Matching Logic: Useful Symbols and Notations}
\label{sec:useful-symbols}

Here we show how to define, in matching logic, several mathematical
instruments of practical importance, such as equality, membership, and
functions.
We also introduce appropriate notations for them, because they will be used
frequently and tacitly in the rest of the paper.

The role of this section is twofold.
On the one hand, it illustrates the expressiveness of matching logic.
Indeed, we can define all the crucial mathematical notions above as matching
logic specifications or as syntactic sugar, without any changes to the matching
logic itself
(recall, for example, that equality cannot be defined in first-order logics;
the logic itself needs to be modified into
``first-order logic with equality''---more details in
Section~\ref{sec:equality}).
On the other hand, it shows that despite the apparently non-conventional
interpretation of patterns as sets of values in matching logic, the
conventional mathematical machinery used to reason about program states is
still available, with its expected meaning.

Unless otherwise mentioned, for the rest of this section we assume an
arbitrary but fixed matching logic specification $(S,\Sigma,F)$.

\subsection{Definedness and Totality}
\label{sec:definedness}

In classical logics, the interpretation of a formula under a given valuation
is either true or false, and there is only one syntactic category for formulae
while multiple syntactic categories for data.
In contrast, matching logic patterns are interpreted as sets of values, those
that match them, where the total set corresponds to the intution of ``true'',
or $\top$, and the empty set corresponds to ``false'', or $\bot$.
Also, each matching logic syntactic category, or sort, admits both data constructs
and its own logical connectives and quantifiers.
These leave two questions open:
\begin{enumerate}
\item How can we interpret patterns in a conventional, two-valued way?
Are the patterns matched by proper (i.e., neither total nor empty)
subsets of elements true, or false?
\item How can we lift reasoning within syntactic category $s_1$ to syntactic
category $s_2$?
\end{enumerate}
These questions are particularly important when attempting to combine matching
logic reasoning with classical reasoning or provers for existing mathematical
domains.

It turns out that the above can be methodologically achieved by adding some
symbols and defining patterns for them to the matching logic specification
$(S,\Sigma,F)$.
Specifically, for any pair of sorts of interest $s_1,s_2 \in S$, which need not 
be distinct, we can add a symbol $\lceil\_\rceil_{s_1}^{s_2}$ to
$\Sigma_{s_1,s_2}$ and an axiom pattern to $F$ that makes
$\lceil\_\rceil_{s_1}^{s_2}$ behave like a {\em definedness predicate}
for any pattern of sort $s_1$, with two-valued result of sort $s_2$:
$\lceil\varphi\rceil_{s_1}^{s_2}$ is either $\bot_{s_2}$ when $\varphi$ is
$\bot_{s_1}$, or $\top_{s_2}$ otherwise (i.e., if $\varphi$ is matched by
some values of sort $s_1$).
The pattern that we can add to $F$ in order to achieve the above is in fact
unexpectedly simple: $\lceil x\!:\! s_1 \rceil_{s_1}^{s_2}$.

Although we do not need it for many of the subsequent results, to simplify the
overall presentation of the rest of the paper, from here on we tacitly work
under the following:

\begin{asm}
\label{assumption:definedness}
For any (not necessarily distinct) sorts $s_1,s_2\in S$, assume the following:
$$
\begin{array}{ll}
\lceil\_\rceil_{s_1}^{s_2} \in \Sigma_{s_1,s_2}
& \mbox{// Definedness symbol} \\
\lceil x\!:\! s_1 \rceil_{s_1}^{s_2} \in F & \mbox{// Definedness pattern}
\end{array}
$$
We call the symbols $\lceil\_\rceil_{s_1}^{s_2}$ \textbf{definedness symbols}.
\end{asm}
We next show that the definedness symbol indeed has the expected meaning:
\begin{prop}
\label{prop:definedness}
If $\varphi \in \Pattern_{s_1}$ then $\lceil\varphi\rceil_{s_1}^{s_2}$
is a predicate (Definition~\ref{dfn:spec-predicates}).  Specifically,
if $\rho : \Var \ra M$ is any valuation
then $\overline{\rho}(\lceil\varphi\rceil_{s_1}^{s_2})$ is either $\emptyset$
(i.e., $\overline{\rho}(\bot_{s_2})$)
when $\overline{\rho}(\varphi) = \emptyset$
(i.e., $\varphi$ undefined in $\rho$), or is $M_{s_2}$
(i.e., $\overline{\rho}(\top_{s_2})$) 
when $\overline{\rho}(\varphi) \neq \emptyset$ (i.e., $\varphi$ defined).
\end{prop}
\begin{proof}
By Definition~\ref{def:rho-bar},
$\overline{\rho}(\lceil\varphi\rceil_{s_1}^{s_2}) = 
(\lceil\_\rceil_{s_1}^{s_2})_M(\overline{\rho}(\varphi))$.
The definedness pattern axiom states that $\lceil x:s_1\rceil_{s_1}^{s_2}$
is valid (Assumption~\ref{assumption:definedness}), which implies
$(\lceil\_\rceil_{s_1}^{s_2})_M(m_1)=M_{s_2}$ 
for all $m_1 \in M_{s_1}$, so if there is any $m_1\in\overline{\rho}(\varphi)$
then $(\lceil\_\rceil_{s_1}^{s_2})_M(\overline{\rho}(\varphi))$ can only be $M_{s_2}$.
On the other hand, if $\overline{\rho}(\varphi) = \emptyset$ then 
$(\lceil\_\rceil_{s_1}^{s_2})_M(\overline{\rho}(\varphi)) = \emptyset$.
\end{proof}

\begin{nota}
\label{notation:totality}
We also define \textbf{totality}, $\lfloor\_\rfloor_{s_1}^{s_2}$, as a derived
construct dual to definedness:
$$
\begin{array}{lcl}
\lfloor\varphi\rfloor_{s_1}^{s_2}
& \equiv &
\neg\lceil\neg\varphi\rceil_{s_1}^{s_2}
\end{array}
$$
\end{nota}

The totality construct states that the enclosed pattern must be matched by
all values:

\begin{prop}
\label{prop:totality}
If $\varphi \in \Pattern_{s_1}$ then $\lfloor\varphi\rfloor_{s_1}^{s_2}$
is a predicate (Definition~\ref{dfn:spec-predicates}).
Specifically, if $\rho : \Var \ra M$ is any valuation
then $\overline{\rho}(\lfloor\varphi\rfloor_{s_1}^{s_2})$ is either $\emptyset$
(i.e., $\overline{\rho}(\bot_{s_2})$)
when $\overline{\rho}(\varphi) \neq M_{s_1}$
(i.e., $\varphi$ not total in $\rho$), or is $M_{s_2}$
(i.e., $\overline{\rho}(\top_{s_2})$) 
when $\overline{\rho}(\varphi) = M_{s_1}$
(i.e., $\varphi$ total).
\end{prop}
\begin{proof}
$\overline{\rho}(\lfloor\varphi\rfloor_{s_1}^{s_2}) =
 \overline{\rho}(\neg\lceil\neg\varphi\rceil_{s_1}^{s_2}) = M_{s_2} \backslash
 \overline{\rho}(\lceil\neg\varphi\rceil_{s_1}^{s_2})$.
So $\overline{\rho}(\lfloor\varphi\rfloor_{s_1}^{s_2}) = \emptyset$ iff
$\overline{\rho}(\lceil\neg\varphi\rceil_{s_1}^{s_2}) = M_{s_2}$ iff
$\overline{\rho}(\neg\varphi) \neq \emptyset$ (by Proposition~\ref{prop:definedness})
iff $\overline{\rho}(\varphi) \neq M_{s_1}$.
Similarly,
$\overline{\rho}(\lfloor\varphi\rfloor_{s_1}^{s_2}) = M_{s_2}$ iff
$\overline{\rho}(\lceil\neg\varphi\rceil_{s_1}^{s_2}) = \emptyset$ iff
$\overline{\rho}(\neg\varphi) = \emptyset$ (by Proposition~\ref{prop:definedness})
iff $\overline{\rho}(\varphi) = M_{s_1}$.
\end{proof}

Totality is useful, for example, to define pattern equality
as the totality of the pattern equivalence relation; this is discussed
in depth shortly (Section~\ref{sec:equality}).
It is also useful when there is a need to restrict a pattern context,
say $\varphi$ of sort $s_2$, to only instances where pattern $\varphi_1$
of sort $s_1$ implies pattern $\varphi_2$ of sort $s_1$:
$\varphi \wedge \lfloor \varphi_1 \ra \varphi_2 \rfloor_{s_1}^{s_2}$.
Indeed,
$\overline{\rho}(\varphi \wedge \lfloor \varphi_1 \ra \varphi_2 \rfloor_{s_1}^{s_2})$
is  $\overline{\rho}(\varphi)$ iff
$\overline{\rho}(\varphi_1) \subseteq \overline{\rho}(\varphi_2)$,
and it is $\emptyset$ otherwise.
For example, $\exists x\,.\,x \wedge\lfloor \varphi_1 \ra \varphi_2 \rfloor_{s_1}^{s_2}$
defines the set of all values of $x$ with the property that if they match
$\varphi_1$ then they also match $\varphi_2$.
A concrete instance of this is the definition of ``magic wand''
in separation logic (Section~\ref{sec:sep-logic}).

The totality constructs satisfy, in a more general sorted setting, some of
the basic properties of modal logic operators, such as
\textbf{(N)}, \textbf{(K)}, \textbf{(M)} and \textbf{(5)}
\cite{Becker1930,kripke1959,Goldblatt:2003:MML:969657.969658}:
\begin{cor}
\label{cor:totality-NKM5}
If $s_1,s_2 \in S$ and $\varphi$, $\varphi_1$ and $\varphi_2$ are patterns of sort $s_1$, then:
\begin{itemize}[label=\rm\textbf{(M)}]
\item[\rm\textbf{(N)}] If $\models \varphi$ then $\models \lfloor \varphi \rfloor_{s_1}^{s_2}$
\item[\rm\textbf{(K)}] $\models \lfloor \varphi_1 \ra \varphi_2 \rfloor_{s_1}^{s_2}
\ra (\lfloor \varphi_1 \rfloor_{s_1}^{s_2} \ra \lfloor \varphi_2 \rfloor_{s_1}^{s_2})$
\item[\rm\textbf{(M)}] $\models \lfloor \varphi \rfloor_{s_1}^{s_1} \ra \varphi$
\item[\rm\textbf{(5)}] $\models \lceil \varphi \rceil_{s_1}^{s_2} \ra
\lfloor \lceil \varphi \rceil_{s_1}^{s_2} \rfloor_{s_2}^{s_2}$
\end{itemize}
\end{cor}
\begin{proof}
The \textbf{(N)} property is an immediate corollary of
Proposition~\ref{prop:totality}.
For the \textbf{(K)} property, let $M$ be a model and $\rho : \Var \ra M$
a valuation.
By Proposition~\ref{prop:totality} and the discussion in the paragraph
following it,
$\overline{\rho}(\lfloor \varphi_1 \ra \varphi_2 \rfloor_{s_1}^{s_2})$ is
either $\emptyset$ or $M_{s_2}$, the latter happening iff
$\overline{\rho}(\varphi_1) \subseteq \overline{\rho}(\varphi_2)$.
The first case makes our property vacuously hold.
In the second case, we have to show that
$\overline{\rho}(\lfloor \varphi_1 \rfloor_{s_1}^{s_2} \ra \lfloor \varphi_2 \rfloor_{s_1}^{s_2}) = M_{s_2}$,
that is, that
$\overline{\rho}(\lfloor \varphi_1 \rfloor_{s_1}^{s_2})
\subseteq \overline{\rho}(\lfloor \varphi_2 \rfloor_{s_1}^{s_2})$,
which follows by Proposition~\ref{prop:totality} from 
$\overline{\rho}(\varphi_1) \subseteq \overline{\rho}(\varphi_2)$.
To show \textbf{(M)}, we have to show 
$\overline{\rho}(\lfloor \varphi \rfloor_{s_1}^{s_1})
\subseteq \overline{\rho}(\varphi)$ for any $\rho : \Var \ra M$.
By Proposition~\ref{prop:totality}, we only need to consider the case where
$\overline{\rho}(\lfloor \varphi \rfloor_{s_1}^{s_1}) = M_{s_1}$; but this
can only happen when $\overline{\rho}(\varphi) = M_{s_1}$, so the property
holds.
For \textbf{(5)}, let $\rho : \Var \ra M$ be such that
$\overline{\rho}(\lceil \varphi \rceil_{s_1}^{s_2}) = M_{s_2}$
(by Proposition~\ref{prop:definedness}, the only other case is
$\overline{\rho}(\lceil \varphi \rceil_{s_1}^{s_2}) = \emptyset$, so the
property holds vacuously for that case).
Then by Proposition~\ref{prop:totality} it follows that
$\overline{\rho}(\lfloor \lceil \varphi \rceil_{s_1}^{s_2} \rfloor_{s_2}^{s_2}) = M_{s_2}$,
so
$\overline{\rho}(\lceil \varphi \rceil_{s_1}^{s_2}) \subseteq
\overline{\rho}(\lfloor \lceil \varphi \rceil_{s_1}^{s_2} \rfloor_{s_2}^{s_2})$.
\end{proof}
In Section~\ref{sec:modal-logic} we show that the modal logic S5 is
equivalent to a matching logic specification, where the definedness and totality
constructs play the role of the $\Diamond$ and $\Box$ modalities.

\begin{nota}
\label{notation:polymorphic-definedness}
Since $s_1$ and $s_2$ can usually be inferred from context,
we write $\lceil\_\rceil$ or $\lfloor\_\rfloor$ instead of
$\lceil\_\rceil_{s_1}^{s_2}$ or $\lfloor\_\rfloor_{s_1}^{s_2}$, respectively.
If the sort decorations cannot be inferred from context, then we assume
the stated property/axiom/rule holds for all such sorts.
\end{nota}

For example, the generic pattern axiom ``$\lceil x \rceil$ where $x\in\Var$''
replaces all the axioms $\lceil x\!:\!s_1 \rceil_{s_1}^{s_2}$ above for all
the definedness symbols for all the sorts $s_1$ and $s_2$.

\begin{nota}
\label{notation:predicates-bracket}
If $\varphi$ is a predicate (Definition~\ref{dfn:spec-predicates},
then we write $[\varphi]$ instead of $\lceil \varphi \rceil$
or $\lfloor \varphi \rfloor$.
This notation is justified, because if $\varphi$ is a predicate then
$\models \lceil \varphi \rceil \lra \lfloor \varphi \rfloor$.
\end{nota}

As Proposition~\ref{prop:constraint-propagation} will shortly show,
if $\varphi$ is a predicate, then by ``wrapping'' it with square brackets,
as $[\varphi]$, we can propagate it through the configuration symbols
and conjunctive constraints to wherever it is needed,
to facilitate local reasoning.

\subsection{Equality}
\label{sec:equality}

Here we show that, unlike in predicate logic or FOL, equality can be
defined in matching logic.
Before that, let us recall why
{\em equality cannot be defined in FOL}.
We only give a short intuitive explanation here;
the interested reader is referred to authoritative FOL textbooks for full
details, e.g., \cite{nla.cat-vn2062435,Harrison:2009:HPL:1540610}.
Suppose that equality were definable in FOL, that is, that there existed some
FOL specification in which a formula $\textit{Eq}(x,y)$ could only be
interpreted as equality in models.
Then we could use such a formula to state that all models have singleton
carriers: $\forall x.\forall y\,.\,\textit{Eq}(x,y)$.
However, FOL is not expressive enough to define models of fixed carrier size.
In FOL, if a specification admits a model of non-empty carrier $A$ then it
also admits a model whose carrier is $A\cup\{b\}$, where $b$ is some element
that is not already in $A$.
Indeed, pick some arbitrary element $a\in A$ and extend all the operations
and predicates in the model to behave on $b$ exactly the same as on $a$.
Since the operation and predicate interpretations cannot distinguish between
$a$ and $b$, the model of carrier $A$ and the model of carrier $A \cup\{b\}$
satisfy exactly the same formulae.
In particular, no FOL specification can admit only models of singleton carrier.
%Consequently, equality cannot be defined in FOL.
One can define equivalence and congruence relations, but not actual equality.
Since precise equality is sometimes desirable, extensions of FOL
{\em with equality} have been
proposed~\cite{nla.cat-vn2062435,Harrison:2009:HPL:1540610},
where a special binary predicate ``$=$'' is added to the logic together with
axioms like equality introduction ``$t = t$'' and elimination
``$(t_1=t_2) \mathrel\wedge \varphi[t_1/x] \ra \varphi[t_2/x]$'',
and interpreted as the equality/identity relation in models.

Let us first discuss why we cannot use $\lra$ as equality in matching logic.
Indeed, since ${M} \models \varphi_1 \lra \varphi_2$ iff 
$\overline{\rho}(\varphi_1) = \overline{\rho}(\varphi_2)$
for all $\rho:\Var\rightarrow M$, one may be tempted to use $\lra$ as
equality.
E.g., given a signature with one sort and one unary
symbol $f$, one may think that the pattern $\exists y \,.\,f(x) \lra y$
defines precisely the models where $f$ is a function
(because a function evaluates to only one value for any given argument,
and the interpretation of variable pattern $y$ has precisely one value).
Unfortunately, that is not true.
Consider model  $M$ with $M=\{1,2\}$ and 
$f_M$ the non-functional relation $f_M(1)=\{1,2\}$, $f_M(2)=\emptyset$.
Let $\rho:\Var \ra M$ be any $M$-valuation;
recall (Definition~\ref{def:rho-bar}) that $\rho$'s extension
$\overline{\rho}$ to patterns interprets ``$\exists$'' as union
and ``$\lra$'' as the complement of the symmetric difference.
If $\rho(x)=1$ then
$\overline{\rho}(\exists y \,.\,f(x) \lra y) =
(M \backslash (\{1,2\} \Delta \{1\})) \cup
(M \backslash (\{1,2\} \Delta \{2\})) = \{1,2\} = M
$.
If $\rho(x)=2$ then
$\overline{\rho}(\exists y \,.\,f(x) \lra y) =
(M \backslash (\emptyset\Delta \{1\})) \cup
(M \backslash (\emptyset\Delta \{2\})) = \{1,2\} = M
$.
Hence, ${M} \models \exists y \,.\,f(x) \lra y$, yet $f_M$ is not a function,
so $\lra$ fails to capture the pattern equality.

The problem above is that the interpretation of $\varphi_1 \lra \varphi_2$,
depicted in Figure~\ref{fig:diagram}, is not two-valued ($\top$ or $\bot$),
as we are used to think in classical logics.
Specifically, $\overline{\rho}(\varphi_1) \neq \overline{\rho}(\varphi_2)$
does not suffice for $\overline{\rho}(\varphi_1 \lra \varphi_2) = \emptyset$
to hold.
Indeed, $\overline{\rho}(\varphi_1 \lra \varphi_2) = M \ \backslash
 \ (\overline{\rho}(\varphi_1) \ \Delta \ \overline{\rho}(\varphi_2))$
and there is nothing to prevent, e.g.,
$\overline{\rho}(\varphi_1) \cap \overline{\rho}(\varphi_2) \neq \emptyset$,
in which case $\overline{\rho}(\varphi_1) \ \Delta \ \overline{\rho}(\varphi_2) \neq M$.
What we would like is a proper equality over patterns, $\varphi_1 = \varphi_2$,
which behaves as a two-valued predicate:
$\overline{\rho}(\varphi_1 = \varphi_2) = \emptyset$
when $\overline{\rho}(\varphi_1) \neq \overline{\rho}(\varphi_2)$, and
$\overline{\rho}(\varphi_1 = \varphi_2) = M$
when $\overline{\rho}(\varphi_1) = \overline{\rho}(\varphi_2)$.
Moreover, we want equalities to be used with patterns of any sort $s_1$
and in contexts of any sort $s_2$, similarly to the definedness and totality
constructs in Section~\ref{sec:definedness}.
% (for example, we may want our $f$ in the
%example above to have a result sort different from that of its argument).

Equality can be defined quite compactly using the pattern totality and
equivalence constructs, which were themselves defined using the assumed
definedness symbols (Assumption~\ref{assumption:definedness},
Section~\ref{sec:definedness}) and, respectively, the
core $\wedge$ and $\neg$ constructs (Section~\ref{sec:matching-logic}).
Specifically,

\begin{nota}
\label{notation:equality}
For each pair of sorts $s_1$ (for the compared patterns) and
$s_2$ (for the context in which the equality is used), we define
$\_=_{s_1}^{s_2}\_$ as the following derived construct:
$$
\begin{array}{@{}rcll}
\varphi =_{s_1}^{s_2} \varphi' & \ \ \ \ \ \equiv \ \ \ \ \ &
\lfloor\varphi \lra \varphi'\rfloor_{s_1}^{s_2}
& \ \ \ \ \ \mbox{where } \varphi,\varphi' \in \Pattern_{s_1}
\end{array}
$$
\end{nota}

Intuitively, $\varphi \lra \varphi'$ matches the grey area in
the diagram depicting pattern equivalence in Figure~\ref{fig:diagram}
(complement of the symmetric difference),
so $\lfloor\varphi \lra \varphi'\rfloor_{s_1}^{s_2}$ is interpreted as
$M_{s_2}$ iff the white area is empty, iff the two patterns match exactly
the same elements.
Formally,

\begin{prop}
\label{prop:equality}
Let $\varphi,\varphi' \in \Pattern_{s_1}$.  Then:
\begin{enumerate}
\item $\overline{\rho}(\varphi =_{s_1}^{s_2} \varphi') = \emptyset$
iff $\overline{\rho}(\varphi) \neq \overline{\rho}(\varphi')$, for any $\rho : \Var \ra M$
\item $\overline{\rho}(\varphi =_{s_1}^{s_2} \varphi') = M_{s_2}$
iff $\overline{\rho}(\varphi) = \overline{\rho}(\varphi')$, for any $\rho : \Var \ra M$
\item ${M} \models \varphi =_{s_1}^{s_2} \varphi'$ iff
${M} \models \varphi \lra \varphi'$, for any model $M$
\item $\models \varphi =_{s_1}^{s_2} \varphi'$ iff $\models \varphi \lra \varphi'$
\end{enumerate}
\end{prop}
\begin{proof}
Recall that 
$\varphi =_{s_1}^{s_2} \varphi'$ stands for
$\lfloor\varphi \lra \varphi'\rfloor_{s_1}^{s_2}$, which stands for
$\neg\lceil\neg(\varphi\lra\varphi')\rceil_{s_1}^{s_2}$.
\begin{enumerate}
\item
Therefore, $\overline{\rho}(\varphi =_{s_1}^{s_2} \varphi')$ is equal to
$M_{s_2}\ \backslash\ (\lceil\_\rceil_{s_1}^{s_2})_M(M_{s_1}\ \backslash\ 
(M_{s_1} \ \backslash\ (\overline{\rho}(\varphi_1) \ \Delta\ \overline{\rho}(\varphi_2))))$,
which is further equal to
$M_{s_2}\ \backslash\ (\lceil\_\rceil_{s_1}^{s_2})_M(
\overline{\rho}(\varphi_1) \ \Delta\ \overline{\rho}(\varphi_2))$.
So $\overline{\rho}(\varphi =_{s_1}^{s_2} \varphi') = \emptyset$
iff $(\lceil\_\rceil_{s_1}^{s_2})_M(
\overline{\rho}(\varphi_1) \ \Delta\ \overline{\rho}(\varphi_2))=M_{s_2}$,
iff $\overline{\rho}(\varphi_1) \ \Delta\ \overline{\rho}(\varphi_2)\neq\emptyset$,
iff $\overline{\rho}(\varphi) \neq \overline{\rho}(\varphi')$.

\item
Similarly to the above, we have $\overline{\rho}(\varphi =_{s_1}^{s_2} \varphi') = M_{s_2}$
iff $(\lceil\_\rceil_{s_1}^{s_2})_M(
\overline{\rho}(\varphi_1) \ \Delta\ \overline{\rho}(\varphi_2))=\emptyset$,
iff $\overline{\rho}(\varphi_1) \ \Delta\ \overline{\rho}(\varphi_2)=\emptyset$,
iff $\overline{\rho}(\varphi) = \overline{\rho}(\varphi')$.

\item
$M \models \varphi =_{s_1}^{s_2} \varphi'$
iff $\overline{\rho}(\varphi =_{s_1}^{s_2} \varphi')=M_{s_2}$ for any
$\rho:\Var\ra M$,
iff $\overline{\rho}(\varphi) = \overline{\rho}(\varphi')$
for any $\rho:\Var\ra M$,
iff (by Proposition~\ref{prop:simple})
$M \models \varphi \lra \varphi'$.

\item
$\models \varphi =_{s_1}^{s_2} \varphi'$ iff
$M \models \varphi =_{s_1}^{s_2} \varphi'$ for any model $M$,
iff (by the above)
$M \models \varphi \lra \varphi'$ for any model $M$,
iff $\models \varphi \lra \varphi'$.
\end{enumerate}
Therefore, pattern equality satisfies all these properties.
\end{proof}

Note that (4) in the proposition above is not in conflict with the
discussion at the beginning of this section concluding that we cannot
use equivalence instead of equality.
The example there illustrated an equivalence which was nested under
a quantifier ($\exists y \,.\,f(x) \lra y$), while (4) above says that
equivalence and equality are interchangeable at the pattern top.

Like for definedness and totality (Section~\ref{sec:definedness}), where we
decided to drop the sorts $s_1$ and $s_2$ from $\lceil\_\rceil_{s_1}^{s_2}$
and instead write $\lceil\_\rceil$ because the sort of the enclosed pattern
and that of the context dictate $s_1$ and $s_2$, we also take the freedom
to drop the sort embellishments of $=_{s_1}^{s_2}$ and instead
write just $=$.
Like for definedness and totality, $s_1$ and $s_2$ can typically be inferred
from context, and, if ambiguity arises, then we assume all instances.
For example, ``$x = x$'' means
``$x =_{s_1}^{s_2} x$'' for any $s_1,s_2\in S$ and $x\in\Var_{s_1}$.
Note that the equality symbol in algebraic specifications and in FOL
(with equality) is also implicitly indexed by the sort of the two terms,
although that sort is typically not mentioned as subscript; but one needs
to exercise more care in matching logic, because equality patterns can be
nested now.
For example, the pattern in Section~\ref{sec:map-patterns} defining
linked list data-structures within maps,
$$
\llist(x) = (x = 0 \wedge \SLemp \vee \exists z \,.\,x \mapsto z \SLstar \llist(z))
$$
is a sugared variant of the explicit patterns (one for each
``equality context'' sort $s$),
$$
\llist(x) =_{\Map}^{s} (x =_{\Nat}^{\Map} 0 \wedge \SLemp \vee \exists z \,.\,x \mapsto z \SLstar \llist(z))
$$
To minimize the number of disambiguation parentheses, we assumed that
equality ($=$) binds tighter than conjunction ($\wedge$).
We also assume that negation ($\neg$) binds tighter than equality ($=$).
To avoid confusion, we may use disambiguation parentheses even if not
strictly needed.

Despite the fact that patterns evaluate to any set of values
and thus are more general than both terms (which evaluate to only one
value) and predicates (which evaluate to one of two values), and despite
the fact that Boolean combinations of patterns and quantification yield other
patterns which can be used under any symbol in $\Sigma$, as we saw in
Proposition~\ref{prop:PL-validity}, the proof rule/axiom
schemas of (pure) predicate logic continue to be sound for matching logic.
Now that we have equality, a natural question is whether the equality
proof rule/axiom schemas of FOL with
equality~\cite{nla.cat-vn2062435,Harrison:2009:HPL:1540610} are also sound.
For example, in FOL with equality, ``equality elimination'' states
that terms can be substituted with equal terms in any context.
A similar result holds for matching logic, where terms are replaced with
arbitrary patterns:\newpage
\begin{prop}
\label{prop:FOL-equality}
The following hold:
\begin{enumerate}
\item Equality introduction: $\models \varphi = \varphi$
\item Equality elimination: $\models (\varphi_1=\varphi_2) \mathrel\wedge \varphi[\varphi_1/x] \ra \varphi[\varphi_2/x]$
\end{enumerate}
\end{prop}
\begin{proof}
(1) follows by (4) in Proposition~\ref{prop:equality} and by Proposition~\ref{prop:propositional}.
For (2), let $M$ be some model and $\rho:\Var\ra M$.
By Proposition~\ref{prop:simple}, it suffices to show
$\overline{\rho}(\varphi_1=\varphi_2) \cap
\overline{\rho}(\varphi[\varphi_1/x]) \subseteq
\overline{\rho}(\varphi[\varphi_2/x])$.
If $\overline{\rho}(\varphi_1) \neq \overline{\rho}(\varphi_2)$ then
$\overline{\rho}(\varphi_1=\varphi_2) =\emptyset$ by
Proposition~\ref{prop:equality}, so the inclusion holds.
Now suppose that 
$\overline{\rho}(\varphi_1) = \overline{\rho}(\varphi_2)$, which
implies $\overline{\rho}(\varphi_1=\varphi_2) = M$ by
Proposition~\ref{prop:equality}, so it suffices to show
$\overline{\rho}(\varphi[\varphi_1/x]) \subseteq
\overline{\rho}(\varphi[\varphi_2/x])$.
The stronger result 
$\overline{\rho}(\varphi[\varphi_1/x]) = 
\overline{\rho}(\varphi[\varphi_2/x])$ in fact holds, because
the first element is a function of $\overline{\rho}(\varphi_1)$,
the second element is the same function but of $\overline{\rho}(\varphi_2)$,
and $\overline{\rho}(\varphi_1)=\overline{\rho}(\varphi_2)$.
\end{proof}

\begin{nota}
\label{notation:negation}
From here on in the rest of the paper we write
$\varphi \neq \varphi'$ instead of $\neg(\varphi = \varphi')$.
\end{nota}

One may wonder what really made it possible to define equality in matching
logic, which is not possible in predicate or first-order logic.
Let us consider the simplest instance of equality, $x=y$ between two
variables, which is sugar for $\neg\lceil \neg(x \lra y) \rceil$.
After all, definedness-like predicates can also be defined in predicate
logic; following the translation in Section~\ref{sec:PL-reduction}, for
example, the unary matching logic symbols $\lceil\_\rceil$ are
associated binary predicates $\pi_{\lceil\_\rceil}$, and the
definedness pattern axioms $\lceil x \rceil$ are translated into
formula axioms $\pi_{\lceil\_\rceil}(x,r)$.
So the definedness symbol is not the key.
The key is the capability to allow logical connectives between ``terms'',
which is not allowed in first-order logic.
For example, $x \lra y$ already tells us whether $x$ and $y$ are
interpreted as the same value or not: for any valuation $\rho$,
it is indeed the case that $\overline{\rho}(x \lra y)$ is the
total set iff $\rho(x)=\rho(y)$ (see Proposition~\ref{prop:simple}).

Equality elimination (Proposition~\ref{prop:FOL-equality}) allows us to
replace patterns by equal patterns in any context.
Further, Proposition~\ref{prop:equality} allows us to replace any
top-level $\lra$ with $=$.
In particular, the equivalences in
Proposition~\ref{prop:symbol-distributivity} become
$\models C_{\sigma,i}[\varphi_i \orx \varphi'_i] =
C_{\sigma,i}[\varphi_i] \orx C_{\sigma,i}[\varphi'_i]$
and 
$\models C_{\sigma,i}[\exists x\,.\,\varphi_i] =
\exists x\,.\,C_{\sigma,i}[\varphi_i]$, respectively,
meaning that we can propagate disjunction and existential quantification
through symbols in any context, not only at the top level.
Because of the stronger nature of equality, from here on we state
properties in terms of equality instead of $\lra$ whenever possible.
Below is an important such property:

\begin{prop}
\label{prop:constraint-propagation}
\textbf{(Constraint propagation)}
Assume the same hypothesis as in
Proposition~\ref{prop:symbol-distributivity}:
$\sigma\in\Sigma_{s_1\ldots s_n,s}$ and
$\varphi_i\in \Pattern_{s_i}$ for all $1\leq i \leq n$,
a particular $1\leq i \leq n$, and let
$C_{\sigma,i}[\square]$ be the context
$\sigma(\varphi_1,\ldots,\varphi_{i-1},\square,\varphi_{i+1},\ldots\varphi_n)$.
Then for any pattern $\varphi$:
\begin{enumerate}
\item
$
\models C_{\sigma,i}[\varphi_i \andx \lceil\varphi\rceil]
 = C_{\sigma,i}[\varphi_i] \andx \lceil\varphi\rceil
$
\item
$
\models C_{\sigma,i}[\varphi_i \andx \lfloor\varphi\rfloor]
 = C_{\sigma,i}[\varphi_i] \andx \lfloor\varphi\rfloor
$
\item
$
\models C_{\sigma,i}[\varphi_i \andx [\varphi]]
 = C_{\sigma,i}[\varphi_i] \andx [\varphi]
$ if $\varphi$ is a predicate
(Definition~\ref{dfn:spec-predicates} and
Notation~\ref{notation:predicates-bracket}).
\end{enumerate}
\end{prop}
\begin{proof}
We only show (1), because (2) and (3) are similar.
Let $\rho:\Var \ra M$ and let
$\overline{\rho}(C_{\sigma,i}):M_{s_i} \ra {\cal P}(M_s)$ be defined as
$\overline{\rho}(C_{\sigma,i})(a)
= \sigma_M(\overline{\rho}(\varphi_1),\ldots,\overline{\rho}(\varphi_{i-1}),a,
\overline{\rho}(\varphi_{i+1}),\ldots,\overline{\rho}(\varphi_n))$.
Then $\overline{\rho}(C_{\sigma,i}[\varphi_i \andx \lceil\varphi\rceil])
= \overline{\rho}(C_{\sigma,i})(\overline{\rho}(\varphi_i) \cap \overline{\rho}(\lceil\varphi\rceil))$ and
$\overline{\rho}(C_{\sigma,i}[\varphi_i] \andx \lceil\varphi\rceil) =
\overline{\rho}(C_{\sigma,i})(\overline{\rho}(\varphi_i)) \cap \overline{\rho}(\lceil\varphi\rceil])$.
By Proposition~\ref{prop:definedness}, $\overline{\rho}(\lceil\varphi\rceil])$ is either
the empty set or the total set, regardless of the result sort context ($s_i$ vs. $s$).
If the empty set, then
$\overline{\rho}(C_{\sigma,i}[\varphi_i \andx \lceil\varphi\rceil]) =
\overline{\rho}(C_{\sigma,i})(\overline{\rho}(\varphi_i) \cap \emptyset) = \emptyset$
and
$\overline{\rho}(C_{\sigma,i}[\varphi_i] \andx \lceil\varphi\rceil) =
\overline{\rho}(C_{\sigma,i})(\overline{\rho}(\varphi_i)) \cap \emptyset = \emptyset$.
If the total set, then
$\overline{\rho}(C_{\sigma,i}[\varphi_i \andx \lceil\varphi\rceil]) =
\overline{\rho}(C_{\sigma,i})(\overline{\rho}(\varphi_i) \cap M_{s_i}) = 
\overline{\rho}(C_{\sigma,i})(\overline{\rho}(\varphi_i))$
and
$\overline{\rho}(C_{\sigma,i}[\varphi_i] \andx \lceil\varphi\rceil) =
\overline{\rho}(C_{\sigma,i})(\overline{\rho}(\varphi_i)) \cap M_s = 
\overline{\rho}(C_{\sigma,i})(\overline{\rho}(\varphi_i))$.
Therefore,
$\overline{\rho}(C_{\sigma,i}[\varphi_i \andx \lceil\varphi\rceil]) =
\overline{\rho}(C_{\sigma,i}[\varphi_i] \andx \lceil\varphi\rceil)$.
\end{proof}

Constraint propagation allows us to propagate, through symbols, any logical
constraints that appear in a conjunctive context.
Indeed, as seen in the rest of this section
(in particular in Section~\ref{sec:builtins}) and in Section~\ref{sec:FOL},
domain constraints can be expressed as equalities or as FOL predicates,
and both of these are instances of matching logic predicates.
Recall from Definition~\ref{dfn:spec-predicates} that (matching logic)
predicates are patterns which interpret to either the empty or the total
set of their carrier.
The definedness symbol applied to a predicate, the square brackets in
$[\varphi]$ (Notation~\ref{notation:predicates-bracket}), does not change
the semantics of the predicate, but thanks to its polymorphic
nature (Notation~\ref{notation:polymorphic-definedness}) we can
syntactically move $\varphi$ from the sort context of the argument pattern
($s_i$) of $\sigma$ to the sort context of $\sigma$'s result ($s$).

Proposition~\ref{prop:constraint-propagation} (constraint propagation)
and Proposition~\ref{prop:structural-framing} (structural framing) are
particularly useful to localize proof efforts, as illustrated in the example
in Section~\ref{sec:example}.

\subsection{Membership}
\label{sec:membership}

Since in matching logic a pattern $\varphi$ evaluates to a set of values
while a variable (pattern) $x$ evaluates to just a (set containing only one)
value, the membership question, ``does $x \in \varphi$ hold?'', is natural.
As seen later in Section~\ref{sec:deduction}, membership in fact plays a
key role in proving the completeness of matching logic reasoning.
Fortunately, membership can be quite easily defined as
a derived construct in matching logic, making use of the definedness
symbol (Section~\ref{sec:definedness}), in a similar way to equality
(Section~\ref{sec:equality}):
\begin{nota}
\label{notation:membership}
If $x\in\Var_{s_1}$, $\varphi\in\Pattern_{s_1}$ and $s_2\in S$, then
we introduce the notation
$$
\begin{array}{@{}rcll}
x \in_{s_1}^{s_2} \varphi & \ \ \ \ \ \equiv \ \ \ \ \ &
\lceil x \wedge \varphi\rceil_{s_1}^{s_2}
& % \ \ \ \ \ \mbox{where } x \in \Var_{s_1},\varphi\in \Pattern_{s_1}\\
\end{array}
$$
Like for definedness, totality and equality, there is a membership construct
for each pair of sorts $s_1$ (for variable and pattern) and $s_2$
(for context); we take the freedom to omit them.
\end{nota}

\begin{prop}
\label{prop:membership}
With the above, the following hold:
\begin{enumerate}
\item $\overline{\rho}(x \in_{s_1}^{s_2} \varphi) = \emptyset$
iff $\rho(x) \not\in \overline{\rho}(\varphi)$, for any $\rho:\Var\ra M$
\item $\overline{\rho}(x \in_{s_1}^{s_2} \varphi) = M_{s_2}$
iff $\rho(x) \in \overline{\rho}(\varphi)$, for any $\rho:\Var\ra M$
\item $\models (x\in_{s_1}^{s_2} \varphi) =_{s_2}^{s_3} (x \wedge \varphi =_{s_1}^{s_2} x)$,
for any sort $s_3$
\end{enumerate}
\end{prop}
\begin{proof}
Recall that $x \in_{s_1}^{s_2} \varphi$ is $[x \wedge \varphi]_{s_1}^{s_2}$.
\begin{enumerate}
\item
Therefore, we have
$\overline{\rho}(x \in_{s_1}^{s_2} \varphi)=
(\lceil\_\rceil_{s_1}^{s_2})_M(\{\rho(x)\}\cap\overline{\rho}(\varphi))$,
so $\overline{\rho}(x \in_{s_1}^{s_2} \varphi) = \emptyset$
iff $\{\rho(x)\}\cap\overline{\rho}(\varphi) = \emptyset$, that is,
iff $\rho(x) \not\in \overline{\rho}(\varphi)$.

\item
Similarly to above,  $\overline{\rho}(x \in_{s_1}^{s_2} \varphi) = M_{s_2}$
iff $\{\rho(x)\}\cap\overline{\rho}(\varphi) \neq \emptyset$, that is,
iff $\rho(x) \in \overline{\rho}(\varphi)$.

\item
Let $M$ be some model and $\rho:\Var\ra M$.
By Proposition~\ref{prop:equality}, the property holds iff we can show
$\overline{\rho}(x\in_{s_1}^{s_2} \varphi) =
\overline{\rho}(x \wedge \varphi =_{s_1}^{s_2} x)$.
Since the membership and equality patterns are predicates
(Definition~\ref{dfn:predicates}), and thus they evaluate either to the
entire set or to the empty set, the following completes the proof:
by (2) we have $\overline{\rho}(x\in_{s_1}^{s_2} \varphi) = M_{s_2}$ iff
$\rho(x) \in \overline{\rho}(\varphi)$, iff
$\{\rho(x)\} \cap \overline{\rho}(\varphi) = \{\rho(x)\}$, iff,
by (2) in Proposition~\ref{prop:equality},
$\overline{\rho}(x \wedge \varphi =_{s_1}^{s_2} x) = M_{s_2}$;
and
by (1) we have $\overline{\rho}(x\in_{s_1}^{s_2} \varphi) = \emptyset$ iff
$\rho(x) \not\in \overline{\rho}(\varphi)$, iff
$\{\rho(x)\} \cap \overline{\rho}(\varphi) \not= \{\rho(x)\}$, iff,
by (1) in Proposition~\ref{prop:equality},
$\overline{\rho}(x \wedge \varphi =_{s_1}^{s_2} x) = \emptyset$;

\end{enumerate}
Therefore, these basic properties hold.
\end{proof}
Property (3) in Proposition~\ref{prop:membership} suggests that
the equality $x \wedge \varphi = x$ can be regarded as an alternative
definition of membership $x \in \varphi$, but we prefer
$\lceil x \wedge \varphi\rceil$ because is simpler (the other one
requires an additional sort, $s_3$, for the context of the equality).

Proposition~\ref{prop:PL-validity} showed that some of the proof rule/axiom
schemas of FOL with equality are already sound for matching logic, namely
the rules corresponding to (pure) predicate logic.
Proposition~\ref{prop:FOL-equality} further showed that the equality-related
rules/axioms are also sound.
The soundness of several other rule/axiom schemas are shown below,
essentially completing the soundness of the matching logic proof system
(discussed later in Section~\ref{sec:deduction}), except for one rule,
Substitution, which needs more discussion and we postpone it to
Section~\ref{sec:deduction}:

\begin{prop}
\label{prop:FOL-membership}
The following hold:
\begin{enumerate}
%\item $\models \varphi = \exists x\,.\, x\wedge(x\in\varphi)$
\item $\models \forall x\,.\,x\in\varphi$ iff $\models \varphi$
\item $\models (x \in y) = (x=y)$ when $x,y\in\Var$
\item $\models (x\in\neg\varphi) = \neg(x\in\varphi)$
\item $\models (x\in\varphi_1\wedge\varphi_2) = (x\in\varphi_1) \wedge (x\in\varphi_2)$
\item $\models (x\in\exists y . \varphi) = \exists y.(x\in\varphi)$, with $x$ and $y$ distinct
\item $\models x\in\sigma(\varphi_1,...,\varphi_{i-1},\varphi_i,\varphi_{i+1},...,\varphi_n)
= \exists y . (y\in\varphi_i \mathrel\wedge x \in\sigma(\varphi_1,...,\varphi_{i-1},y,\varphi_{i+1},...,\varphi_n))$
%\item $\models \varphi_1 = \varphi_2 \wedge \varphi[\varphi_1/x] \ra \varphi[\varphi_2/x]$
%\item $\models \varphi_1 = \varphi_2$ and $\models\varphi[\varphi_1/x]$ imply $\models \varphi[\varphi_2/x]$
%\item $\models \varphi$ iff $\models \forall x\,.\, (x \in \varphi)$
%\item $\models x \in y = (x=y)$ where $x,y\in\Var$ have the same sort
%\item $\models x\in\neg\varphi = \neg(x\in\varphi)$
%\item $\models x \in \varphi_1 \wedge \varphi_2 = (x \in \varphi_1) \wedge (x \in \varphi_2)$
%\item $x\in\exists y\,.\,\varphi = \exists y\,.\, x \in \varphi$
%\item $\models x \in \sigma(\varphi_1,\ldots,\varphi_{i-1},\varphi_i,\varphi_{i+1},\ldots,\varphi_n) \\
%\hspace*{2ex} = \exists x_i\,.\,(x_i \in \varphi_i \wedge x\in\sigma(\varphi_1,\ldots,\varphi_{i-1},x_i
%\varphi_{i+1},\ldots,\varphi_n))$
\end{enumerate}
\end{prop}
\begin{proof}
The proofs below make repetitive use of
Propositions~\ref{prop:equality} and ~\ref{prop:membership}:
\begin{enumerate}
\item
Let $M$ be a model.
Then $M \models \forall x\,.\,x\in\varphi$ iff $M \models x\in\varphi$
(Proposition~\ref{prop:simple}), iff
$\overline{\rho}(x\in\varphi) = M$ for any $\rho:\Var\ra M$,
iff $\rho(x)\in\overline{\rho}(\varphi)$ for any $\rho:\Var\ra M$,
iff $\overline{\rho}(\varphi)=M$ for any $\rho:\Var\ra M$,
iff $M\models \varphi$.

\item
It suffices to show $\overline{\rho}(x\in y) = M$ iff
$\overline{\rho}(x = y) = M$ for any model $M$ and any $\rho:\Var\ra M$,
that is, that $\rho(x) \in \{{\rho}(y)\}$ iff 
$\rho(x) = {\rho}(y)$, which obviously holds.

\item
It suffices to show
$\overline{\rho}(x\in \neg\varphi) = M$ iff
$\overline{\rho}(x\in \varphi) = \emptyset$ for any model $M$ and any
$\rho:\Var\ra M$,
that is, that
$\rho(x) \in M\backslash\overline{\rho}(\varphi)$ iff
$\rho(x) \not\in \overline{\rho}(\varphi)$, which obviously holds.

\item
It suffices to show
$\rho(x)\in\overline{\rho}(\varphi_1)\cap\overline{\rho}(\varphi_2)$ iff
$\rho(x)\in\overline{\rho}(\varphi_1)$ and $\rho(x)\in\overline{\rho}(\varphi_2)$
for any model $M$ and any $\rho:\Var\ra M$, which obviously holds.

\item
It suffices to show
for any model $M$ and any $\rho:\Var\ra M$, that
$\rho(x)\in\bigcup\{\overline{\rho'}(\varphi) \mid \rho':\Var\ra M,\ 
\rho'\!\!\upharpoonright_{\Var\backslash\{y\}} =
\rho\!\!\upharpoonright_{\Var\backslash\{y\}}\}$ iff
$\bigcup\{\overline{\rho'}(x\in\varphi) \mid \rho':\Var\ra M,\ 
\rho'\!\!\upharpoonright_{\Var\backslash\{y\}} =
\rho\!\!\upharpoonright_{\Var\backslash\{y\}}\}=M$.
It is easy to see that each of the two statements holds iff
there exists some $\rho':\Var\ra M$ with
$\rho'\!\!\upharpoonright_{\Var\backslash\{y\}} =
\rho\!\!\upharpoonright_{\Var\backslash\{y\}}$ such that
$\rho(x)\in\overline{\rho'}(\varphi)$.

\item
It suffices to prove for any model $M$ and
any valuation $\rho:\Var\ra M$, that
$$\rho(x)\in\sigma_M(\overline{\rho}(\varphi_1),\ldots,\overline{\rho}(\varphi_{i-1}),\overline{\rho}(\varphi_i),\overline{\rho}(\varphi_{i+1}),\ldots,\overline{\rho}(\varphi_n))$$
iff there exists a $\rho':\Var\ra M$ with
$\rho'\!\!\upharpoonright_{\Var\backslash\{y\}} =
\rho\!\!\upharpoonright_{\Var\backslash\{y\}}$ such that
$\rho'(y)\in\overline{\rho}(\varphi_i)$ and
$$\rho(x)\in\sigma_M(\overline{\rho}(\varphi_1),\ldots,
\overline{\rho}(\varphi_{i-1}),\{\rho'(y)\},\overline{\rho}(\varphi_{i+1}),\ldots,
\overline{\rho}(\varphi_n)),$$
which obviously holds.
\end{enumerate}
The proof is complete.
\end{proof}

We next define several common relations using patterns, such as functions.

\subsection{Functions}
\label{sec:functions}

Matching logic makes no distinction between function and predicate
symbols, treating all symbols uniformly as pattern symbols which are
interpreted relationally.
A natural question is whether there is any way, in matching logic, to
state that a symbol is to be interpreted as a function in all models.
We show a more general result, namely that there is a way to state
that any pattern, not only a symbol, has a functional interpretation.
%Assume some background matching logic specification $(S,\Sigma,F)$.

\begin{defi}
\label{dfn:functional}
Pattern $\varphi$ is \textbf{functional in a model $M$} iff 
$|\overline{\rho}(\varphi)|=1$ for any valuation $\rho:\Var \ra M$.
Furthermore, $\varphi$ is \textbf{functional in} $F\subseteq\Pattern$,
or simply \textbf{functional} when $F$ is understood, 
iff it is functional in all models $M$ with $M \models F$.
\end{defi}

Recall from the preamble of Section~\ref{sec:useful-symbols} that
$(S,\Sigma,F)$ was assumed to be an arbitrary but fixed matching logic
specification.
Therefore $F$ is understood, so we take the freedom to just say
``$\varphi$ is functional'' instead of ``$\varphi$ is functional in $F$''.

The following trivial result relates functional patterns to (total) functions:

\begin{prop}
\label{prop:functions}
If $\sigma\in\Sigma_{s_1\ldots s_n,s}$ and $M$ is a $\Sigma$-model,
then pattern $\sigma(x_1,\ldots,x_n)$ is functional in $M$ iff
$\sigma_M:M_{s_1}\times\cdots\times M_{s_n} \ra M_s$ is a total function
in $M$, that is, iff $\sigma_M(a_1,...,a_n)$ contains
precisely one element for any elements $a_1\in M_{s_1}$, ..., $a_n\in M_{s_n}$.
\end{prop}
\begin{proof}
Pattern $\sigma(x_1,\ldots,x_n)$ is functional in $M$ iff
$|\overline{\rho}(\sigma(x_1,\ldots,x_n))|=1$ for any valuation $\rho:\Var \ra M$
(by Definition~\ref{dfn:functional}),
iff $|\sigma_M(a_1,...,a_n)|=1$ for any
$a_1\in M_{s_1}$, ..., $a_n\in M_{s_n}$.
\end{proof}

The following proposition gives an axiomatic characterization of
functional patterns:

\begin{prop}
\label{prop:functional}
Pattern $\varphi$ is functional in model $M$ iff $M \models \exists y\,.\,(\varphi = y)$,
where variable $y$ is chosen so that
$y\not\in\FV(\varphi)$.
Therefore, $\varphi$ is functional iff $F \models \exists y\,.\,(\varphi = y)$.
\end{prop}
\begin{proof}
$\varphi$ is functional in $M$ iff $|\overline{\rho}(\varphi)|=1$
for any $\rho:\Var \ra M$ (by Definition~\ref{dfn:functional}), iff
for any $\rho:\Var \ra M$ there is some $a\in M$ such that
$\overline{\rho}(\varphi)=\{a\}$, iff
for any $\rho:\Var \ra M$ there is some $\rho':\Var\ra M$
with
$\rho'\!\!\upharpoonright_{\Var\backslash\{y\}} =
\rho\!\!\upharpoonright_{\Var\backslash\{y\}}$
such that
$\overline{\rho'}(\varphi)=\overline{\rho'}(y)$
(by (1) in Proposition~\ref{prop:simple}), iff
$M \models \exists y\,.\,(\varphi = y)$ (by Definition~\ref{def:rho-bar} and Proposition~\ref{prop:equality}).
\end{proof}

\begin{cor}
\label{cor:variables}
Variables are functional: $\models \exists y\,.\,x=y$ for any variable $x$.
\end{cor}
\begin{proof}
Immediate consequence of Definition~\ref{dfn:functional} and
Proposition~\ref{prop:functional}, because variables are interpreted as
singletons: $\overline{\rho}(x)=\{\rho(x)\}$ for any valuation
$\rho:\Var \ra M$.
\end{proof}

We have seen in the discussion at the beginning of
Section~\ref{sec:equality} that if $f$ is a one-argument symbol, the
pattern $\exists y \,.\,f(x) \lra y$ is not strong enough to enforce $f(x)$
to be functional.
However, thanks to Proposition~\ref{prop:functional}, using equality instead
of equivalence works:

\begin{cor}
\label{cor:functions}
If $\sigma\in\Sigma_{s_1\ldots s_n,s}$ and $M$ is a $\Sigma$-model, then
$\sigma_M$ is a total function iff
$M\models\exists y\,.\,\sigma(x_1,\ldots,x_n) = y$.
\end{cor}
\begin{proof}
By Propositions~\ref{prop:functions} and \ref{prop:functional}.
\end{proof}

Hence, we can state that a symbol $\sigma\in\Sigma_{s_1\ldots s_n,s}$
is a function in all models by requiring
$\sigma(x_1,\ldots,x_n)$ to be a functional pattern, which by
Proposition~\ref{prop:functional} is equivalent to stating that
the pattern $\exists y\,.\,\sigma(x_1,\ldots,x_n) = y$ holds
(i.e., it is entailed by $F$), where $x_1$, ..., $x_n$ are free
variables.
The simplest way to ensure this is to add this pattern directly
to $F$, as an axiom.
To avoid manually writing such trivial pattern axioms for lots of symbols
which are meant to be interpreted as functions, we adopt
the following notation and terminology:

\begin{defi}
\label{dfn:functional-notation}
For a symbol $\sigma\in\Sigma_{s_1\ldots s_n,s}$, the notation
$$\sigma : s_1 \times \cdots \times s_n \ra s$$
is syntactic sugar for stating that $F$ contains the pattern
$
\exists y\,.\,\sigma(x_1,\ldots,x_n) = y
$.
If $\sigma\in\Sigma_{s_1\ldots s_n,s}$ is a symbol such that
$\sigma : s_1 \times \cdots \times s_n \ra s$, then we call
$\sigma$ a \textbf{function symbol}.
Patterns built with only function symbols are called
\textbf{term patterns}, or simply just \textbf{terms}.
\end{defi}

Definition~\ref{dfn:functional-notation} is instrumental to capturing
algebraic specifications and first-order logic as instances of
matching logic; full details are given in Sections~\ref{sec:alg-spec}
and \ref{sec:FOL}.

\begin{cor}
\label{cor:terms}
Term patterns are functional: $\models \exists y\,.\,t=y$ for any term
pattern $t$.
\end{cor}
\begin{proof}
Structural induction on term pattern $t$.
Obvious when $t$ is a variable.
Let $t$ be $\sigma(t_1,\ldots,t_n)$ with $\sigma : s_1 \times \ldots \times s_n \ra s$ and
$t_1,\ldots,t_n$ term patterns of sorts $s_1$, \ldots, $s_n$, respectively, and
let $\rho : \Var \ra M$.
Then $\overline{\rho}(t) = \sigma_M(\overline{\rho}(t_1),\ldots,\overline{\rho}(t_n))$.
By the induction hypothesis, $t_1,\ldots,t_n$ are functional, that is,
$\overline{\rho}(t_1)$, \ldots, $\overline{\rho}(t_n)$ are singleton sets
(Definition~\ref{dfn:functional}), say
$\{m_1\}$, \ldots, $\{m_n\}$, respectively.
Then $\overline{\rho}(t) = \sigma_M(m_1,\ldots,m_n)$, which is also a singleton
set, say $\{m\}$, as $\sigma_M$ is a function
(Corollary~\ref{cor:functions}).
The rest follows by Proposition~\ref{prop:functional}.
\end{proof}

In FOL, the Substitution axiom ($(\forall x\!:\!s\,.\,\varphi) \ra \varphi[t/x]$)
allows for universally quantified variables to be substituted with any terms.
Together with the proof rules and axioms of predicate logic,
Substitution makes FOL deduction complete.
An important property of term patterns in matching logic is that the
Substitution axiom of FOL holds for them:

\begin{cor}
\label{cor:FOL-substitution}
If $\varphi$ is any pattern and $t$ is a term pattern of sort $s$, then
\begin{quote}
Term Substitution: $\models (\forall x\!:\!s\,.\,\varphi) \ra \varphi[t/x]$
\end{quote}
\end{cor}
\begin{proof}
Let $\rho:\Var \ra M$ be any valuation.
Then
$$\overline{\rho}(\forall x\,.\,\varphi)=\bigcap\{ \overline{\rho'}(\varphi) \mid
\rho'\!\!\upharpoonright_{\Var\backslash\{x\}} =
\rho\!\!\upharpoonright_{\Var\backslash\{x\}}\} \subseteq
\overline{\rho''}(\varphi)$$
where $\rho'':\Var \ra M$ is such that
$\rho''\!\!\upharpoonright_{\Var\backslash\{x\}} =
\rho\!\!\upharpoonright_{\Var\backslash\{x\}}$ and
$\{\rho''(x)\}=\overline{\rho}(t)$.
Such a $\rho''$ exists thanks to Corollary~\ref{cor:terms} and can only be
$\rho''=\rho[m/x]$ where $\overline{\rho}(t)=\{m\}$,
so $\overline{\rho''}(\varphi)=\overline{\rho}(\varphi[t/x])$.
\end{proof}

Note Corollary~\ref{cor:FOL-substitution} generalizes
Corollary~\ref{cor:terms}: pick $\varphi$ to be
$\exists y\,.\,x=y$; then $\forall x:s\,.\,\varphi$
is a tautology and $\varphi[t/x]$ is $\exists y\,.\,t=y$, which
by Proposition~\ref{prop:functional} implies that $t$ is functional.
Corollary~\ref{cor:FOL-substitution} also generalizes (5) in
Proposition~\ref{prop:PL-validity}, because variables are particular terms.

Corollary~\ref{cor:FOL-substitution}, Proposition~\ref{prop:FOL-equality}
and Proposition~\ref{prop:PL-validity} imply that FOL reasoning with or without
equality is sound for matching logic, provided that the Substitution axiom of
FOL is only applied when $t$ is a term pattern.
To avoid confusing the FOL Substitution axiom schema with the matching logic
variant in Corollary~\ref{cor:FOL-substitution}, we called the later Term Substitution.
As shown in Section~\ref{sec:deduction}, Term Substitution can be generalized
a bit into Functional Substitution, which takes functional patterns instead
of term patterns $t$, but in general it is not sound for arbitrary patterns instead of $t$.

Since functional patterns evaluate to singleton sets for any valuation, the
conjunction of two functional patterns evaluate either to the empty set when
the two patterns evaluate to different singleton sets, or to the same singleton
set when the two patterns evaluate to the same singleton set.
Formally,

\begin{prop}
\label{prop:functional-conjunction}
If $\varphi$ and $\varphi'$ are functional patterns of the same sort, then:
\begin{enumerate}
\item $\models ((\varphi \wedge \varphi') = \bot) = (\varphi \neq \varphi')$
\item $\models ((\varphi \wedge \varphi') \neq \bot) = (\varphi = \varphi')$
\item $\models (\varphi \wedge \varphi') = (\varphi \wedge (\varphi = \varphi'))$
\end{enumerate}
\end{prop}
\begin{proof}
Trivial: pick a model $M$ and a valuation $\rho:\Var \ra M$, and apply the
definitions.
\end{proof}

Particular functions with particular properties, such as injective of surjective
functions, can be defined in a conventional way using conventional FOL.
For example, pattern (one-argument functions only,
for simplicity)
$$f(x) = f(y) \ra x = y$$
states that $f$ is injective and pattern
$$
\exists x\,.\,f(x) = y
$$
states that $f$ is surjective.
We only show the former:
if $(M,f_M:M \ra M)$ is any model satisfying $f(x) = f(y) \ra x = y$,
then $f_M$ must be injective.
Indeed, let $a,b\in M$ such that $a\neq b$ and $f_M(a)=f_M(b)$.
Pick $\rho:\Var\ra M$ such that $\rho(x)=a$ and $\rho(y)=b$.
Since $M$ satisfies the axiom above, we get
$\overline{\rho}(f(x) = f(y)) \subseteq \overline{\rho}(x=y)$.
But Proposition~\ref{prop:equality} implies that
$\overline{\rho}(x=y)=\emptyset$ and $\overline{\rho}(f(x) = f(y))=M$,
which is a contradiction.
We can also show that any model whose $f$ is injective satisfies the axiom.
Let $(M,f_M:M \ra M)$ be any model such that $f_M$ is injective.
It suffices to show
$\overline{\rho}(f(x) = f(y)) \subseteq \overline{\rho}(x=y)$
for any $\rho:\Var\ra M$, which follows by Proposition~\ref{prop:equality}:
if $\rho(x)=\rho(y)$ then $\overline{\rho}(f(x) = f(y)) = \overline{\rho}(x=y) = M$,
and if $\rho(x)\neq\rho(y)$ then
$\overline{\rho}(f(x) = f(y)) = \overline{\rho}(x=y) = \emptyset$
because $f_M$ is injective.

With the notation $\varphi \neq \varphi'$ for $\neg(\varphi = \varphi')$
introduced in Section~\ref{sec:equality},
$
x \neq y \ra f(x) \neq f(y)
$
is an alternative way to capture the
injectivity of $f$.

\subsection{Partial Functions}
\label{sec:partial-functions}

In FOL, operation symbols are interpreted as {\em total} functions by default,
meaning that they are defined on all the elements in their domain.
Interpreting function symbols as partial functions leads to a completely
different logic, called {\em partial FOL} in the literature
(see, e.g., \cite{DBLP:journals/sLogica/FarmerG00}), which has many different
axioms to properly capture the desired properties of definedness and
undefinedness.
Our interpretations of symbols into powersets allows us not only to elegantly
define definedness (Section~\ref{sec:definedness}), but also to define
partial functions without a need to develop a different logic.
Specifically, the pattern 
$$
\neg \sigma(x_1,\ldots,x_n) \ \
\vee \ \ \exists y\,.\,\sigma(x_1,\ldots,x_n) = y,
$$
where $\neg \sigma(x_1,\ldots,x_n)$ can be equivalently replaced with
$\sigma(x_1,\ldots,x_n) = \bot_s$, states that
$\sigma\in\Sigma_{s_1\ldots s_n,s}$ is a partial function.
From now on we use the notation (note the ``$\rightharpoonup$'' symbol)
$$
\sigma : s_1 \times \cdots \times s_n \rightharpoonup s
$$
to automatically assume a pattern like the above for $\sigma$.
For example, a division partial function which is undefined
iff the denominator is 0 can be specified as:
$$
\begin{array}{l}
\_\,/\_ : \Nat \times \Nat \rightharpoonup \Nat
\\
x/y = \bot \lra y = 0
\end{array}
$$
which means a symbol $\_\,/\!\_\in\Sigma_{\Nat\times\Nat,\Nat}$ with
pattern axioms $\neg(x/y) \vee \exists z \, .\, x/y = z$ and
$x/y = \bot \lra y = 0$; the latter is equivalent to
$\lceil x/y\rceil = (y\neq0)$ and to
$\lfloor\neg(x/y)\rfloor = (y=0)$.

\subsection{Total Relations}
\label{sec:total-relations}

Recall from Section~\ref{sec:functions} that we can define total functions
using patterns of the form $\exists y\,.\,\sigma(x_1,\ldots,x_n) = y$,
stating not only that the interpretation of $\sigma$ in model $M$,
$\sigma_M$, is defined in any of its arguments, but also that it
has only one value.
We sometimes want to state that relations, not only functions, are total.
All we have to do is to say that the relation is non-empty for any
arguments, which can be easily stated with a pattern of the form
$
\lceil\sigma(x_1,\ldots,x_n)\rceil_s^s
$, equivalent to $\sigma(x_1,\ldots,x_n)\neq\bot_s$.
We write
$$
\sigma : s_1 \times \cdots \times s_n \Ra s
$$
to automatically state that $\sigma$ is a total relation.

\subsection{Constructors, Unification, Anti-Unification}
\label{sec:constructors}

Constructors can be used to build programs, data, as well as semantic
structures to define and reason about languages and programs.
Hence, constructors have been extensively studied in the literature,
using various approaches and logical formalisms.
We believe that classic approaches to constructors can also be
adapted to our matching logic setting, either directly by redefining
the corresponding concepts (e.g., the matching logic analogous to
initial algebras~\cite{Goguen:1977:IAS:321992.321997}, etc.) or
indirectly by translating them together with their underlying logic
to matching logic (e.g., following the translations of FOL and algebraic
specifications to matching logic in Sections~\ref{sec:FOL} and
\ref{sec:alg-spec}, respectively).
Here we discuss a different approach.
Specifically, we show how the dual nature of patterns to specify both
structure and constraints can make some of the definitions and notions
related to constructors more elegant and appealing.
For example, unification and anti-unification can be seen as
conjunction and, respectively, disjunction of patterns.

One main property of constructors is that they collectively can
construct all the elements of their target domain; i.e., the target domain
has ``no junk''~\cite{Goguen:1977:IAS:321992.321997}.
One simple pattern stating that a unary symbol $f$ is to be interpreted as
a surjective relation in every model is $\exists x\,.\,f(x)$.
Generalizing this idea to an arbitrary set of $n$ symbols
$$\{c_i \in \Sigma_{s_i^1...s_i^{m_i},s_i} \ \mid \ 1 \leq i \leq n\}$$
that we want to be constructors for target sort $s$, we get the
following ``no junk'' pattern:
$$
\bigvee_{s_i=s} \exists x_i^1:s_i^1\dots\exists x_i^{m_i}:s_i^{m_i}\ .\ 
c_i(x_i^1,\dots,x_i^{m_i})
$$
For example, applied to the usual $0$ and $\it succ$ constructors of natural
numbers, the above becomes
$0 \ \vee \ \exists x:\Nat \,.\, {\it succ}(x)$.
Indeed, the interpretation of this pattern in the model of natural numbers
is the entire domain; or said differently, any number matches this pattern.

The other main property of constructors is that they yield a unique way to
construct each element in the target domain; i.e., the target domain has
``no confusion''~\cite{Goguen:1977:IAS:321992.321997}.
That means two separate types of conditions, each specifiable with patterns.
First, each constructor $c_i$ builds a set of elements distinct from
that of any other constructor $c_j$ with $s_j=s_i$:
$$
\neg(c_i(x_i^1,\dots,x_i^{m_i}) \wedge c_j(x_j^1,\dots,x_j^{m_j}))
$$
Recall that, by convention, free variables in pattern axioms are assumed
universally quantified.
For our $0$ and $\it succ$ example, the above becomes
$\neg(0 \wedge {\it succ}(x))$, stating that ${\it succ}(x)$
is different from $0$ for any $x$.
Second, each constructor $c_i$ is injective in all its arguments at once
(regarded as a tuple), which can be specified with a pattern as follows:
$$
c_i(x_i^1,\dots,x_i^{m_i}) \wedge
c_i(y_i^1,\dots,y_i^{m_i}) \ra
c_i(x_i^1 \wedge y_i^1,\dots,x_i^{m_i}\wedge y_i^{m_i})
$$
Indeed, the above pattern ensures that in any model $M$, if
$$(c_i)_M(a_i^1,\dots,a_i^{m_i}) \cap
(c_i)_M(b_i^1,\dots,b_i^{m_i}) \neq 0$$
then it must be that $a_i^1=b_i^1$, ..., $a_i^{m_i}=b_i^{m_i}$.

Putting all the above together, below we formally introduce constructors:

\begin{defi}
\label{dfn:constructors}
Given a specification $(S,\Sigma,F)$, the symbols in set
$$\{c_i \in \Sigma_{s_i^1...s_i^{m_i},s_i} \ \mid \ 1 \leq i \leq n\}$$
are called \textbf{constructors} iff they have the following properties:
\begin{description}

\item[No junk]
For any sort $s \in \{s_i\mid 1 \leq i \leq n\}$,
$$
F \models
\bigvee_{s_i=s} \exists x_i^1:s_i^1\dots\exists x_i^{m_i}:s_i^{m_i}\ .\ 
c_i(x_i^1,\dots,x_i^{m_i})
$$
\item[No confusion, different constructors]
For any $1 \leq i \neq j \leq n$ with $s_j=s_i$,
$$
F \models \neg(c_i(x_i^1,\dots,x_i^{m_i}) \wedge c_j(x_j^1,\dots,x_j^{m_j}))
$$
\item[No confusion, same constructor]
For any $1\leq i\leq n$,
$$
F \models
c_i(x_i^1,\dots,x_i^{m_i}) \wedge
c_i(y_i^1,\dots,y_i^{m_i}) \ra
c_i(x_i^1 \wedge y_i^1,\dots,x_i^{m_i}\wedge y_i^{m_i})
$$
\end{description}
Additionally, if each $c_i$ is functional, then we call them
\textbf{functional constructors}.
The usual way to define a set of constructors is to have $F$
include all the patterns above.
\end{defi}

It is easy to see that if the symbol $\sigma$ that occurs in the context
$C_{\sigma,i}[\square]$ in Definition~\ref{dfn:injectivity} is a constructor,
then $C_{\sigma,i}[\square]$ is injective, and thus, by
Propositions~\ref{prop:injective-symbol-distributivity} and
\ref{prop:symbol-distributivity}, it enjoys full distributivity
w.r.t. the matching logic constructs $\wedge$, $\vee$, $\forall$ and
$\exists$.
Thanks to Propositions~\ref{prop:equality} and \ref{prop:FOL-equality},
we can therefore apply these distributivity properties of constructors
in any context.
In addition to the distributivity properties, the following equality
properties of constructors turned out to also be very useful in program
verification efforts:

\begin{prop}
\label{prop:constructors}
Given a set of constructors
$\{c_i \in \Sigma_{s_i^1...s_i^{m_i},s_i} \mid 1 \leq i \leq n\}$
for a specification $(S,\Sigma,F)$, the following hold:
\begin{description}
\item[Case analysis] If $\varphi$ is a pattern of sort
$s \in \{s_i \mid 1 \leq i \leq n\}$, then
$$
F \models 
\varphi = \bigvee_{s_i=s} \exists x_i^1:s_i^1\dots\exists x_i^{m_i}:s_i^{m_i}\ .\ 
\varphi \wedge c_i(x_i^1,\dots,x_i^{m_i})
$$
Additionally, if the constructors and $\varphi$ are all functional, then
$$
F \models 
\varphi = \bigvee_{s_i=s} \exists x_i^1:s_i^1\dots\exists x_i^{m_i}:s_i^{m_i}\ .\ 
c_i(x_i^1,\dots,x_i^{m_i}) \wedge (\varphi = c_i(x_i^1,\dots,x_i^{m_i}))
$$
\item[Different constructors]
If $1\leq i \neq j \leq n$ with $s_i=s_j$, and 
$\varphi_i^1 \in \Pattern_{s_i^1}$,
...,
$\varphi_i^{m_i} \in \Pattern_{s_i^{m_i}}$,
and
$\varphi_j^1 \in \Pattern_{s_j^1}$,
...,
$\varphi_j^{m_j} \in \Pattern_{s_j^{m_j}}$,
then
$$
F \models c_i(\varphi_i^1,\dots,\varphi_i^{m_i})
\wedge c_j(\varphi_j^1,\dots,\varphi_j^{m_j}) = \bot
$$
\item[Same constructor]
If 
$\varphi_i^1,\psi_i^1 \in \Pattern_{s_i^1}$,
...,
$\varphi_i^{m_i},\psi_i^{m_i} \in \Pattern_{s_i^{m_i}}$,
then
$$
F \models
c_i(\varphi_i^1,\dots,\varphi_i^{m_i})
\wedge
c_i(\psi_i^1,\dots,\psi_i^{m_i})
=
c_i(\varphi_i^1\wedge\psi_i^1,\dots,\varphi_i^{m_i}\wedge\psi_i^{m_i})
$$
\end{description}
\end{prop}
\begin{proof}
Case analysis follows by Proposition~\ref{prop:equality}, which reduces equality
($=$) to double implication ($\lra$).
The latter follows by the propositional distributivity of $\wedge$
over $\vee$ and, of course, the ``no junk'' requirement of constructors
in Definition~\ref{dfn:constructors}.
The part where the constructors and $\varphi$ are functional is an immediate
corollary, making use of Proposition~\ref{prop:functional-conjunction}.

Different constructors:
Suppose that there is some model $M$ and valuation
$\rho : \Var \ra M$ such that
$\overline{\rho}(c_i(\varphi_i^1,\dots,\varphi_i^{m_i})
\wedge c_j(\varphi_j^1,\dots,\varphi_j^{m_j})) \neq \emptyset$.
Then there are some elements
$a_i^1 \in \overline{\rho}(\varphi_i^1)$,
...,
$a_i^{m_i} \in \overline{\rho}(\varphi_i^{m_i})$,
and
$a_j^1 \in \overline{\rho}(\varphi_j^1)$,
...,
$a_j^{m_j} \in \overline{\rho}(\varphi_j^{m_j})$
such that
$(c_i)_M(a_i^1,\dots,a_i^{m_i}) \cap (c_j)_M(a_j^1,\dots,a_j^{m_j})
\neq \emptyset$, which contradicts the ``no confusion'' requirement for
different constructors in Definition~\ref{dfn:constructors}.

Same constructor:
By Proposition~\ref{prop:equality}, we can replace $=$ with $\lra$.
The $\leftarrow$ implication is immediate by structural framing,
Proposition~\ref{prop:structural-framing}.
For the $\ra$ implication, let 
$M$ be a model, $\rho : \Var \ra M$ a valuation, and
$b \in
\overline{\rho}(c_i(\varphi_i^1,\dots,\varphi_i^{m_i})
\wedge
c_i(\psi_i^1,\dots,\psi_i^{m_i}))$.
Then there are some elements
$u_i^1 \in \overline{\rho}(\varphi_i^1)$,
...,
$u_i^{m_i} \in \overline{\rho}(\varphi_i^{m_i})$,
and
$v_i^1 \in \overline{\rho}(\psi_i^1)$,
...,
$v_i^{m_i} \in \overline{\rho}(\psi_i^{m_i})$
such that
$b \in
(c_i)_M(u_i^1,\dots,u_i^{m_i})
\wedge
(c_i)_M(v_i^1,\dots,v_i^{m_i}))$,
that is,
$b\in\overline{\rho'}(c_i(x_i^1,\dots,x_i^{m_i}) \wedge
c_i(y_i^1,\dots,y_i^{m_i}))$, where
$\rho':\Var \ra M$ is some valuation that takes
$x_i^1$ to $u_i^1$, ..., $x_i^{m_i}$ to $u_i^{m_i}$,
and
$y_i^1$ to $v_i^1$, ..., $y_i^{m_i}$ to $v_i^{m_i}$.
By the ``no confusion'' requirement for the same constructor
in Definition~\ref{dfn:constructors} we conclude that
$b\in\overline{\rho'}(c_i(x_i^1\wedge y_i^1,\dots,x_i^{m_i}\wedge y_i^{m_i}))$,
that is,
$b\in(c_i)_M(\{u_i^1\}\wedge\{v_i^1\},\dots,\{u_i^{m_i}\}\wedge\{v_i^{m_i}\})$.
This can only happen when $u_i^1 = v_i^1$, ..., $u_i^{m_i} = v_i^{m_i}$.
And in that case it can only be that 
$b \in
\overline{\rho}(c_i(\varphi_i^1 \wedge \psi_i^1,\dots,
\varphi_i^{m_i} \wedge \psi_i^{m_i}))$.
\end{proof}

The case analysis property is useful when additional constraints are
needed on a pattern in order to reason with it.
For example, if $b$ is a Boolean (symbolic) expression in a given
positive context, $C[b]$, then we can replace $b$ with
$\ttrue \wedge (b=\ttrue) \vee \ffalse \wedge (b=\ffalse)$, and then by
Propsitions~\ref{prop:symbol-distributivity} (distributivity) and
\ref{prop:constraint-propagation} (constraint propagation) we can reduce
$C[b]$ to 
$C[\ttrue] \wedge (b=\ttrue) \vee C[\ffalse] \wedge (b=\ffalse)$.
Similarly, if $e$ is a $\Nat$ (symbolic) expression in a positive context 
$C[e]$, then we can reduce the pattern $C[e]$ to the pattern
$C[0] \wedge (e=0)\vee\exists x\,.\,C[\textit{succ}(x)]\wedge(e = succ(x))$,
which may be further reducible.
The other two properties in Proposition~\ref{prop:constructors} are
self-explanatory and clearly useful for the same reasons why constructors
are useful, but we found them particularly useful when attempting to do
symbolic execution using the operational semantics rules of a language.
As explained in \cite{stefanescu-park-yuwen-li-rosu-2016-oopsla}, the main
technical instrument there is unification: indeed, in order to check if
a symbolic program configuration $\varphi$ can be executed with an operational
rule $\textit{left} \Ra \textit{right}$, a unification of $\varphi$ and
$\textit{left}$ is attempted.
If it fails then the rule cannot be applied; if it succeeds then the rule can
be applied and $\varphi$ is advanced to $\textit{right} \wedge \psi$, where
$\psi$ is the constraints resulting from unifying $\varphi$ and $\textit{left}$.
As discussed below, unification can be achieved in matching logic
by pattern conjunction, which makes the last properties in
Proposition~\ref{prop:constraint-propagation} indispensable.

We next show how unification and anti-unification can be explained
as conjunction and, respectively, disjunction of matching logic patterns.
To fall into the conventional setting, for the reminder of this section
assume that all the symbols are functional constructors and all the
starting patterns are term patterns (Definition~\ref{dfn:functional-notation})
built with such constructors.

Let us re-think unification in terms of matching logic and pattern matching.
Consider two patterns $\varphi_1$ and $\varphi_2$ having the same sort.
Each of them is matched by a potentially infinite set of elements.
We are interested in the elements that match both $\varphi_1$ and
$\varphi_2$.
Moreover, we are interested in some pattern $\varphi$ that captures all
these elements.
Clearly $\varphi \ra \varphi_1$ and $\varphi \ra \varphi_2$, and we would
like the most general such $\varphi$, that is, if $\varphi' \ra \varphi_1$
and $\varphi' \ra \varphi_2$ then $\varphi' \ra \varphi$.
We do not have to search for such a $\varphi$ any further, because it is,
by definition, $\varphi_1 \wedge \varphi_2$.
All we have to do then is to simplify $\varphi_1 \wedge \varphi_2$ to a
convenient form, say a term pattern constrained with equalities telling
how $\varphi_1$ and $\varphi_2$ fall as instances,
using matching logic reasoning.
Let us exemplify this when $\varphi_1 \equiv f(g(x),x)$ and
$\varphi_2 \equiv f(y,0)$ where $f$ is a binary symbol (functional construct),
$g$ is unary, and $0$ is a constant:
$$
\begin{array}{rcr}
f(g(x),x) \wedge f(y,0) & = &
\hspace*{10ex} \textrm{(by Proposition~\ref{prop:constructors})}
\\
f(g(x) \wedge y, x \wedge 0) & = & 
\textrm{(by Proposition~\ref{prop:functional-conjunction})}
\\
f(g(x) \wedge (y=g(x)), x \wedge (x = 0)) & = & 
\textrm{(by Proposition~\ref{prop:constraint-propagation})}
\\
f(g(x),x) \wedge (y = g(x)) \wedge (x = 0) & = &
\textrm{(by Proposition~\ref{prop:FOL-equality})}
\\
f(g(0),0) \wedge (y = g(0)) \wedge (x = 0)
\end{array} 
$$
Therefore, using matching logic reasoning we obtained both
the most general unifier of the two term patterns, encoded as
a conjunction of equalities $(y = g(0)) \wedge (x = 0)$, and
the unifying term pattern $f(g(0),0)$.
We believe that unification algorithms should not be difficult to adapt
into matching logic proof search heuristics capable of producing
proofs like above, thus narrowing the gap between tools and certifiable
program verification.

Anti-unification, or generalization, can be dually regarded as disjunction
of patterns\footnote{The author thanks Traian Florin \c{S}erb\u{a}nu\c{t}\u{a}
for observing this.}.
Let us briefly recall Plotkin's original two-rule algorithm \cite{Plotkin:1970}:
\begin{enumerate}
\item
$ f(u_{1},\dots ,u_{n})\sqcup f(v_{1},\ldots ,v_{n})
\rightsquigarrow
f(u_{1}\sqcup v_{1},\ldots ,u_{n}\sqcup v_{n})$,
\item
$u \sqcup v \rightsquigarrow x_{u,v}$ otherwise,
where $x_{u,v}$ is a variable uniquely determined by $u$ and $v$.
\end{enumerate}
Given terms $s$ and $t$, the term obtained by applying this algorithm
to $s \sqcup t$ is their
anti-unification; the corresponding substitutions instantiating 
it to $s$ and, respectively, $t$ are obtained by
instantiating each variable $x_{u,v}$ to $u$ and, respectively, $v$.
For example, $f(g(0),0) \sqcup f(g(1),1) \rightsquigarrow^*
f(g(x_{0,1}),x_{0,1})$.
Indeed, the term $f(g(x_{0,1}),x_{0,1})$ containing one variable, $x_{0,1}$,
is the least general term that is more general than both 
$f(g(0),0)$ and $f(g(1),1)$.
Also, the two original terms can be recovered by substituting $x_{0,1}$ with
$0$ and, respectively, $1$.

Plotkin's algorithm can be mimicked with inference in matching logic.
For the example above, the following matching logic proof blindly follows the
application of Plotkin's algorithm (all proof steps correspond to applications of
proof rules of FOL with equality):
$$
\begin{array}{rcr}
f(g(0),0) \vee f(g(1),1) & = & 
\\
\exists x\,.\,\exists y\,.\,f(x,y) \wedge ( x=g(0) \wedge y = 0 \vee x=g(1) \wedge y=1) & = &
\\
\exists x\,.\,\exists y\,.\,\exists z\,.\,f(x,y) \wedge x = g(z) \wedge
( z=0 \wedge y = 0 \vee z=1 \wedge y=1) & = &
\\
\exists x\,.\,\exists y\,.\,\exists z\,.\,f(x,y) \wedge x = g(z) \wedge
( z=0 \wedge y = z \vee z=1 \wedge y=z) & = &
\\
\exists x\,.\,\exists y\,.\,\exists z\,.\,f(x,y) \wedge x = g(z) \wedge
y = z \wedge (z=0 \vee z=1) & = &
\\
\exists z\,.\,f(g(z),z) \wedge (z=0 \vee z=1)
\end{array}
$$
The resulting pattern contains both the generalization,
$f(g(z),z)$, and the two witness substitutions that can recover
the original terms (encoded as a disjunction of equalities).

\subsection{Built-in Domains}
\label{sec:builtins}

Dedicated solvers and decision procedures specialized for particular but
important mathematical domains abound in the literature.
Some domains may not even have finite descriptions in certain
logics; a classic example is the domain of natural numbers, which does
not admit a finite, not even a recursively enumerable axiomatization
in FOL (G\"odel's incompleteness \cite{Godel1930,Godel1931}).
Therefore, to reason about certain properties that involve natural
numbers, we need to leave FOL.
The standard approach is to assume some oracle for the domain of
interest, which is capable of answering validity questions within
that domain.
At the practical level, such an oracle may be implemented using
specialized procedures and algorithms for that domain, such as those
provided by Z3~\cite{z3-solver}, Yices~\cite{yices},
CVC~\cite{DBLP:conf/cav/BarrettT07,DBLP:conf/cav/BarrettCDHJKRT11}, etc.
At the theoretical level, the set of models of the FOL specification
in question is restricted to those that inherit the desired model
for the built-in data-types, and thus we can assume all the valid
FOL properties of that domain in the rest of the proof even if
those properties are not provable using FOL.

Reasoning with built-in domains can be done exactly the same way
in matching logic: assume the desired sorts and symbols for the built-in
domains, together with all their valid patterns.
Due to their ubiquitous nature, from now on we tacitly assume
definitions of integers and of natural numbers, as well
as of Boolean values, with common operations on them.
We assume that these come with three sorts, $\Int$, $\Nat$ and $\Bool$,
and the operations on them use the conventional syntax;
e.g.,
$\_+\_ : \Int\times\Int \ra \Int$, 
$\_+\_ : \Nat\times\Nat \ra \Nat$, 
$\_>\_ : \Nat\times\Nat \ra \Bool$, 
$\_{\it and}\_ : \Bool \times \Bool \ra \Bool$,
${\it not}\_ : \Bool \ra \Bool$, etc.

Boolean expressions are frequently used in matching logic specifications
to constrain variables that occur in patterns of possibly other sorts.
For example, suppose that in some domain $\Real$ of real numbers
we want to refer to all numbers of the form $1/x$ where $x$ is a
strictly positive integer.
These numbers are precisely matched by the pattern
$1/x \wedge (x > 0 =_\Bool^\Real {\it true})$.
However, this pattern is too verbose.
For the sake of a more compact and easy to read notation, we introduce the
following:

\begin{nota}
\label{notation:bool}
If $b$ is a \underline{proper} Boolean expression, that is, a term pattern of
sort $\Bool$ (Definition~\ref{dfn:functional-notation}), then we will write
just $b$ instead of $b={\it true}$ in any \underline{other} sort context.
\end{nota}

Notation~\ref{notation:bool} allows us to use Boolean expressions
in any sort context, thanks to the additional notational conventions for
equality in Section~\ref{sec:equality}.
For example, we can write $1/x \wedge x>0$ instead of
$1/x \wedge (x > 0 =_\Bool^\Real {\it true})$.

To avoid confusion or even introducing inconsistencies, we urge the reader
to respect the underlined words \underline{proper} and \underline{other} in
Notation~\ref{notation:bool}.
That's because $\Bool$ expressions, when regarded as patterns,
evaluate to one of the singleton sets $\{t\}$ (the true value) or
$\{f \}$ (the false value), while patterns of sort $\Bool$ can evaluate
to any of the four sets $\emptyset$, $\{t\}$, $\{f\}$, or $\{t,f\}$.
For example, consider the $\Bool$ patterns $\top_\Bool$ and $\bot_\Bool$,
which are not proper $\Bool$ expressions and evaluate
to the sets $\{t,f\}$ and $\emptyset$, respectively, and an equality pattern
$\top_\Bool = \bot_\Bool$ which is obviously $\bot$
(regardless of the sort context).
If we carelessly apply the notation above to this pattern we get
$(\top_\Bool = {\it true}) = (\bot_\Bool = {\it true})$,
that is, $\top$;
that's because both $\top_\Bool = {\it true}$ and
$\bot_\Bool = {\it true}$ are $\bot$, and $\bot = \bot$ is $\top$.
So it is important to apply the notation
``$b$ as a shortcut for $b = {\it true}$''
only when it is guaranteed that $b$ evaluates to either $\{t\}$ or $\{f\}$,
such as when $b$ is a proper Boolean expression term.
It is also important to apply the notation only when $b$ occurs in sort
contexts \underline{other} than \Bool.
For example, consider the Boolean expression
$b \equiv ({\it true \,or\, false})$.
If we carelessly apply the notation above to the Boolean sub-expressions
$\it true$ and $\it false$, which occur in $b$ above in $\Bool$ contexts, then
we get $({\it true} = {\it true})\,{\it or}\,({\it false} = {\it true})$,
which is $\top_\Bool\,{\it or}\,\bot_\Bool$, which is $\bot_\Bool$
(by the second item in Definition~\ref{def:rho-bar}).
On the other hand, $b = {\it true}$ is $\top_\Bool$.

When reasoning with matching logic patterns, it is often the case
that various Boolean expression constraints come from various sub-patterns.
We would like to combine them into larger Boolean expressions, which we can
then send to SMT solvers.
The following result allows us to do that:
\begin{prop}
\label{prop:combine-bool}
If $b$, $b_1$ and $b_2$ are proper $\Bool$ expressions, then
\begin{itemize}
\item $\models (b_1 = {\it true} \wedge b_2 = {\it true}) = (b_1\,{\it and}\,b_2 = {\it true})$
\item $\models \neg(b = {\it true}) = ({\it not}\,b = {\it true})$
\end{itemize}
Other similar properties for derived Boolean constructs can be derived from these.
\end{prop}
\begin{proof}
Trivial: each of
$\overline{\rho}(b)$,
$\overline{\rho}(b_1)$,
and $\overline{\rho}(b_2)$
can only be $\{t\}$ or $\{f\}$,
for any valuation $\rho$.
\end{proof}

\newcommand{\Truth}{{\it Truth}}

\section{Instance: Algebraic Specifications and Beyond}
\label{sec:alg-spec}

An algebraic specification is a many-sorted signature
$(S,\Sigma)$ together with a set of
equations\footnote{We only consider unconditional equations here.}
$E$ over $\Sigma$-terms with variables.
The variables are assumed universally quantified over the entire equation.
The models, called $\Sigma$-algebras, are first-order $\Sigma$-models
without predicates where equality is interpreted as the identity relation.
We let $\models_{\it Alg}$ denote the algebraic specification satisfaction
relation; in particular, $E \models_{\it Alg} e$ means that any model
that satisfies all the equations in $E$ also satisfies $e$.

Algebraic specifications play a crucial role in theoretical computer
science and program reasoning, because they are often used to define
abstract data types (ADTs) \cite{Liskov:1974:PAD:942572.807045,goguen1976initial}.
Some common ADTs, which have proved useful in many applications,
are lists (or sequences), sets, multisets, maps, multimaps, graphs, stacks,
queues, priority queues, double-ended queues, double-ended priority queues,
etc.\ \cite{wikiADT}.
These ADTs can be found a variety of formal definitions using algebraic
specifications in the literature, not necessarily always equivalent,
and are easily definable or already pre-defined in algebraic specification
languages such as Maude~\cite{maude-book},
CASL~\cite{mosses_caslReferenceManual_2004},
CafeOBJ~\cite{caferep98},
OBJ~\cite{iobj},
Clear~\cite{Burstall:1977:PTT:1622943.1623045}, etc.,
as well as in many other languages with support for ADTs.

Here we show not only that algebraic specifications can be regarded as
matching logic specifications, but also that the use of matching logic often
allows for more expressive, concise and intuitive specifications of ADTs.
To capture conventional ADTs as matching logic specifications, we need
to do almost nothing besides recalling the conventions and notations
introduced in Section~\ref{sec:useful-symbols}.
Specifically, algebraic equations $t=t'$ in $E$ are regarded as matching
logic equality patterns $t=t'$ (Section~\ref{sec:equality}),
algebraic symbols in $\Sigma$ are interpreted as functional symbols
(Section~\ref{sec:functions}, Definition~\ref{dfn:functional-notation}),
and no other patterns but equations are allowed in specifications.
The resulting matching logic specifications are then precisely
the algebraic specifications not only syntactically, but also semantically:

\begin{prop}
\label{prop:alg-spec}
Let $(S,\Sigma,F)$ be the matching logic specification associated
to the algebraic specification $(S,\Sigma,E)$ as above, that is,
$F$ contains an equality pattern for each equation in $E$, as well
as all the patterns stating that the symbols in $\Sigma$ are
interpreted as functions (see Definition~\ref{dfn:functional-notation}).
Then for any $\Sigma$-equation $e$, we have
$E \models_{\it Alg} e$ iff $F \models e$.
\end{prop}
\begin{proof}
The key observation is that, in a similar style to the proof of
Proposition~\ref{prop:pure-predicate}, there is a bijection between
the matching logic models $M$ satisfying $F$ and the
$(S,\Sigma)$-algebras $M'$ satisfying $E$, such that
${M} \models e$ iff ${M}' \models_{\it Alg} e$
for any $\Sigma$-equation $e$.
The model bijection is defined as follows:
\begin{itemize}
\item $M'_s=M_s$ for each sort $s \in S$;
\item ${\sigma}_{M'} : M_{s_1}\times \cdots \times M_{s_n} \ra M_s$ with
$\sigma_{M'}(a_1,\ldots,a_n)=a$ iff 
$\sigma_M:M_{s_1}\times \cdots \times M_{s_n} \ra {\cal P}(M_s)$ with
$\sigma_M(a_1,\ldots,a_n)=\{a\}$.
Note that this is well-defined because $F$ includes all the patterns stating
that all the symbols are functional, so $\sigma_M$ is a function.
\end{itemize}
This model relationship is easy to prove a bijection, and everything else
follows from it.
\end{proof}

Using the notations introduced so far, Peano natural numbers, for example,
can be defined as the following matching logic specification:
$$
\begin{array}{l@{\ \ \ \ \ \ \ \ \ }l}
0: \ \ra\PNat \\
{\it succ}:\PNat\ra\PNat
& {\it plus}(0,y) = y \\
{\it plus}:\PNat\times\PNat\ra\PNat
& {\it plus}({\it succ}(x),y) = {\it succ}({\it plus}(x,y))
\end{array}
$$
This looks identical to the conventional algebraic specification definition,
which is precisely the point and justifies our notation conventions in
Section~\ref{sec:useful-symbols}.
In particular, the functional notation
(Definition~\ref{dfn:functional-notation})
for the three symbols ensures that they will be interpreted as
functions in the matching logic models.
Also, as seen in the proof of Proposition~\ref{prop:alg-spec}, there is a
bijective correspondence between the matching logic models of the specification
above and the conventional models of the Peano algebraic specification
(we only discuss the loose semantics here, not the initial-algebra semantics
\cite{Goguen:1977:IAS:321992.321997}).

Note, however, that matching logic allows more than just
equational patterns.
%, even without the addition of predicate
%symbols.
For example, we can add to $F$ the pattern
$0 \ \vee \ \exists x \,.\, {\it succ}(x)$
stating that any number is either 0 or the successor of another number.
%This pattern is satisfied, for example, by the canonical model of
%natural numbers.
%Indeed, let $M$ be the model of natural numbers:
%$M = M_\Nat=\{0,1,2,\ldots\}$, $M_0=\{0\}$ and
%$M_{\it succ}(n)=\{n+_\Nat 1\}$.
%Then for any $\rho:\Var \ra M$, we have $\overline{\rho}({0})=\{0\}$
%and $\overline{\rho}(\exists x \,.\, {\it succ}(x)) =
%\bigcup\{\overline{\rho'}({\it succ}(x)) \mid
%\rho'\!\!\upharpoonright_{\Var\backslash\{x\}} =
%\rho\!\!\upharpoonright_{\Var\backslash\{x\}} \}
%= \{n \mid n > 0\}$,
%so $\overline{\rho}({0} \vee \exists x \,.\, {\it succ}(x))=M$.
Nevertheless, since matching logic ultimately has the same expressive
power as predicate logic (Proposition~\ref{prop:mlTopred}), we cannot
finitely axiomatize in matching logic any mathematical domains that
do not already admit finite FOL axiomatizations.
As indicated in Section~\ref{sec:builtins}, we follow the standard
approach to deal with built-in domains, namely we assume them theoretically
presented with potentially infinitely many axioms but implemented using
specialized decision procedures.
Indeed, our matching logic implementation prototype in \K
defers the solving of all domain constraints to Z3~\cite{z3-solver}.

%Also, ${M}\models {0} \vee \exists x . {\it succ}(x)$.
%Indeed, for any $\rho:\Var \ra M$, we have $\overline{\rho}({0})=\{0\}$
%and $\overline{\rho}(\exists x . {\it succ}(x)) =
%\bigcup\{\overline{\rho'}({\it succ}(x)) \mid
%\rho'\!\!\upharpoonright_{\Var\backslash\{x\}} =
%\rho\!\!\upharpoonright_{\Var\backslash\{x\}} \}
%= \{n \mid n > 0\}$,
%so $\overline{\rho}({0} \vee \exists x . {\it succ}(x))=M$.
%

\subsection{Sequences, Multisets and Sets}
\label{sec:collections}

Sequences, multisets and sets are typical ADTs.
Matching logic enables, however, some useful developments and shortcuts.
For simplicity, we only discuss collections over \Nat, and name the
corresponding sorts $\Seq$, $\MSet$, and $\Set$, respectively.
Ideally, we would build upon an order-sorted algebraic signature setting,
e.g. following the approach in \cite{GOGUEN1992217},
so that we can regard $x\!:\!\Nat$ not only as an element of sort $\Nat$,
but also as one of sort $\Seq$ (a one-element sequence), as one of sort
$\MSet$, as well as one of sort $\Set$.
Extending matching logic to an order-sorted setting is beyond the scope of
this paper.
The reader who is not familiar with order-sorted concepts can safely
assume that elements of sort
$\Nat$ used in a $\Seq$, $\MSet$, or $\Set$ context are wrapped with
injection symbols.
The symbols below can have many equivalent definitions.

Sequences can be defined with two symbols and corresponding equations:
$$
\begin{array}{l@{\hspace*{8ex}}l}
\epsilon : \ \ra \Seq & \epsilon \cdot x = x \\
\_\cdot\_ : \Seq \times \Seq \ra \Seq & 
x \cdot\epsilon = x \\
& (x \cdot y) \cdot z = x \cdot (y \cdot z)
\end{array}
$$
We assume that lowercase variables have sort $\Nat$ and uppercase
variables have the appropriate collection sort.
%Also, since all signature symbols are interpreted as functions, to avoid
%clutter we tacitly assume their function axioms.
To avoid adding initiality constraints on models, yet be able to do proofs
by case analysis and elementwise equality, we can add the following patterns:
$$
\begin{array}{l}
\epsilon \mathrel\vee \exists x\,.\,\exists S.\,x\cdot S
\\
(x\cdot S = x'\cdot S') = (x = x') \wedge (S = S')
\end{array}
$$
The second axiom above holds strictly for sequences, but not for
commutative collections, so it needs to be removed later
when we add the commutativity axiom to define multisets and sets.
We next define two common operations on sequences:
$$
\begin{array}{l@{\ \ \ \ \ \ \ \ \ \ \ \ \ }l}
\reverse : \Seq \ra \Seq &
\_\in\_ : \Nat \times \Seq \ra \Bool
\\
\reverse(\epsilon) = \epsilon & 
x \in \epsilon = \ffalse \hfill \\ %x \in x \cdot S \\
\reverse(x \cdot S) = \reverse(S) \cdot x &
 x \in y \cdot S = (x = y \wedge \ttrue) \mathrel{\textit{or}} x \in S
\end{array}
$$
%The above may look like a conventional FOL definition, but note that
%here are no predicates involved: each axiom is a pattern.

To illustrate the flexibility of matching logic, we next define up-to
and Fibonacci sequences of natural numbers.
%\vspace*{-1ex}
$$
\begin{array}{l@{\hspace*{10ex}}l}
\upto : \Nat \ra \Seq &
%\upto(0) = \epsilon \\
%& (\upto(n) \wedge n>0) = \upto(n-1) \cdot n
%\\ & 
\upto(n) = (n = 0 \wedge \epsilon \vee n > 0 \wedge \upto(n-1) \cdot n)
\end{array}
%\vspace*{-1ex}
$$
This specification needs to be explained.
Let $M$ be a model satisfying the above.
First recall Notation~\ref{notation:bool}. % and the discussion following it.
For notational simplicity, assume that $M_\Nat$ and $M_\Seq$ are the sets
of natural numbers and of comma-separated sequences of natural numbers,
respectively.
We show by induction on $m$ that $\upto_M(m)=\{1 \cdot 2 \cdots m\}$.
If $m=0$ then the second disjunct of the axiom is $\emptyset$ and thus the
first disjunct ensures that $\upto_M(0)=\epsilon_M=\epsilon$.
If $m>0$ then the first disjunct is $\emptyset$ and thus the second disjunct,
with $\rho:\Var\ra M$ such that
$\rho(n)=m$, yields
$\overline{\rho}(\upto(n)) = \overline{\rho}(\upto(n-1) \cdot n)$,
that is,
$\upto_M(m) =
 \upto_M(m-1) \cdot m$,
which by the induction hypothesis implies
$\upto_M(m)= \{1 \cdot 2 \cdots m-1\} \cdot m = \{1 \cdot 2 \cdots m\}$.

Similarly, we can specify a function that defines the sequence of Fibonacci
numbers up to a given number $n>0$:
$$
\begin{array}{l@{\hspace*{10ex}}l}
{\it fib} : \Nat \ra \Seq &
{\it fib}(n) = (n = 0 \wedge 0 \vee n > 0 \wedge {\it aux}(n,\ 0 \cdot 1)) \\
{\it aux} : \Nat \times \Seq \rightharpoonup \Seq &
{\it aux}(1,\ S) = S \\
& n >1 \ra {\it aux}(n,\ S \cdot x \cdot y) = {\it aux}(n-1,\ S \cdot x \cdot y \cdot x+y)
\end{array}
$$

We conclude the discussion on sequences with an elegant
means to sort sequences following a bubble-sort methodology:
$$
(x \cdot y \wedge x>y) = y \cdot x
$$ 
Since equations are symmetric, the above effectively allows
to prove (so far only semantically, in a model) a sequence equal to any of
its permutations, i.e., sequences become multisets.
If the equations were oriented, like they are in reachability
logic~\cite{rosu-stefanescu-ciobaca-moore-2013-lics},
then the above would yield a sequence sorting procedure.

We can transform sequences into multisets adding the equality axiom
$
x \cdot y = y \cdot x
$,
%Now we may be interested in counting how many times a given number occurs
%in a multiset:
%$$
%\begin{array}{l}
%{\it count} : \Nat \times \MSet \ra \Nat \\[1.5ex]
%{\it count}(x,x \cdot S) = 1+{\it count}(x,S) \\
%{\it count}(x,S) \wedge \neg(x\in S) = 0
%\end{array}
%$$
%The membership is the same as above, but with
%multisets instead of sequences.
%Note that we used the convention that terms of sort $\Bool$ are
%automatically assumed equal to $\it true$.
%
and into sets by also including $x \cdot x = \bot$ or $x \cdot x = x$.
Here is one way to axiomatize intersection:
$$
\begin{array}{@{}l@{\hspace*{7ex}}l}
\_\cap\_ : \Set \times \Set \ra \Set &
\epsilon \cap S_2 = \epsilon \\
& (x \cdot S_1) \cap S_2 = ((x\in S_2 \ra x)\wedge(\neg(x\in S_2)\ra\epsilon)) \cdot (S_1 \cap S_2)
\end{array}
$$

%$$
%\begin{array}{l|l}
%\_\cap\_: \Set \times \Set \ra \Set &
%\ \_\Delta\_ : \Set \times \Set \ra \Set \\
%\epsilon \cap S_2 = S_2 &
%x \cdot S_1 \mathrel\Delta x \cdot S_2 = S_1 \mathrel\Delta S_2 \\
%x \!\cdot\! S_1 \!\cap\! S_2 = ((x\!\in\! S_2 \ra x)\wedge(\neg(x\!\in\! S_2)\ra \epsilon)) \!\cdot\! (S_1 \!\cap\! S_2) &
%(S_1 \!\cap\! S_2 = \epsilon)  \ra (S_1 \Delta S_2 = S_1 \!\cdot\! S_2)
%\end{array}
%$$

\subsection{Maps}
\label{sec:maps}

Finite-domain maps are also a typical ADT.
Due to their ubiquity in defining algebraic specifications, maps are
usually predefined in languages like those mentioned in the preamble
of this section (Section~\ref{sec:alg-spec}).
For example, below is a snippet of Maude's \verb|MAP|
module~\cite{maude-book}, which is
parametric in both the source and the target domains:

\begin{quote}
\footnotesize
\begin{lstlisting}
fmod MAP{X :: TRIV, Y :: TRIV} is
  sorts Entry{X,Y} Map{X,Y} .
  subsorts Entry{X,Y} < Map{X,Y} .
  op _|->_ : X$Elt Y$Elt -> Entry{X,Y} [ctor] .
  op empty : -> Map{X,Y} [ctor] .
  op _,_ : Map{X,Y} Map{X,Y} -> Map{X,Y} [assoc comm id: empty ctor] .
...
endfm
\end{lstlisting}
\end{quote}

\noindent
Note that the map concatenation symbol, ``\verb|_,_|'', is
associative (attribute ``\verb|assoc|''),
commutative (``\verb|comm|''),
and has the ``\verb|empty|'' map as identity (``\verb|id: empty|'').
The attribute ``\verb|ctor|'' states that the corresponding symbols
are constructors.
Additional axioms, not shown here, ensure that maps are always
well-formed, that is, maps with multiple bindings of the same key are
disallowed.
When instantiated with natural numbers for both the domain and the target,
this \verb|MAP| module defines well-formed finite-domain maps such as
``\verb#1 |-> 5,  2 |-> 0, 7 |-> 9, 8 |-> 1#''.
In Section~\ref{sec:sep-logic}, to show how separation logic can be framed
as a matching logic theory with essentially zero representational/encoding
distance, we will pick an instance of maps with natural numbers as both
the domain and the co-domain,
and will rename \verb|empty| to $\SLemp$ and \verb|_,_| to
$\_\SLstar\_$.

\section{Instance: First-Order Logic}
\label{sec:FOL}

First-order logic (FOL) extends predicate logic with
function symbols and allows predicates to apply to terms
with variables built using the function symbols.
Recall from Section~\ref{sec:predicate-logic} that by
``predicate logic'' in this paper we mean what is also called
``pure predicate logic'' in the literature, namely FOL
without any function or constant symbols.

Formally, given a FOL signature $(S,\Sigma,\Pi)$,
the syntax of its (many-sorted) formulae is:
%%\vspace*{-1ex}
$$
\begin{array}{rrl}
t_s & ::= & x \in \Var_s \\
& \mid & \sigma(t_{s_1},\ldots,t_{s_n}) \ 
\mbox{ with } \sigma\in\Sigma_{s_1 \ldots s_n,s} \\
\varphi & ::= & \pi(t_{s_1},\ldots,t_{s_n}) \ 
\mbox{ with }\pi\in\Pi_{s_1 \ldots s_n} \\
& \mid & 
 \neg \varphi \\
& \mid & \varphi \wedge \varphi \\
& \mid & \exists x . \varphi
\end{array}
$$
Compare the above with the syntax of matching logic in
Section~\ref{sec:matching-logic}.
Unlike FOL, matching logic does not distinguish between the
term and predicate syntactic categories, reason for which
its syntax is in fact more compact than FOL's.
Moreover, matching logic allows logical constructs
over all the syntactic categories, not only over predicates,
and locally where they are needed instead of only at the top,
predicate level.
Also, matching logic allows quantification over any sorts,
including over sorts of symbols thought of as predicates.

Like with predicate logic (Section~\ref{sec:predicate-logic}),
we can instantiate matching logic to capture FOL as is, modulo
the notational conventions in Section~\ref{sec:useful-symbols}
but without any translations from one logic to the other.
Like in predicate logic, we add a {\Pred} sort and regard
the FOL predicate symbols as matching logic symbols of result
{\Pred}, and disallow variables of sort {\Pred} and restrict
the use of logical connectives and quantifiers to only patterns
of sort $\Pred$.
Then there is a one-to-one correspondence between FOL formulae
and matching logic patterns of sort $\Pred$; we let $\varphi$ range
over them.
Moreover, following the approach in Section~\ref{sec:functions},
we constrain each FOL operational symbol
$\sigma\in\Sigma_{s_1 \ldots s_n,s}$ to be interpreted as
a function, that is, with the notation introduced in
Section~\ref{sec:functions}, we write the symbols meant to be
functions as $\sigma : s_1 \ldots s_n \ra s$.
Formally, let $(S^{\it ML},\Sigma^{\it ML})$ be the matching
logic signature with $S^{\it ML}=S\mathrel\cup\{{\Pred}\}$ and
$\Sigma^{\it ML}=\Sigma \mathrel\cup
\{\pi : s_1\ldots s_n \ra {\Pred} \mid \pi\in\Pi_{s_1\ldots s_n}\}$,
and let $F$ be
$
\{
  \exists z\!:\!s \,.\,\sigma(x_1\!:\!s_1,\ldots,x_n\!:\!s_n) = z
  \mid
  \sigma\in\Sigma_{s_1\ldots s_n,s}
\}
$
stating that each symbol in $\Sigma$ is a function.

\begin{prop}
\label{prop:FOL}
For any FOL formula $\varphi$, we have
$\models_{\it FOL}\varphi$ iff $F \models \varphi$.
\end{prop}
\begin{proof}
The proof is similar to that of Proposition~\ref{prop:pure-predicate}.
Like there, the implication 
``$\models_\FOL \varphi$ implies $F \models \varphi$''
follows by the completeness of 
FOL.
Indeed, it is well-known that the properties in
Proposition~\ref{prop:PL-validity} and Corollary~\ref{cor:FOL-substitution},
when regarded as proof rules for deriving FOL formulae $\varphi$,
yield a sound and complete proof system for FOL \cite{Godel1930}.
That is, ``$\models_\FOL \varphi$ iff $\vdash_\FOL\varphi$''.
However, since Proposition~\ref{prop:PL-validity} and
Corollary~\ref{cor:FOL-substitution} show that these proof rules are also
sound for matching logic in the context of $F$
(Corollary~\ref{cor:FOL-substitution} requires that $t$ is a term pattern in
the substitution rule, which holds in the context of $F$), we conclude
that ``$\vdash_\FOL\varphi$ implies $F \models \varphi$''.
Therefore, ``$\models_\FOL \varphi$ implies $F \models \varphi$''.

For the other implication, we also follow the idea in
Proposition~\ref{prop:pure-predicate}.
From any FOL model 
$M^{\it FOL}=(\{M_s^{\it FOL}\}_{s\in S},\{\sigma_{M^{\it FOL}}\}_{\sigma\in\Sigma},\{\pi_{M^{\it PL}}\}_{\pi\in\Pi})$
we can construct a matching logic model
$M^{\it ML}=(\{M_s^{\it ML}\}_{s\in S\cup\{\Pred\}},\{\sigma_{M^{\it ML}}\}_{\sigma\in\Sigma}\cup\{\pi_{M^{\it ML}}\}_{\pi\in\Pi})$,
where $M^{\it ML}_s=M^{\it FOL}_s$ for all sorts $s\in S$ and
$M_\Pred^{\it ML}=\{\star\}$ (with $\star$ some arbitrary but fixed element),
$\sigma_{M^{\it ML}}(a_1,\ldots,a_n)=\{\sigma_{M^{\it FOL}}(a_1,\ldots,a_n)\}$,
and
$\pi_{M^{\it ML}}(a_1,\ldots,a_n)=\{\star\}$ iff
$\pi_{M^{\it PL}}(a_1,\ldots,a_n)$ holds,
and otherwise $\pi_{M^{\it ML}}(a_1,\ldots,a_n)=\emptyset$.
Note that $M^{\it ML} \models F$: indeed,
$\sigma_{M^{\it ML}}(a_1,\ldots,a_n)$ contains precisely one element
for any $\sigma\in\Sigma_{s_1\ldots s_n,s}$ and any
$a_1\in M^{\it ML}_{s_1}$, ..., $a_n\in M^{\it ML}_{s_n}$,
namely $\sigma_{M^{\it FOL}}(a_1,\ldots,a_n)$.
It therefore suffices to show, for any FOL formula $\varphi$, that 
$M^{\it FOL}\models_\FOL\varphi$ iff $M^{\it ML}\models_\ML\varphi$.
Like in Proposition~\ref{prop:pure-predicate}, we can show by structural
induction on $\varphi$ that for any $\rho : \Var \ra M^{\it FOL}$, it is the
case that $M^{\it FOL},\rho\models_\FOL\varphi$ iff
$\overline{\rho}(\varphi)=\{\star\}$.
The induction proof differs from that in Proposition~\ref{prop:pure-predicate}
only in the base case, where we need to notice that term patterns are
functional in $M^{\it ML}$, thanks to Corollary~\ref{cor:terms}, and that
$\overline{\rho}(t) = a$ for a term $t$ in the FOL context
(with $\rho : \Var \ra M^{\it FOL}$) iff
$\overline{\rho}(t) = \{a\}$ in the matching logic context
(where $\rho$ is regarded as $\rho : \Var \ra M^{\it ML}$, which is possible
as we disallow $\Pred$ variables).
\end{proof}

Consequently, FOL is also a methodological fragment of matching logic.
Moreover, since the rules of FOL where we replace all the
predicate meta-variables with pattern meta-variables are sound
for matching logic, we can use off-the-shelf decision procedures
and solvers for FOL or fragments of it {\em unchanged} when doing
matching logic reasoning.
The only thing we have to be careful about is to distinguish the
term patterns from arbitrary patterns, and make sure that the
Substitution rule of FOL is only applied with $t$ a term pattern,
otherwise it may be unsound.
Section~\ref{sec:deduction} discusses this in detail.

Predicate logic and FOL with equality also fall as methodological
fragments of matching logic.
In addition to the constructions in Section~\ref{sec:predicate-logic}
and, respectively, above in this section, all we have to do is to
is to make use of the definedness symbols that we assume by convention
included in all signatures (Section~\ref{sec:definedness}), which give
us equality as an alias as described in Section~\ref{sec:equality}.
We leave the details as an exercise for the interested reader.

Like Boolean expressions, FOL is also frequently used in matching logic
specifications to constrain variables that occur in patterns of possibly
other sorts.
Consider the same example we discussed before and after
Notation~\ref{notation:bool}, where we want to refer to all real numbers
of the form $1/x$ with $x$ a strictly positive integer, but this time
using a given predicate ${\it positive?}(x)$ that tells whether $x$ is
positive.
We can use the pattern
$1/x \wedge ({\it positive?}(x) =_\Pred^\Real \top_\Pred)$,
but that is too verbose.
We would like to just write
$1/x \wedge {\it positive?}(x)$.
Following Notation~\ref{notation:bool}, we introduce the following
similar notation for predicates instead of Boolean expressions:

\begin{nota}
\label{notation:pred}
If $\varphi$ is a FOL formula,
we take the freedom to write $\varphi$ instead of $\varphi=\top_\Pred$.
\end{nota}

Since both the FOL formulae and the patterns of $\Pred$ sort evaluate
to only two possible values, $\emptyset$ or $\{\star\}$, unlike
Notation~\ref{notation:bool} we can freely apply the notation above
in any contexts, including of sort $\Pred$.
Note, however, that care must be exercised to ensure that $\varphi$ is indeed
a FOL formula.
For example, if one extends FOL with additional formula constructs,
like separation logic does for example (Section~\ref{sec:sep-logic}),
then the above notation may lead to inconsistencies.
As discussed in Section~\ref{sec:sep-logic} in detail, matching logic
has a different way to deal with such extensions (allowing different sorts
of ``predicates''), without polluting the universe of FOL formulae and thus
allowing the notation above to still apply.

Same as with Boolean expressions in Proposition~\ref{prop:combine-bool},
we sometimes need to combine various FOL constraints resulting from
various sub-patterns in order to make appropriate calls to FOL provers,
e.g., SMT solvers.
The following result allows us to do that:
\begin{prop}
\label{prop:combine-prop}
If $p$, $p_1$ and $p_2$ are FOL formulae, then
\begin{itemize}
\item $\models (p_1 = \top_\Pred \wedge p_2 = \top_\Pred) = (p_1 \wedge p_2 = \top_\Pred)$
\item $\models \neg(p = \top_\Pred) = (\neg p = \top_\Pred)$
\end{itemize}
Other similar properties for derived FOL constructs can be derived from these.
\end{prop}
\begin{proof}
Trivial: each of
$\overline{\rho}(p)$,
$\overline{\rho}(p_1)$,
and $\overline{\rho}(p_2)$
can only be $\emptyset$ or $\{\star\}$,
for any valuation $\rho$.
\end{proof}

\section{Instance: Modal Logic}
\label{sec:modal-logic}

It turns out that the vanilla matching logic over just one sort with
(countably many) constants and definedness (as defined in
Section~\ref{sec:definedness}) captures one of the most popular
modal logics,
S5~\cite{Becker1930,kripke1959,Goldblatt:2003:MML:969657.969658}.
At the end of this section we briefly discuss how other modal logics can
also be framed as matching logic instances, but until there we only discuss
S5 and thus take the liberty to implicitly mean the ``S5 modal logic''
whenever we say ``modal logic''.

We start by giving the syntax and semantics of modal logic.
Let $\Var_\Prop$ be a countable set of {\em propositional variables}
$p$, $q$, etc.
Then the modal logic syntax is defined as follows:
$$
\begin{array}{rrl}
\varphi & ::= & \Var_\Prop \\
& \mid & \neg \varphi \\
& \mid & \varphi \ra \varphi \\
& \mid & \Box \varphi
\end{array}
$$
The remaining propositional constructs $\wedge$, $\vee$ and $\lra$, can
be defined as derived constructs.
Therefore, syntactically, modal logic adds the $\Box$ construct to
propositional logic, which is called {\em necessity}: $\Box \varphi$ is read
``it is necessary that $\varphi$''.
The dual {\em possibility} construct can be defined as a derived construct:
$\Diamond \varphi \equiv \neg \Box \neg \varphi$ is read
``it is possible that $\varphi$''.
Semantically, the truth value of a formula is relative to a ``world''.
Propositions can be true in some worlds but false in others, and thus formulae
can also be true in some worlds but not in others:

\newcommand{\World}{{\it World}}

\begin{defi}
\label{dfn:modal-logic}
Let $W$ be a set of \textbf{worlds}.
Mappings $v : \Var_\Prop \times W \ra \{\ttrue,\ffalse\}$, called
\textbf{(modal logic) $W$-valuations}, state that each proposition
only holds in a given (possibly empty or total) subset of worlds.
Valuations extend to modal logic formulae:
\begin{itemize}
\item $v(\neg \varphi, w) = \ttrue$ iff $v(\varphi, w) = \ffalse$
\item $v(\varphi_1 \ra \varphi_2, w) = \ttrue$ iff $v(\varphi_1, w) = \ffalse$ or $v(\varphi_2, w) = \ttrue$
\item $v(\Box \varphi, w) = \ttrue$ iff $v(\varphi, w') = \ttrue$ for every $w' \in W$
\end{itemize}
Formula $\varphi$ is \textbf{valid in $W$}, written $W \models_{\rm S5} \varphi$,
iff $v(\varphi, w) = \ttrue$ for any $W$-valuation $v$ and any $w \in W$.
Formula $\varphi$ is \textbf{valid}, written $\models_{\rm S5} \varphi$, iff
$W \models_{\rm S5} \varphi$ for all $W$.
\end{defi}

Modal logic (S5) admits the following sound and complete
proof system~\cite{Becker1930,kripke1959}:
\begin{itemize}[label=\rm\textbf{(M)}]
\item[\rm\textbf{(N)}] Rule: If $\varphi$ derivable then $\Box \varphi$ derivable
\item[\rm\textbf{(K)}] Axiom: $\Box(\varphi_1 \ra \varphi_2) \ra (\Box \varphi_1 \ra \Box \varphi_2)$
\item[\rm\textbf{(M)}] Axiom: $\Box \varphi \ra \varphi$
\item[\rm\textbf{(5)}] Axiom: $\Diamond \varphi \ra \Box\Diamond\varphi$
\end{itemize}

We next show that we can define a matching logic specification $(S,\Sigma,F)$
which faithfully captures modal logic, both syntactically and semantically.
The idea is quite simple: $S$ contains precisely one sort, say $\World$;
$\Sigma$ contains one constant symbol $p\in\Sigma_{\lambda,\World}$ for
each propositional variable $p\in\Var_\Prop$,
plus a unary symbol $\Diamond\in\Sigma_{\World,\World}$;
and $F$ contains precisely one axiom stating
that $\Diamond$ is the definedness symbol (Section~\ref{sec:definedness}),
namely $\Diamond x\!:\!\World$ ($x$ is a free $\World$ variable in this
pattern).
Then we let $\Box \varphi$ be the totality construct
(Notation~\ref{notation:totality}), that is, syntactic sugar for
$\neg\Diamond\neg\varphi$.
Note that any modal logic formula $\varphi$ can be regarded, as is,
as a ground matching logic pattern over this signature;
by ``ground'' we mean a pattern without variables, so the other implication
is also true, because disallowing variables includes disallowing quantifiers.
Moreover, Corollary~\ref{cor:totality-NKM5} implies that the modal logic proof
system above is sound for the resulting matching logic specification,
so $\models_{\rm S5} \varphi$ implies $\models \varphi$.
We show the stronger result that the world/valuation models of modal logic
are essentially identical the the matching logic $(S,\Sigma,F)$-models, and
thus:

\begin{prop}
\label{prop:modal-logic}
For any modal logic formula $\varphi$, we have
$\models_{\rm S5} \varphi$ iff $\models \varphi$.
\end{prop}
\begin{proof}
For any world $W$ and $W$-valuation
$v : \Var_\Prop \times W \ra \{\ttrue,\ffalse\}$
(Definition~\ref{dfn:modal-logic}),
let $M_{W,v}$ be the matching logic $(S,\Sigma,F)$-model
whose carrier is $W$, whose constant symbols
$p \in \Sigma_{\lambda,\World}$
(i.e., $p\in\Var_\Prop$) are interpreted as the sets of worlds
$p_{M_{W,v}} = \{w \in W \mid v(p,w) = \ttrue\}$, and
$\Diamond w$ is the total set $W$ for each $w\in W$.
Similarly, for each matching $(S,\Sigma,F)$-model $M$ let
$W_M = M_\World$ be its carrier and let
$v_M : \Var_\Prop \times W_M \ra \{\ttrue,\ffalse\}$
be defined as $v(p,w) = \ttrue$ iff $w\in p_M$.
It is clear that the two mappings defined above,
$(W,v) \mapsto M_{W,v}$ and respectively $M \mapsto (W_M,v_M)$, are inverse
to each other.

Since a modal logic formula $\varphi$ can be regarded as a matching logic
pattern with no variables,
$\overline{\rho}(\varphi)$ only depends on the model $M$ but not on any
particular valuation $\rho:\Var\ra M$ (by (1) in Proposition~\ref{prop:simple}).
Let us then use the notation $\varphi_M$ for the (unique) interpretation of
$\varphi$ in matching logic model $M$; note that
$\varphi_M \subseteq M_\World$.

We show that for any $W$ and any $W$-valuation $v$,
we have $v(\varphi,w) = \ttrue$ iff $w \in \varphi_{M_{W,v}}$.
We show it by structural induction on $\varphi$.
The cases when $\varphi$ is a propositional symbol or a logical connective are
trivial.
For the necessity modal construct $\Box$, we have
$v(\Box\varphi,w) = \ttrue$ iff
$v(\varphi,w') = \ttrue$ for all $w' \in W$ (Definition~\ref{dfn:modal-logic}),
iff $w' \in \varphi_{M_{W,v}}$ for all $w'\in W$ (induction hypothesis), iff
$\varphi_{M_{W,v}} = W = (M_{W,v})_\World$, iff
$(\Box\varphi)_{M_{W,v}} = W$ (by Proposition~\ref{prop:totality}),
iff $w \in (\Box\varphi)_{M_{W,v}}$ (by Proposition~\ref{prop:totality}).
Therefore, $v(\varphi,w) = \ttrue$ iff $w \in \varphi_{M_{W,v}}$.

We are now ready to prove the main result:
$\models_{\rm S5} \varphi$ iff 
$v(\varphi,w) = \ttrue$ for any $W$ and $W$-valuation $v$ and $w \in W$
(Definition~\ref{dfn:modal-logic}),
iff
$w \in \varphi_{M_{W,v}}$ for any $W$ and $W$-valuation $v$ and
$w \in W$ (by the property proved above by structural induction),
iff
$\varphi_{M_{W,v}}=W$ for any $W$ and $W$-valuation $v$,
iff
$M_{W,v} \models \varphi$ for any $W$ and $W$-valuation $v$
(by Proposition~\ref{dfn:satisfaction}), iff
$M \models \varphi$ for any matching logic $(S,\Sigma,F)$-model $M$
(because of the bijective correspondence between pairs $(W,v)$ and
$(S,\Sigma,F)$-models $M$ proved above), 
iff
$\models \varphi$.
\end{proof}

The result above, together with the general translation of matching
logic to predicate logic with equality discussed in
Section~\ref{sec:PL-reduction}, will also give us a translation of
modal logic to predicate logic with equality.
Translations from modal logic to various types of first-order (or
second-order or other even more expressive) logics are well-known in the
literature, one of them to predicate logic being called the ``standard
translation''~\cite{vanBenthem1983-VANMLA,Blackburn:2001:ML:381193}.
Our goal in this section was not to propose yet another translation,
but to show how modal logic can be framed as a matching logic specification
{\em without any translation}.

There are many variants of modal logic~\cite{vanBenthem1983-VANMLA,Blackburn:2001:ML:381193,Goldblatt:2003:MML:969657.969658}.
One may naturally wonder if all of them can be similarly regarded as
matching logic theories.
While systematically investigating each and everyone of them seems tedious
and likely not worth the effort, it is nevertheless interesting to note that
there is an immediate connection between one of the most general variants
of modal logic, called {\em multi-dimensional} or {\em polyadic modal logic}
\cite{Blackburn:2001:ML:381193}, and
matching logic.
Instead of particular unary modal operators like $\Box$ and $\Diamond$,
polyadic modal logic allows arbitrary operators taking any number of formula
arguments; if $\Delta$ is such an operator of $n$ arguments and
$\varphi_1$, ..., $\varphi_n$ are formulae, then
$\Delta(\varphi_1,...,\varphi_n)$ is also a formula.
In models, called {\em general frames}, each such operator $\Delta$ is
associated a relation $R_\Delta$ of $n+1$ arguments.
Propositional variables are also interpreted as sets (of ``worlds in which
they hold'') by valuations, and given 
set of worlds $W$, valuation $v : \Var_\Prop \times W \ra \{\ttrue,\ffalse\}$
and world $w \in W$,
we have
$v(\Delta(\varphi_1,...,\varphi_n),w) = \ttrue$
iff there are
$w_1,...,w_n \in W$ such that
$v(\varphi_1,w_1) = \ttrue$,
...,
$v(\varphi_n,w_n) = \ttrue$
and
$R_\Delta(w,w_1,...,w_n)$.

It is easy to associate a matching logic specification $(S,\Sigma,F)$
to any polyadic modal logic.
Like for $S5$, we let $S$ contain precisely one sort, $\World$, and
$\Sigma$ contain one constant symbol $p\in\Sigma_{\lambda,\World}$ for
each propositional variable $p\in\Var_\Prop$.
Further, we add a symbol
$\Delta \in \Sigma_{\World \times \cdots \times\World,\World}$ of $n$ arguments
for each polyadic modal operator $\Delta$ taking $n$ arguments.
Then any polyadic modal logic formula $\varphi$ can be regarded without
any change/translation as a matching logic formula.
Further, any axioms/schemas in polyadic modal logic can be added as matching
logic axioms/schemas in $F$.
Then we can extend Proposition~\ref{prop:modal-logic} to show $\models \varphi$
in the polyadic modal logic iff $F \models \varphi$ in matching logic;
the key technical insight here is that there is a bijective correspondence
between relations of $n+1$ arguments and functions of $n$ arguments returning
sets of elements.

When compared to polyadic modal logic, matching logic has a couple of
advantages which, in our view, make it more appealing to use in practice.
First, it has sorts.
Thus, unlike polyadic modal logic which only has ``formulae'', matching logic
allows us to have patterns of various types.
For example, in Section~\ref{sec:example} we show how heap patterns interact
with program patterns and how all can be put together in configuration patterns;
while possible in theory, it would be quite inconvenient to force all patterns
to have the same sort.
Second and more importantly, modal logic and matching logic have a
different interpretation for ``variables''.
In modal logic (propositional) variables are interpreted as sets and we are not
allowed to quantify over them,
while in matching logic variables are interpreted as just elements and we can
quantify over them.
Like shown above, the set interpretation can be recovered in matching logic
by associating constant symbols to propositional variables.
But the singleton interpretation of variables in matching logic, combined with
the capability to quantify over variables of any sort, allows us to
elegantly define many useful properties, such as those in
Section~\ref{sec:useful-symbols}.
For example, the simple pattern
$\forall x\,.\,\lceil x \rceil$ defines the semantics of the definedness
symbol $\lceil\_\rceil$, which as seen above gives us the $\Diamond$ construct
of S5.
It is critical that $x$ ranges over singleton elements in models.
If one attempts to do the same in polyadic modal logic naively replacing $x$ with
a propositional variable $p$, then one gets an inconsistent theory (because
we want $\lceil p \rceil$ to be false when $p$ is interpreted as the empty set
of worlds).
Definedness then allows us to define membership and equality, and thus allows
us to use patterns like $\forall x\,.\,\exists y\,.\,f(x)=y$ to state that
symbol $f$ is a function, etc.

Whether the results and observations above have practical relevance
remains to be seen.
We hope they at least enhance our understanding of both matching logic
and modal logic.

\section{Instance: Separation Logic}
\label{sec:sep-logic}

\newcommand{\SLisatom}{{\sf isatom?}}
\newcommand{\SLisloc}{{\sf isloc?}}
\newcommand{\SLfalse}{{\sf false}}
\newcommand{\SL}{{\it SL}}

Matching logic has inherent support for structural separation, without a
need for any special logic constructs or extensions.
Indeed, pattern matching has a spatial meaning by its very nature:
matching a subterm already separates that subterm from the rest of the
context, so matching two or more terms can only happen when there is no
overlapping between them.
Moreover, matching logic patterns can combine structure with logical
constraints, which allows not only to define very sophisticated patterns,
but also to reason about patterns as if they were logical formulae,
and to achieve modularity by globalizing local reasoning.
Finally, since matching logic allows variables of any sorts, including
of sort $\Map$ when heaps are concerned, it has inherent support
for heap framing and local reasoning, too.

\subsection{Separation Logic Basics}

Separation logic (originating with ideas in
\cite{OHearn99thelogic,sep-logic-csl01}, followed by canonical work in
\cite{reynolds-02}, with more recent developments in
\cite{Pek2014,Brotherston2014,Chu2015,Brotherston2016,Krebbers2017} and
with several provers supporting it in
\cite{Berdine2005,Appel2007,Botincan2009,Perez2011,Berdine2011,NavarroPerez2013,Parkinson2011,Piskac:2013:ASL:2526861.2526927,Pek2014}),
is a logic specifically crafted for reasoning about
heap structures.
There are many variants, but here we only consider
the original variant in~\cite{sep-logic-csl01,reynolds-02}.
Moreover, here we only discuss separation logic as an assertion-language,
used for specifying state properties, and not its extension
as an axiomatic programming language semantic framework.
We regard the latter as an orthogonal aspect, which can similarly be
approached using matching logic.
%In fact, thanks to its generality to include any structure in its patterns,
%including programs themselves, matching logic allows us to avoid the
%axiomatic approach to program verification all together.
%; this is beyond
%our scope in this section, though, and is discussed in some detail in
%Section~\ref{sec:program-reasoning}.

Separation logic extends the syntax of formulae in FOL (Section~\ref{sec:FOL})
as follows:
$$
\begin{array}{rrl}
\varphi & ::= & \textit{(FOL syntax)} \\
& \mid & \SLemp \\
& \mid & \Nat \mapsto \Nat \\
& \mid & \varphi \SLstar \varphi \\
& \mid & \varphi \,\magicwand \varphi \ \ \ \ \ \ \ \ \textrm{``magic wand''}
\end{array}
$$
Its semantics is based on a fixed model of stores and heaps, which
are finite-domain maps from variables and, respectively, locations
(which are particular numbers), to integers.
Below we recall the formal definition of satisfaction in the original
variant of separation logic, noting that subsequent variants of separation
logic extend the underlying model to include stacks (instead of stores)
as well as various types of resources that are encountered in modern
programming languages.
Such extensions are ignored here because they would only complicate the
presentation without changing the overall message: they would only add
more symbols to the corresponding matching logic signature with
appropriate interpretations in the underlying model, and Theorem~\ref{prop:SL}
would still hold.
Nevertheless, we leave the thorough analysis of the diversity of separation
logic variants proposed in the last 15 years through the lenses of matching
logic as a subject for future work.

\newcommand{\PVar}{\Var}

\begin{defi}
\label{dfn:sep-logic}
(Separation logic semantics, adapted from \cite{sep-logic-csl01,reynolds-02})
%Let $\PVar$ be a set of \textbf{program variables}.
Partial finite-domain maps $s: \PVar \rightharpoonup \Nat$
are called \textbf{stores},
partial finite-domain maps $h: \Nat \rightharpoonup \Nat$
are called \textbf{heaps}, and pairs
$(s,h)$ of a store and a heap are called \textbf{states}.
The semantics of the separation logic constructs are given in terms of such
states, as follows:
\begin{itemize}
\item  $(s,h) \models_\SL \varphi$ for a FOL formula $\varphi$ iff  $s \models_\FOL \varphi$
(the heap portion of the model is irrelevant for the FOL fragment);

\item $(s,h) \models_\SL \SLemp$ iff $\Dom(h) = \emptyset$;

\item
$(s,h) \models_\SL e_1 \mapsto e_2$ where $e_1$ and $e_2$ are terms of sort
$\Nat$ (thought of as ``expressions'') iff
$\Dom(h)=\overline{s}(e_1)\neq 0$ and $h(\overline{s}(e_1))=\overline{s}(e_2)$,
where $\overline{s}$ is the (partial function) extension of $s$ to expressions
(with variables) of sort $\Nat$, defined similarly to the extension of
valuations to patterns in Definition~\ref{def:rho-bar};

\item
$(s,h) \models_\SL \varphi_1 \SLstar \varphi_2$ iff there exist $h_1$
and $h_2$ such that $\Dom(h_1) \cap \Dom(h_2)=\emptyset$ and $h = h_1*h_2$
(the merge of $h_1$ and $h_2$, a partial function on maps written as an
associative/commutative comma in Section~\ref{sec:maps})
and $(s,h_1) \models_\SL \varphi_1$, $(s,h_2) \models_\SL \varphi_2$;

\item
$(s,h) \models_\SL \varphi_1 \magicwand \varphi_2$ iff
for any $h_1$ with $\Dom(h_1) \cap \Dom(h) = \emptyset$,
if $(s,h_1) \models_\SL \varphi_1$ then $(s,h*h_1) \models_\SL \varphi_2$;
i.e., the semantics of ``magic wand'' is defined as the states whose heaps
extended with a fragment satisfying $\varphi_1$ result in ones satisfying
$\varphi_2$.
\end{itemize}
Separation logic formula $\varphi$ is \textbf{valid}, written
$\models_\SL\varphi$, iff $(s,h) \models_\SL \varphi$ for any state $(s,h)$.
\end{defi}

\subsection{Map Patterns}
\label{sec:map-patterns}

One of the most appealing aspects of separation logic is that it allows us
to define compact and elegant specifications of heap abstractions using
inductively defined predicates.
Such an abstraction which is quite common is the linked-list abstraction
$\llist(x,S)$ stating that $x$ points to a linked list containing an abstract
sequence of natural numbers $S$:
$$
\begin{array}{rcl}
\llist(x,\epsilon) & \stackrel{\it{def}}{=} & \SLemp \wedge x = 0 \\
\llist(x,n\cdot S) & \stackrel{\it{def}}{=} & \exists z \,.\,x \mapsto [n,z] \SLstar \llist(z,S)
\end{array}
$$
Above, $\epsilon$ is the empty sequence, $n \cdot S$ is the sequence starting
with natural number $n$ and followed by sequence $S$, and $x \mapsto [n,z]$ is
syntactic sugar for $x \mapsto n \SLstar (x+1) \mapsto z$.
So a linked list starting with address $x$ takes either empty heap space,
in which case $x$ must be 0 and the abstracted sequence is $\epsilon$,
or there is some node in the linked list at location $x$ in the heap that
holds the head of the abstracted sequence ($n$) and a link ($z$) to another
linked list that holds the tail of the abstracted sequence ($S$).
The implicit properties of the implicit map model (the heap) in
Definition~\ref{dfn:sep-logic} ensures the well-definedness of this
abstraction, which essentially means that all the locations reached by
unfolding a list abstraction using the inductive properties above are
distinct.
The symbol $\stackrel{\it{def}}{=}$, sometimes written $\equiv$
in the literature, is not part of separation logic; it is
a meta-logical means to define inductive, or recursive predicates, also
encountered in the context of first-order logic: the predicate in question is
interpreted in models as the least-fixed point of its defining
(meta-)equations.

We next show that similar heap patterns can be defined directly in matching
logic.
Specifically, we pick a particular signature (for maps/heaps) together
with desired axioms, that is, a matching logic specification, and show
how additional patterns can be defined in the context of that specification.
The definitions are as compact and elegant as those in separation logic,
and no meta-logical constructs, such as $\stackrel{\it{def}}{=}$ or
$\equiv$, appear to be necessary.
%In the Section~\ref{sec:SLasML} we further show that separation logic can be
%regarded as a matching logic specification.

In what follows, we only discuss maps from natural numbers to natural numbers,
but they can be similarly defined over arbitrary domains as keys and as values.
Consider a matching logic specification of maps like the one shown in
Section~\ref{sec:maps}, but instantiated to natural numbers as both keys and
values, with its axioms explicitly listed, and with a syntax that deliberately
resembles that of separation logic (i.e., we use ``*'' instead of ``,''):
%\vspace*{-1ex}
$$
\begin{array}{l@{\hspace*{10ex}}l}
\_\mapsto\_ : \Nat \times \Nat \rightharpoonup \Map &
  \SLemp \SLstar H = H \\

\SLemp : \ \ra \Map &
H_1 \SLstar H_2 = H_2 \SLstar H_1 \\

\_\SLstar\_ : \Map \times \Map \rightharpoonup \Map &
(H_1 \SLstar H_2) \SLstar H_3 = H_1 \SLstar (H_2 \SLstar H_3) \\

0 \mapsto a = \bot &
x \mapsto a \SLstar x \mapsto b = \bot
\end{array}
%\vspace*{-1ex}
$$
Recall that there are no predicates here, only patterns.
When regarding the above ADT as a matching logic specification, we can
prove that the bottom two pattern equations above are equivalent to
$\neg(0 \mapsto a)$ and, respectively,
$(x \mapsto a \SLstar y \mapsto b) \ra x \neq y$, giving the
$\_\mapsto\_$ and $\_\SLstar\_$ the feel of ``predicates''.
Maps, like natural numbers, do not admit finite
(or even recursively enumberable) equational (or first-order)
axiomatizations, so adding a ``good enough'' subset of equations
is the best we can do in practice.
We chose ones that have been proposed by algebraic specification
languages and by separation logics for several reasons.
First, they have been extensively used, so there is
a good chance they are ``good enough'' for many purposes.
Second, we do not want to imply that we propose a novel axiomatization
of maps; our novelty is the presentation of known specifications of maps
using the general infrastructure of matching logic at no additional
translation cost, without a need to craft a new logic to address the
particularities of maps.
Third, this will ease our presentation in Section~\ref{sec:SLasML}
where the connection with such a logic specifically crafted for maps
is discussed.

Consider the canonical model of partial maps $M$, where:
$M_\Nat \mathrel{=} \{0,1,2,\ldots\}$;
$M_\Map =$ partial maps from natural numbers to natural numbers with
finite domains and undefined in 0, with $\SLemp$ interpreted as the map
undefined everywhere, with
$\_\mapsto\_$ interpreted as the corresponding one-element partial map except
when the first argument is 0 in which case it is undefined (note that
$\_\mapsto\_$ was declared using $\rightharpoonup$), and with $\_\SLstar\_$
interpreted as map merge when the two maps have disjoint domains, or undefined
otherwise (note that $\_\SLstar\_$ was also declared using $\rightharpoonup$).
$M$ satisfies all axioms above.

%Let us now define a basic data-type over maps, the lists.
Following similar definitions in the context of separation logic,
we next define several patterns that are commonly used when proving
properties about programs that can allocate and de-allocate
data-structures in the heap.
We emphasize that our matching logic specifications below look almost
identical, if not identical, to their separation logic variants.
Which is, in fact, the main point we are making in this subsection.
That is, that matching logic allows us to specify the same complex heap
predicates as separation logic, equally compactly and elegantly, but without
a need to devise any new heap-specific logic for that.

We start with matching logic definitions for complete linked lists and for
list fragments.
Let $\llist \in \Sigma_{\Nat,\Map}$ and
$\listf \in \Sigma_{ \Nat \times \Nat, \Map}$
be two symbols together with patterns
$$
\begin{array}{l@{\hspace*{10ex}}l}
\llist(0) = \SLemp &
\listf(x,x) = \SLemp \\

\llist(x) \wedge x\neq 0 = \exists z \,.\,x \mapsto z \SLstar \llist(z) &
\listf(x,y) \wedge x\neq y = \exists z \,.\,x \mapsto z \SLstar \listf(z,y)
\end{array}
%\vspace*{-1ex}
$$
%Note that we purposely said nothing about list fragments of the form
%$\listf(0,y)$ where $y\neq 0$; we will prove these undefined.
Note that $\llist$ and $\listf$ are not meant to be functions,
so we did not use the functional notation (Section~\ref{sec:functions}) for
them.
Moreover, note that $\listf$ is not even meant to be a totally defined relation
(Section~\ref{sec:total-relations}); indeed, $\listf(0,m)$ is $\emptyset$
(and not $\SLemp$) for all $m > 0$.

The main difference between our definitions above and their separation logic
variants is that the latter cannot use the (FOL) equality symbol as we did.
Indeed, $\llist$ and $\listf$ are predicates there, same as equality,
and predicates cannot take predicates as arguments.
To define predicates like $\llist$ and $\listf$, as seen at the beginning of
this section, we have to {\em explicitly} use the meta-logical notation 
$\stackrel{\textrm{def}}{=}$ or $\equiv$ for defining
``recursive predicates'': predicates satisfying desired properties which
have a least fixed-point interpretation in models.
We emphasized ``explicitly'' above to distinguish it from the {\em implicit}
least fixed-point principles used when mathematically defining the semantics
of any logic.
For example, in our context, the extension of $\rho$ to $\overline{\rho}$ in
Definition~\ref{def:rho-bar} uses a least-fixed point construction, similar
to any other logic with terms, but that is orthogonal to how symbols are
interpreted in the given model (symbol interpretation is given by the
model, not by $\rho$).

There are two important questions about the matching logic specification above:
\begin{enumerate}
\item Does this specification admit any solution in $M$, i.e.,
total relations $\llist_M:M_\Nat \ra {\cal P}(M_\Map)$
and $\listf_M:M_\Nat \times M_\Nat \ra {\cal P}(M_\Map)$ satisfying the
patterns above?
\item If yes, is the solution unique?  This is particularly
important because we do not require initiality constraints on
$M$ nor smallest fixed-point constraints on solutions.
\end{enumerate}
We answer these questions positively.
We only discuss $\listf_M$, because the other is similar and simpler.
%Specifically, we show the following:
%(1) any solution can only contain in $\listf_M(n,m)$ when $n\neq m$ maps of the
%form $n \mapsto n_1, n_1 \mapsto n_2, \ldots, n_k \mapsto 0$
%with $k\geq 0$ and $n_0,n_1,\ldots,n_k$ are all different, where we consider
%$n=n_0$; and
%(2) such maps form a solution.
A solution $\listf_M:M_\Nat \times M_\Nat \ra {\cal P}(M_\Map)$ exists
iff it satisfies the two pattern axioms for $\listf$ above; explicitly,
that means that any solution must satisfy:
$$
\begin{array}{@{}l}
\listf_M(n,n)= \{\SLempM\} \mbox{ for all } n \geq 0 \\
\listf_M(0,m) = \emptyset \mbox{ for all } m \neq 0 \\
\listf_M(n,m)=\bigcup\{(\{n \mapstoM n_1\} \SLstarM \listf_M(n_1,m)) \mid n_1 \geq 0 \}
\mbox{ for all }
n \neq 0 \mbox{ and } n\neq m
\end{array}
$$
where $\_\SLstarM\_$ is $M$'s merge function explained above extended to sets of
maps for each argument; recall that the map merge function is undefined
(i.e., it yields an empty set of maps) when the two argument maps are not merge-able.
Note that we had to split the interpretation of the second equation pattern for
$\listf$ into two equalities, reflecting a case analysis on whether the first
argument is 0 or not.
Note also that $\listf(n,m)\neq \emptyset$ when $n\neq 0$, and that
$\listf(n,m)$ contains only non-empty maps when $n\neq 0$ and $n\neq m$.

First, we claim that the following is a solution:
$$
\begin{array}{l}
\listf_M(n,n) \mathrel{=} \{\SLempM\} \mbox{ for all } n \geq 0 \\
\listf_M(0,m) \mathrel{=} \emptyset \mbox{ for all } m \neq 0 \\
\listf_M(n,m) \mathrel{=} \{\ n \mapstoM n_1 \SLstarM n_1 \mapstoM n_2 \SLstarM \cdots \SLstarM n_{k-1} \mapstoM m \\
  \hspace*{9ex} \mid k > 0, \mbox{ and } n_0=n,n_1,n_2,\ldots,n_{k-1} > 0 \mbox{ all different and different from $m$}\}
\end{array}
$$
Indeed, the first two equalities that need to be satisfied by any solution
vacuously hold, while for the third all we need to note is that the
``junk'' maps where $n$ is 0 or in the domain of a map in
$\listf_M(n_1,m)$ are simply discarded by the map merge interpretation of
$\_\SLstar\_$.

Second, we claim that the above is the unique solution.
Let $\listf_M:M_\Nat \times M_\Nat \ra {\cal P}(M_\Map)$ be some
solution satisfying the three equality constraints.
It suffices to prove, by induction on the size $k$ of the domain of
$h\in M_\Map$ that:
$h \in \listf_M(n,m)$ for $n,m\in M_\Nat$ iff either $n=m$ and
$h=\SLempM$ (i.e., $k=0$), or otherwise $n\neq 0$ and $n\neq m$ and $k>0$ and
there are distinct $n_0 = n$, $n_1$,  \ldots, $n_{k-1}$ distinct from $m$ such that
$h = (n\mapstoM n_1 \SLstarM n_1 \mapstoM n_2 \SLstarM \cdots \SLstarM n_{k-1} \mapstoM m)$.
Since the maps in $\listf_M(n,m)$ when $n\neq 0$ and $n\neq m$ contain at least
one binding, we conclude $k=0$ can only happen iff $h\in \listf_M(n,n)$,
and then $h=\SLempM$.
Now suppose $k>0$, which can only happen iff $h \in \listf_M(n,m)$ for
$n\neq 0$ and $n\neq m$, which can only happen iff $n\neq 0$ and $n\neq m$ and
$h=n\mapstoM n_1 \SLstarM h_1$ for some $n_1\geq 0$ and
$h_1\in \listf_M(n_1,m)$.
It all follows now by the induction hypothesis applied to $h_1$.

It should be clear that patterns can be specified in many different ways.
E.g., the first list pattern can also be specified with a single pattern:
$$
\llist(x) = (x = 0 \wedge \SLemp \vee \exists z \,.\,x \mapsto z \SLstar \llist(z))
$$
We can similarly define more complex patterns, such as lists with data.
But first, we specify a convenient operation for defining maps over contiguous
keys, making use of a sequence data-type.
The latter can be defined like in Section~\ref{sec:collections};
for notational convenience, we take the freedom to use comma ``$,$'' instead of
``$\cdot$'' for sequence concatenation in some places:
$$
\begin{array}{l@{\hspace*{15ex}}l}
\_\mapsto[\_] : \Nat \times \Seq \ra \Map &
x \mapsto [\epsilon] = \SLemp \\
& x \mapsto [a, S] = x \mapsto a \SLstar (x+1) \mapsto [S]
\end{array}
$$
In our model $M$, we can take $M_\Seq$ to be the finite sequences of
natural numbers, with $\epsilon_M$ and
$\_\cdot_M\_$ interpreted as the empty sequence and,
respectively, sequence concatenation.

We can now define lists with data as follows:
%\vspace*{-1ex}
$$
\begin{array}{@{}l@{\hspace*{7.7ex}}l@{}}
\llist \in \Sigma_{\Nat \times \Seq, \Map} &
\listf \in \Sigma_{\Nat \times \Seq \times \Nat, \Map} \\

\llist(x,\epsilon) = (\SLemp \wedge x = 0) &
\listf(x,\epsilon,y) = (\SLemp \wedge x = y) \\

\llist(x,n\cdot S) = \exists z \,.\,x \mapsto [n,z] \SLstar \llist(z,S) &
\listf(x, n\cdot S,y) =  \exists z \,.\,x \mapsto [n,z] \SLstar \listf(z,S,y)
\end{array}
%\vspace*{-1ex}
$$
Note that, unlike in the case of lists without data, this time we
have not required the side conditions $x\neq 0$ and $x\neq y$, respectively.
The side conditions were needed in the former case because without them
we can infer, e.g., $\llist(0)=\bot$ (from the second equation of $\llist$),
which using the first equation would imply $\SLemp = \bot$.
However, they are not needed in the latter case because it is safe (and
even desired) to infer $\llist(0,n\cdot S) = \bot$ for any $n$ and $S$.
We can show, using a similar approach like for lists without data, that
the pattern $\listf(x,S,y)$ matches in $M$ precisely the lists
starting with $x$, exiting to $y$, and holding data sequence $S$.

We can similarly define other data-type specifications, such as
trees with data:
$$
\begin{array}{l}
\emptytree : \ \ra \Tree \\
\node : \Nat \times \Tree \times \Tree \ra \Tree \\
\tree \in \Sigma_{\Nat \times \Tree, \Map} \\[2ex]
\tree(0,\emptytree) = \SLemp \\
\tree(x,\node(n,\!t_1\!,\!t_2)) = \exists y\,z\,.\,x \!\mapsto\![n,y,z]\SLstar\tree(y,t_1)\SLstar\tree(z,t_2))
\end{array}
$$

Therefore, in the model $M$ of partial maps described above,
there is a unique way to interpret $\llist$ and $\listf$, namely
as the expected linked lists and, respectively, linked list fragments.
Fixing the interpretations of the basic mathematical domains, such
as those of natural numbers, sequences, maps, etc., suffices in order to
define interesting map patterns that appear in verification of heap properties
of programs, in the sense that the axioms themselves uniquely define the
desired data-types.
No logic extensions (such as adding free models with induction/recursion
principles as a matching logic equivalent to ``recursive predicates'',
or least fixed-point constraints, or even fixed points of any kind)
were needed to define them.
The defining axioms were precise enough to capture the intended concept
in the intended model.
Choosing the right basic mathematical domains is, however, crucial.
For example, if we allow the maps in $M_\Map$ to have infinite domains then
the list patterns without data above (the first ones) also include infinite
lists.
The lists with data cannot include infinite lists, because we only allow finite
sequences.
This would, of course, change if we allow infinite sequences, or streams,
in the model.
In that case, $\llist$ and $\listf$ would not admit unique interpretations
anymore, because we can interpret them to be either all the finite domain
lists, or both the finite and the infinite-domain lists.
Writing patterns which admit the desired solution in the desired model suffices
in practice; our reasoning techniques developed in the sequel allow us to
derive properties that hold in all models satisfying the axioms, so any derived
property is sound also for the intended model and interpretations.

\subsection{Separation Logic as an Instance of Matching Logic}
\label{sec:SLasML}

Consider the FOL fragment in Section~\ref{sec:FOL},
where the signature $\Sigma$ includes the signature of maps in
Section~\ref{sec:map-patterns}.
Any additional FOL constructs, background theories, and/or built-in domains
that one wants to consider in separation logic specifications, are handled
as already explained in Sections~\ref{sec:FOL} and ~\ref{sec:builtins}.
%, including
%all the notations and derived constructs (syntactic sugar) introduced there.
It is clear then that all the syntactic constructs of separation logic,
except for the magic wand, $\magicwand$, are given by the above matching
logic signature.
The magic wand, on the other hand, can be defined as the following
derived construct:
$$
\begin{array}{rcl}
\varphi_1 \magicwand \varphi_2 & \equiv & \exists H\!:\!\Map \,.\,H \wedge \lfloor H * \varphi_1 \ra \varphi_2\rfloor
\end{array}
$$
Recall from Section~\ref{sec:definedness} that $\lfloor\varphi\rfloor$ is $\top$
(it matches the entire set) iff its enclosed pattern $\varphi$ is $\top$;
otherwise, if $\varphi$ does not match some elements, then
$\lfloor\varphi\rfloor$ is $\bot$ (it matches nothing).
In words, $\varphi_1 \magicwand \varphi_2$ matches all maps $h$ which merged with
maps matching $\varphi_1$ yield only maps matching $\varphi_2$.
Thanks to the notational convention that Booleans $b$, respectively usual
predicates $p$, stand for equalities $b={\it true}$, respectively $p = \top_\Pred$
(Notation~\ref{notation:bool}),
\begin{quote}
\em 
Any separation logic formula is a matching logic pattern of sort $\Map$.
\end{quote} 

%We next construct our model.
%Let $M$ be the model for maps in Section~\ref{sec:maps}, where we
%define the partial function
%$\_\magicwand_{\!M}\_:\Map \times \Map \ra {\cal P}(\Map)$
%as follows:
%$$
%h_1 \magicwand_{\!M}\ h_2 = \{ h \mid \Dom(h) \cap \Dom(h_1) = \emptyset
%\mbox{ and } h *_M h_1 = h_2\}
%$$
%Note that $h_1 \magicwand_{\!M}\ h_2$ is either the empty set or it is a set of
%precisely one map.

Semantically, note that separation logic hard-wires a particular model of maps.
That is, its satisfaction relation $\models_\SL\varphi$ is defined using a
pre-defined universe of maps, which is conceptually the same as our model of
maps in Section~\ref{sec:map-patterns}.
Since separation logic extends FOL, its models are still allowed to
instantiate the FOL part of its syntax in the usual FOL way,
but the maps are rigid and the models cannot touch them.
It is therefore not surprising that we also have to fix the maps in our
matching logic models corresponding to the syntax described so far in order
to faithfully capture separation logic semantically.
For the rest of this section, we only consider models $M$ for the matching
logic specification above whose reduct to the syntax of maps is
precisely the map domain in Section~\ref{sec:map-patterns}.
We let $\Map \models \varphi$ denote the resulting satisfaction relation: $\Map \models \varphi$ iff
$M \models \varphi$ for any model $M$ like above.

In separation logic formulae, variables range only over the domains of data
(e.g., natural numbers), but not over heaps/maps;
for example, ``$\exists H\!:\!\Map \, .\, 1 \mapsto 2 \SLstar H$'' is not
a proper separation logic formula (although it is one in matching logic).
Also, stores $s$ are mappings of variables to particular values.
In matching logic, variables range over all syntactic categories, including
over heaps in our case, and valuations $\rho$ can map such variables to any
values in the model;
for example, the variable $H$ of sort $\Map$ in the pattern defining $\magicwand$
above is a heap variable.
Since separation logic formulae $\varphi$ contain no heap variables, we can
regard separation logic stores $s$ as $M$-valuations with the property that
$\overline{s}(\varphi)$ contains precisely the heaps which together with
$s$ satisfy the original separation logic formula $\varphi$.
We prove this in the next proposition showing that separation
logic is not only syntactically an instance of matching logic (modulo notations
in Section~\ref{sec:useful-symbols}), but also semantically:

\begin{prop}
\label{prop:SL}
If $\varphi$ is a separation logic formula, then
$\models_\SL \varphi \mbox{ iff } \Map \models \varphi$.
\end{prop}
\begin{proof}
Like in the proofs of Propositions~\ref{prop:pure-predicate},
\ref{prop:alg-spec}, and \ref{prop:FOL}, there is a bijection between the
models of separation logic and the matching logic $\Map$-models.
The bijection works as described in the aforementioned propositions
for the FOL part, and adds the map model in Section~\ref{sec:map-patterns}
to the resulting matching logic models.
To keep the notation simple, we use the same name, $M$, to refer both to a
separation logic model and to its corresponding matching logic model,
remembering from the proofs of Propositions~\ref{prop:pure-predicate} and
\ref{prop:alg-spec} that the latter's carrier of sort $\Pred$ is a singleton
$\{\star\}$.
The context makes it clear which one we are referring to.
%Let $M$ be such a model.

We show by structural induction on the separation logic formula $\varphi$
the more general result that for any store $s$ and any heap $h$, we have
$(s,h) \models_\SL \varphi$ iff $h \in \overline{s}(\varphi)$.

If $\varphi$ is a FOL formula then its desugared matching logic
correspondent is $\varphi =_\Pred^\Map \top_\Pred$
(Notation~\ref{notation:pred}).
Then $(s,h) \models_\SL \varphi$ iff
$s \models_\FOL \varphi$ (Definition~\ref{dfn:sep-logic}), iff
$\overline{s}(\varphi)=\{\star\}$
(see proof of Proposition~\ref{prop:FOL}), iff
$\overline{s}(\varphi) = \overline{s}(\top_\Pred)$, iff
$\overline{s}(\varphi =_\Pred^\Map \top_\Pred) = M_\Map$
(by Proposition~\ref{prop:equality}),
iff $h\in\overline{s}(\varphi =_\Pred^\Map \top_\Pred)$
(Proposition~\ref{prop:equality} again: equality is
interpreted as either $M_\Map$ or $\emptyset$).

Conjunction and negation are trivial.
Existential quantification:
$(s,h) \models_\SL \exists x\,.\,\varphi$
iff there exists some $a\in M$ of appropriate (non-heap) sort
such that $(s[a/x],h) \models \varphi$, iff
$h\in\overline{s[a/x]}(\varphi)$ (induction hypothesis), iff
$h\in\bigcup\{\overline{s'}(\varphi) \mid
s':\Var\rightarrow M,\ 
s'\!\!\upharpoonright_{\Var\backslash\{x\}} =
s\!\!\upharpoonright_{\Var\backslash\{x\}}
\}$,
iff
$h\in\overline{s}(\exists x\,.\,\varphi)$.
We next discuss the heap-related constructs of separation logic.

If $\varphi \equiv \SLemp$ then $(s,h)\models_\SL\SLemp$
iff $h=\SLempM$, iff $h\in\{\SLempM\}$, iff $h\in\overline{s}(\SLemp)$.

If $\varphi \equiv e_1 \mapsto e_2$ then $(s,h) \models_\SL e_1\mapsto e_2$ iff 
$\Dom(h)=\overline{s}(e_1)\neq 0$ and
$h(\overline{s}(e_1))=\overline{s}(e_2)$ (Definition~\ref{dfn:sep-logic}),
iff $h$ is the partial-domain map
$\overline{s}(e_1)\mapsto_M\overline{s}(e_2)$
(which is undefined when $\overline{s}(e_1)=0$---see Section~\ref{sec:map-patterns}),
iff $h\in\overline{s}(e_1 \mapsto e_2)$.

If $\varphi \equiv \varphi_1 \SLstar \varphi_2$ then
$(s,h) \models_\SL \varphi_1 \SLstar \varphi_2$ iff there exist $h_1$
and $h_2$ of disjoint domains such that $h = h_1 *_M h_2$
(the merge of $h_1$ and $h_2$, which is a partial function on
maps---see Definition~\ref{dfn:sep-logic} and Section~\ref{sec:map-patterns})
and $(s,h_1) \models_\SL \varphi_1$ and $(s,h_2) \models_\SL \varphi_2$, iff
$h = h_1 *_M h_2$ and
$h_1\in \overline{s}(\varphi_1)$
and
$h_2\in \overline{s}(\varphi_2)$
(induction hypothesis),
iff
$h\in \overline{s}(\varphi_1) \SLstar_M \overline{s}(\varphi_2)$, iff
$h\in \overline{s}(\varphi_1 \SLstar \varphi_2)$.

The only case left is the ``magic wand'',
$\varphi \equiv \varphi_1 \magicwand \varphi_2$:
$$
\begin{array}{ll}
& h \in \overline{s}(\varphi_1 \magicwand \varphi_2) \\
\mbox{iff} & h \in \overline{s}(\exists H\,.\,H\wedge \lfloor H*\varphi_1 \ra \varphi_2 \rfloor) \\
\mbox{iff} & \{h\} \SLstar_M \overline{s}(\varphi_1) \subseteq \overline{s}(\varphi_2) \\
\mbox{iff} & h*h_1 \in \overline{s}(\varphi_2) \mbox{ for any } h_1\in\overline{s}(\varphi_1)
  \mbox{ compatible with }h \\
\mbox{iff} & (s,h*h_1) \models_\SL\varphi_2 \mbox{ for any } h_1 \mbox{ compatible with $h$ such that }
(s,h_1) \models_\SL \varphi_1 \\
& \mbox{(previous step followed by the induction hypothesis)} \\
\mbox{iff} & (s,h) \models_\SL \varphi_1 \magicwand \varphi_2
\end{array}
$$
The proof is complete.
\end{proof}

Although matching logic is complete (Section~\ref{sec:deduction}), so
its validity $\models$ is semi-decidable, while results in
\cite{Calcagno2001,joel-2014-fossacs} state that the validity problem in separation
logics is not semi-decidable, note that there is no conflict here
because we restricted matching logic validity to $\Map$-models.
As an analogy, it is well-known that the validity of predicate logic formulae
can be arbitrarily hard when particular (and not all) models are concerned.
All the above says is that the results in \cite{Calcagno2001,joel-2014-fossacs} carry
over to the particular matching logic theory restricted to $\Map$-models
discussed in this section.
Most likely one can obtain even more general instances of the results
\cite{Calcagno2001,joel-2014-fossacs} for matching logic, but that is beyond the
scope of this paper.

The loose-model approach of matching logic is in sharp
technical, but not conceptual, contrast to separation logic.
In separation logic, the syntax of maps and separation constructs is part of
the syntax of the logic itself, and the model of maps is intrinsically integrated
within the semantics of the logic: its satisfaction relation is defined in
terms of a fixed syntax and the fixed model of the basic domains
(maps, sequences, etc.).
Then specialized proof rules and theorem provers need to be devised.
If any changes to the syntax or semantics are desired, for example adding a
new stack, or an I/O buffer, etc., then a new logic is obtained.
Proof rules and theorem provers may need to change as the logic changes.
In matching logic, the basic ingredients of separation logic form one
particular specification, with its particular signature and pattern axioms,
together with particular but carefully chosen models.
This enables us to use generic first-order reasoning in matching logic
(Section~\ref{sec:deduction}), as well as theorem provers or SMT solvers
for reasoning about the intended models.
Nevertheless, several high performance automated provers for separation logics
have been developed, e.g.
\cite{Berdine2005,Appel2007,Botincan2009,Perez2011,Berdine2011,NavarroPerez2013,Parkinson2011,Piskac:2013:ASL:2526861.2526927,Pek2014},
while there are no automated provers available for matching logic yet.
A technical challenge, left for future work, is to investigate the techniques
and algorithms underlying the existing separation logic provers and to generalize
them if possible to work with matching logic in general or at least with common
fragments of it.

%The main drawback of loose approaches to models is that
%in order to do certain proofs one may need to restrict the class of models
%by adding new axioms.
%For example, if one needs maps to be either empty or otherwise contain
%at least one binding, then in pattern logic we have to add one more pattern
%axiom, e.g.,
%$
%\SLemp \vee \exists x \,.\,\exists a\,.\,\exists h\,.\, x\mapsto a \SLstar h
%$,
%which can possibly be proved off-line using a formalization of the intended
%model, while in a fixed model like in separation logics such properties are
%tautologies in the logic.
%%; one does not need to leave the logic in order to
%%prove them.

Like for modal logic (Section~\ref{sec:modal-logic}), the result above in
combination with the reduction of matching logic to predicate logic with
equality in Section~\ref{sec:PL-reduction} yields a translation from
separation logic to predicate logic with equality.
Note that many of the separation logic provers above are implicitly or
explicitly based on translations to FOL, and specific translations to FOL or
fragments of it have been already
studied~\cite{Calcagno2005,10.1109/SKG.2010.29,Bobot2012}.
Like for modal logic (Section~\ref{sec:modal-logic}), our goal in this section
was not to propose yet another translation.
Our goal was to show how separation logic can be framed as a matching logic
specification both syntactically and semantically, without any translation
(but only with syntactic sugar notations).
Such results can help us better understand both logics, as well as their
strengths and limitations.

\section{Matching Logic: Reduction to Predicate Logic with Equality}
\label{sec:PL-reduction}

It is known that FOL formulae can be translated into equivalent
predicate logic with equality formulae (i.e., no function or constant
symbols---see Section~\ref{sec:predicate-logic}), by replacing all
functions with their graph relations (see, e.g., \cite{nla.cat-vn2062435}).
Specifically, function symbols $\sigma:s_1\times\cdots\times s_n \ra s$
are replaced with predicate symbols
$\pi_\sigma : s_1 \times \cdots \times s_n \times s$,
and then terms are transformed into formulae by adding existential
quantifiers for subterms.
Let us define such a translation, say $\PL$.
It takes each FOL predicate
$\pi(t_1,\ldots,t_n)$ into a pure predicate logic formula
as follows:
$$\PL(\pi(t_1,\ldots,t_n)) = \exists r_1 \cdots r_n\,.\,\PL_2(t_1,r_1)\wedge\cdots\wedge\PL_2(t_n,r_n)\wedge\pi(r_1,\ldots,r_n)$$
where $\PL_2(t,r)$ is a translation of term $t$ into a predicate stating
that $t$ equals variable $r$:
$$
\begin{array}{rcl}
\PL_2(x,r) & = & (x=r) \\
\PL_2(\sigma(t_1,\ldots,t_n),r) & = & 
  \exists r_1 \cdots \exists r_n\,.\,\PL_2(t_1,r_1) \wedge
 \cdots \wedge \PL_2(t_n,r_n) \wedge \pi_\sigma(r_1,\ldots,r_n,r)
\end{array}
$$
Axioms stating that the predicate symbols
$\pi_\sigma : s_1 \times \cdots \times s_n \times s$ associated to function symbols
$\sigma:s_1\times\cdots\times s_n \ra s$
are interpreted as total function relations are also needed:
$$
\forall x_1:s_1,\ldots,x_n:s_n\,.\,\exists y:s\,.\,\forall z:s\,.\,
(\pi_\sigma(x_1,\ldots,x_n,z)
\lra
y=z)
$$

We can similarly translate matching logic patterns into equivalent predicate
logic formulae.
Consider predicate logic with equality (and no function or constant symbols)
whose satisfaction relation is $\PLmodels$.
For matching logic signature $(S,\Sigma)$, let $(S,\Pi_\Sigma)$ be the
predicate logic signature with
$\Pi_\Sigma=\{\pi_\sigma : s_1 \times \cdots \times s_n \times s \mid
\sigma \in \Sigma_{s_1\ldots s_n,s}\}$, like above but without the axioms stating that these predicates have a functional interpretation in models (because the matching logic symbols need not be interpreted as functions).
We define the translation $\PL$ of matching logic $(S,\Sigma)$-patterns
into predicate logic $(S,\Pi_\Sigma)$-formulae inductively:
$$
\begin{array}{@{}r@{\ }c@{\ }l}
\PL(\varphi) & = & \forall r\,.\,\PL_2(\varphi,r) \\[1.5ex]
\PL_2(x,r) & = & (x = r) \\
\PL_2(\sigma(\varphi_1,\ldots,\varphi_n),r) & = & 
  \exists r_1 \cdots \exists r_n\,.\,\PL_2(\varphi_1,r_1) \wedge
 \cdots \wedge \PL_2(\varphi_n,r_n) \wedge \pi_\sigma(r_1,\ldots,r_n,r) \\
\PL_2(\neg \varphi,r) & = & \neg \PL_2(\varphi,r) \\
\PL_2(\varphi_1 \wedge \varphi_2,r) & = & \PL_2(\varphi_1,r) \wedge \PL_2(\varphi_2,r) \\
\PL_2(\exists x\,.\,\varphi,r) & = & \exists x\,.\,\PL_2(\varphi,r) \\[1.5ex]
\PL(\{\varphi_1,\ldots,\varphi_n\}) & = &
  \{\PL(\varphi_1),\ldots,\PL(\varphi_n)\}
\end{array}
$$
The predicate logic formula $\PL_2(\varphi,r)$ captures the
intuition that ``$r$ matches $\varphi$''.
The top transformation above, $\PL(\varphi)=\forall r\,.\,\PL_2(\varphi,r)$,
is different from (and simpler than) the corresponding translation
of predicates from FOL to predicate logic.
It captures the intuition that the pattern $\varphi$ is valid iff it is matched
by {\em all} values $r$.
Then the following result holds:
\begin{prop}
\label{prop:mlTopred}
If $F$ is a set of patterns and $\varphi$ is a pattern,
$F \models \varphi$ iff $\PL(F) \PLmodels \PL(\varphi)$.
\end{prop}
\begin{proof}
It suffices to show that there is a bijective correspondence
between matching logic $(S,\Sigma)$-models $M$ and
predicate logic $(S,\Pi_\Sigma)$-models ${M}'$, such that
${M} \models \varphi$ iff ${M}' \PLmodels \PL(\varphi)$
for any $(S,\Sigma)$-pattern $\varphi$.
The bijection is defined as follows:
\begin{itemize}
\item $M'_s=M_s$ for each sort $s \in S$;
\item ${\pi_\sigma}_{M'} \subseteq M_{s_1}\times \cdots \times M_{s_n} \times M_s$ with
$(a_1,\ldots,a_n,a) \in {\pi_\sigma}_{M'}$ iff 
$\sigma_M:M_{s_1}\times \cdots \times M_{s_n} \ra {\cal P}(M_s)$ with
$a\in\sigma_M(a_1,\ldots,a_n)$.
\end{itemize}
To show ${M} \models \varphi$ iff ${M}' \PLmodels \PL(\varphi)$,
it suffices to show $a \in \overline{\rho}(\varphi)$ iff
$\rho[a/r] \PLmodels \PL_2(\varphi,r)$ for any $\rho:\Var\ra M$,
which follows easily by structural induction on $\varphi$.
\end{proof}

It is informative to translate the definedness and equality patterns in
Sections~\ref{sec:definedness} and Section~\ref{sec:equality} using the above,
and especially to sanity check that the equality pattern of matching logic
indeed translates to the equality predicate of predicate logic with equality.
Recall that the definedness symbols were axiomatized with pattern axioms
$\lceil x \rceil$, and that we assumed them always available
(Assumption~\ref{assumption:definedness}).
Then $\PL(\lceil x \rceil)$ is
$\forall r \,.\, \pi_{\lceil\_\rceil}(x,r)$.
We can drop the universal quantifier and therefore assume
$\pi_{\lceil\_\rceil}(x,r)$ as an axiom formula in the translated predicate
logic specification.
Let us now show that the matching logic equality $x=y$, which is syntactic
sugar for $\neg \lceil \neg(x \lra y) \rceil$, translates to the equality
$x=y$ in predicate logic.
Applying the translation above, we get
$\PL(x=y)$ is $\forall r\,.\,\neg(\exists r_1\,.\,\neg(x=r_1 \lra y=r_1)\wedge \pi_{\lceil\_\rceil}(r_1,r))$,
which is equivalent, in predicate logic with equality, to
$\forall r\,.\,\forall r_1\,.\,(x=r_1 \lra y=r_1)$,
which is further equivalent to $x=y$.
Similarly, we can show that the translation of the equational
pattern stating that $\sigma$ is functional,
namely $\forall x_1\ldots x_n\,.\,\exists y\,.\,\sigma(x_1,\ldots,x_n)=y$,
indeed corresponds to the predicate logic formula
$
\forall x_1,\ldots,x_n\,.\,\exists y\,.\,\forall z\,.\,
(\pi_\sigma(x_1,\ldots,x_n,z)
\lra
y=z)
$, as expected.
We leave this as an exercise to the interested reader.

Proposition~\ref{prop:mlTopred} gives us a sound and complete procedure for
matching logic reasoning: translate the specification $(S,\Sigma,F)$ and
pattern to prove $\varphi$ into the predicate logic
specification $(S,\Pi_\Sigma,\PL(F))$ and formula $\PL(\varphi)$,
respectively, and then derive it using the sound and complete proof
system of predicate logic.
However, translating patterns to predicate logic formulae makes reasoning
harder not only for humans, but also for computers, since new quantifiers
are introduced.
For example,
$$(1 \mapsto 5 \SLstar 2 \mapsto 0 \SLstar 7 \mapsto 9 \SLstar 8 \mapsto 1)
\ra \llist(7,9\cdot 5)$$
discussed and proved in a few steps in Section~\ref{sec:deduction},
translates into the following formula (to keep it small, we do not
translate the numbers), which takes dozens, if not hundreds of steps
to prove using the predicate logic proof system:
$$
\begin{array}{r}
\forall r\,.\,
(
\exists r_1\,.\,\exists r_2\,.\,
\pi_{\mapsto}(1,5,r_1)
\mathrel\wedge
(
\exists r_3\,.\,\exists r_4\,.\,
\pi_{\mapsto}(2,0,r_3)
\mathrel\wedge
(
\exists r_5\,.\,\exists r_6\,.\,
\pi_{\mapsto}(7,9,r_5)
\mathrel\wedge
\pi_{\mapsto}(8,1,r_6)
\\
\mathrel\wedge \pi_\SLstar(r_5,r_6,r_4)
)
\mathrel\wedge \pi_\SLstar(r_3,r_4,r_2)
)
\mathrel\wedge \pi_\SLstar(r_1,r_2,r)
) \rightarrow
\exists r_7\,.\,\pi_\cdot(9,5,r_7) \mathrel\wedge \pi_\llist(7,r_7,r)
)
\end{array}
$$
What we would like is to reason directly with matching logic patterns, the same
way we reason directly with terms in FOL without translating them to
predicate logic.

\section{Matching Logic: Sound and Complete Deduction}
\label{sec:deduction}

In Figure~\ref{fig:proof-system}, we propose a sound and complete
proof system for matching logic (under Assumption~\ref{assumption:definedness}).
The first group of rules/axioms are those of FOL with equality, discussed
and proved sound in Section~\ref{sec:matching-logic}
(predicate logic: Proposition~\ref{prop:PL-validity}),
Section~\ref{sec:equality}
(equational: Proposition~\ref{prop:equality}),
and Section~\ref{sec:functions}
(FOL Substitution, called Term Substitution there:
Corollary~\ref{cor:FOL-substitution}),
with a slightly generalized Substitution axiom that we call
Functional Substitution (discussed below), which requires another
axiom (shown sound by Corollary~\ref{cor:variables}), called
Functional Variable, stating that variables are functional.
The second group of rules/axioms are about membership and were proved
sound in Section~\ref{sec:membership} (Proposition~\ref{prop:membership}).

Substitution must be used with care.
Recall FOL's Substitution: $(\forall x\,.\,\varphi) \ra \varphi[t/x]$.
Since matching logic makes no syntactic distinction between terms and
predicates, we would like to have a proof system that does not make such a
distinction either.
Ideally, since terms and predicates are particular patterns,
we would like to take the proof system of FOL with equality and turn it into a
proof system for matching logic by simply replacing ``predicate'' and ``term''
with ``pattern''.
This actually worked for all the rules and axioms, except for Substitution:
$(\forall x\,.\,\varphi) \ra \varphi[t/x]$.
Unfortunately, Substitution is not sound if we replace $t$ with any pattern.
For example, let $\varphi$ be $\exists y\,.\,x=y$ (Corollary~\ref{cor:variables}).
If FOL's Substitution were sound for arbitrary patterns $\varphi'$
instead of $t$, then the formula $\exists y\,.\,\varphi'=y$, stating that
$\varphi'$ is a functional pattern (i.e., it evaluates to a unique element
for any valuation: Definition~\ref{dfn:functional}), would be valid for any
pattern $\varphi'$.
That is, any pattern would be functional, which is neither true nor
desired
(e.g., $\top$ evaluates to the total set, $\bot$ to the empty set, etc.).

Nevertheless, as proved in Corollary~\ref{cor:FOL-substitution}, Substitution
stays sound if $t$ is a term pattern
(Definition~\ref{dfn:functional-notation}), that is, a pattern build
with only functional symbols (interpreted as functions in all models) and no
other constructs:
$\models (\forall x\,.\,\varphi) \ra \varphi[t/x]$ holds if $\varphi$ is any
pattern but $t$ is a term pattern.
It turns out that the fact that $t$ is built with only functional symbols is
irrelevant, and all that matters is that $t$ is a functional pattern
(all term patterns are functional: Corollary~\ref{cor:terms}).
We therefore generalize the Term Substitution axiom:
\begin{quote}
Functional Substitution: $\vdash (\forall x\,.\,\varphi) \wedge
  (\exists y\,.\,\varphi'=y) \ra \varphi[\varphi'/x]$
\end{quote}
This is more general than the original Substitution in FOL
(which allowed only predicates for $\varphi$) and
than Term Substitution (Corollary~\ref{cor:FOL-substitution}): it can also
apply when $\varphi'$
is not a term pattern but can be proved to be functional.
It is interesting to note that a similar modification of Substitution was needed in
the context of {\em partial} FOL (PFOL) \cite{DBLP:journals/sLogica/FarmerG00}, where the
interpretations of functional symbols are partial functions, so terms may be
undefined; axiom PFOL5 in \cite{DBLP:journals/sLogica/FarmerG00} requires
$\varphi'$ to be {\em defined} in the Substitution rule, and several rules for proving
definedness are provided.
Note that our condition $\exists y\,.\,\varphi'=y$ is equivalent to definedness
in the special case of PFOL, and that, thanks to the definability of equality in
matching logic, we do not need any special axiomatic or rule support for proving
definedness.

\begin{figure}
\underline{FOL axioms and rules:}
\vspace*{-2ex}
\begin{quote}
\ 

\mbox{\phantom01.} $\vdash$ propositional tautologies

\mbox{\phantom02. Modus Ponens: $\vdash \varphi_1$ and $\vdash \varphi_1 \ra \varphi_2$
imply $\vdash\varphi_2$}

\mbox{\phantom03. $\vdash (\forall x\,.\,\varphi_1\ra\varphi_2) \ra (\varphi_1 \ra \forall x\,.\,\varphi_2)$
when $x\not\in\FV(\varphi_1)$}

\mbox{\phantom04.} Universal Generalization: $\vdash\varphi$ implies $\vdash \forall x\,.\,\varphi$

\mbox{\phantom05.} Functional Substitution: $\vdash (\forall x\,.\,\varphi) \wedge
  (\exists y\,.\,\varphi'=y) \ra \varphi[\varphi'/x]$

\mbox{\phantom05'.} Functional Variable: $\vdash \exists\,y.\,x=y$

\mbox{\phantom06.} Equality Introduction: $\vdash \varphi = \varphi$

\mbox{\phantom07. Equality Elimination: $\vdash \varphi_1=\varphi_2 \mathrel\wedge \varphi[\varphi_1/x] \ra \varphi[\varphi_2/x]$}
\end{quote}\medskip
\underline{Membership axioms and rules:}
\vspace*{-2ex}
\begin{quote}
\ 

\mbox{\phantom08.} $\vdash \forall x\,.\,x\in\varphi$ iff $\vdash \varphi$
%$\vdash \varphi = \exists x\,.\, x\wedge(x\in\varphi)$

\mbox{\phantom09.} $\vdash x \in y = (x=y)$ when $x,y\in\Var$

10. $\vdash x\in\neg\varphi = \neg(x\in\varphi)$

11. $\vdash x\in\varphi_1\wedge\varphi_2 = (x\in\varphi_1) \wedge (x\in\varphi_2)$

12. $\vdash (x\in\exists y . \varphi) = \exists y.(x\in\varphi)$, with $x$ and $y$ distinct

13. \mbox{$\vdash x\!\in\!\sigma(\varphi_1,\!..,\varphi_{i-1},\varphi_i,\varphi_{i+1},\!..,\varphi_n) = \exists y . (y\!\in\!\varphi_i \mathrel\wedge x \!\in\!\sigma(\varphi_1,\!..,\varphi_{i-1},y,\varphi_{i+1},\!..,\varphi_n))$}
\end{quote}
%%\vspace*{-2ex}
\caption{Sound and complete proof system of matching logic.}
\label{fig:proof-system}
%\vspace*{-3ex}
\end{figure}

We have made no effort to minimize the number of rules and axioms in our
proof system in Figure~\ref{fig:proof-system}.
On the contrary, our approach was to include all the rules and axioms that
turned out to be useful in proof derivations, especially if they already
existed in FOL.
Moreover, we preferred to frame ``unexpected'' properties of matching
logic as axioms or proof rules, so that users of the proof system are
fully aware of them.
For example, we could have merged the Functional Substitution and Functional
Variable axioms into the conventional predicate logic Substitution
((5) in Proposition~\ref{prop:PL-validity}) or the FOL Term Substitution
(Corollary~\ref{cor:FOL-substitution}), but we preferred not to, because
we want the user of our proof system to be fully aware of the fact that
they cannot substitute arbitrary patterns for variables; they should first
prove that the pattern is functional.
Additionally, our Functional Substitution is more general, in that it applies
in more instances, so proof derivations are shorter.

\begin{prop}
\label{prop:substitutions}
With the proof system in Figure~\ref{fig:proof-system}, the following are derivable:
\begin{enumerate}
\item Predicate Logic Substitution ((5) in Proposition~\ref{prop:PL-validity}): $\vdash (\forall x\,.\,\varphi) \ra \varphi[y/x]$
\item
Term patterns are functional (Corollary~\ref{cor:terms}): $\vdash \exists y\,.\,t=y$ for any term
pattern $t$
\item Term Substitution (Corollary~\ref{cor:FOL-substitution}): $\vdash (\forall x\,.\,\varphi) \ra \varphi[t/x]$
\end{enumerate}
\end{prop}
\begin{proof}
By propositional calculus reasoning, which is subsumed by
our proof system (1. and 2. in Figure~\ref{fig:proof-system}),
for any patterns $A$, $B$, and $C$, if $\vdash A \wedge B \ra C$ and $\vdash B$
then $\vdash A \ra C$.
To prove (1), pick $A$ as $\forall x\,.\,\varphi$, $B$ as $\exists z\,.\,y=z$,
$C$ as $\varphi[y/x]$.
Then $\vdash A \wedge B \ra C$ by Functional Substitution and $\vdash B$ by
Functional Variable, so $\vdash A \ra C$, i.e.,
$\vdash (\forall x\,.\,\varphi) \ra \varphi[y/x]$.

We prove (2) and (3) together, by structural induction on $t$.
If $t$ is a variable then they follow by Functional Variable and,
respectively, by (1).
Suppose that $t$ is $\sigma(t_1,\ldots,t_n)$ for some functional symbol
$\sigma$, i.e., one for which we have an axiom
$\exists y\,.\,\sigma(x_1,\ldots,x_n)=y$
(Definition~\ref{dfn:functional-notation}), and for some appropriate
term patterns $t_1$, \ldots, $t_n$.
By the induction hypothesis of (2), we have
$\vdash \exists y_1\,.\,t_1=y_1$, \ldots, $\vdash \exists y_n\,.\,t_n=y_n$.
By the induction hypothesis on (3) with $x$ as $x_1$ and $\varphi$ as
$\exists y\,.\,\sigma(x_1,\ldots,x_n)=y$, we derive
$$\vdash (\forall x_1\,.\,\exists y\,.\,\sigma(x_1,\ldots,x_n)=y)\ra 
\exists y\,.\,\sigma(t_1,\ldots,x_n)=y$$
Since $\vdash \forall x_1\,.\,\exists y\,.\,\sigma(x_1,\ldots,x_n)=y$
by the functionality axiom of $\sigma$ and Universal Generalization,
we derive
$\vdash \exists y\,.\,\sigma(t_1,x_2,\ldots,x_n)=y$.
By the induction hypothesis on (3) with $x$ as $x_2$ and $\varphi$ as
$\exists y\,.\,\sigma(t_1,x_2,\ldots,x_n)=y$, we derive
$\vdash (\forall x_2\,.\,\exists y\,.\,\sigma(t_1,x_2,\ldots,x_n)=y)\ra 
\exists y\,.\,\sigma(t_1,t_2,\ldots,x_n)=y$.
Since $\vdash \forall x_2\,.\,\exists y\,.\,\sigma(t_1,x_2,\ldots,x_n)=y$
by the previous derivation and Universal Generalization, we derive
$\vdash \exists y\,.\,\sigma(t_1,t_2,x_3,\ldots,x_n)=y$.
Iterating this process for all the arguments of $\sigma$, we eventually
derive $\vdash \exists y\,.\,\sigma(t_1,\ldots,t_n)=y$, that is,
$\vdash \exists y\,.\,t=y$.
The only thing left is to prove $(3)$.
We prove it similarly to (1), using (2): in the propositional calculus
property at the beginning of the proof,
pick $A$ as $\forall x\,.\,\varphi$, $B$ as $\exists y\,.\,t=y$,
and $C$ as $\varphi[t/x]$.
Then $\vdash A \wedge B \ra C$ by Functional Substitution and $\vdash B$ by
(2) above, so $\vdash A \ra C$, i.e.,
$\vdash (\forall x\,.\,\varphi) \ra \varphi[t/x]$.
\end{proof}

Our approach to obtain a sound and complete proof system for matching logic
is to build upon its reduction to predicate logic with equality in
Section~\ref{sec:PL-reduction}.
Specifically, to use Proposition~\ref{prop:mlTopred} and the complete proof system
of predicate logic with equality.
Given a matching logic signature $(S,\Sigma)$, let $(S,\Pi_\Sigma)$ be the
predicate logic (with equality) signature obtained like in
Section~\ref{sec:PL-reduction}.
Besides the $\PL$ translation there, we also define a backwards
translation $\ML$ of predicate logic with equality
$(S,\Pi_\Sigma)$-formulae into $(S,\Sigma)$-patterns:
% inductively as follows:
$$
\begin{array}{rcl}
\ML(x=r) & = & x=r \\
\ML(\pi_\sigma(r_1,\ldots,r_n,r)) & = & r \in \sigma(r_1,\ldots,r_n) \\
\ML(\neg\psi) &=& \neg\ML(\psi) \\
\ML(\psi_1 \wedge \psi_2) &=& \ML(\psi_1) \wedge \ML(\psi_2) \\
\ML(\exists x\,.\,\psi) &=& \exists x\,.\,\ML(\psi) \\[1ex]
\ML(\{\psi_1,\ldots,\psi_n\}) &=& \{\ML(\psi_1),\ldots,\ML(\psi_n)\}
\end{array}
$$
Recall from Section~\ref{sec:equality} that we assume equality and
membership in all specifications.
%, and that they can be defined
%using the definedness symbol $[\_]$.

%\begin{prop}
%\label{prop:mlpl}
%If $G$ is any set of $(S,\Pi_\Sigma)$-formulae, $\psi$ is any
%$(S,\Pi_\Sigma)$-formula, and $\varphi$ is any $(S,\Sigma)$-pattern,
%then:
%\begin{enumerate}
%\item $G \PLmodels \psi$ iff $\ML(G) \models \ML(\psi)$
%\item $\models \varphi = \ML(\PL(\varphi))$
%\item $\PLmodels \psi \lra \PL(\ML(\psi))$
%\end{enumerate}
%\end{prop}

%Figure~\ref{fig:proof-system} shows our sound and complete proof system for
%matching logic reasoning, which was specifically crafted to include the proof
%system of first-order logic.
%Indeed, the first group of axiom and rule schemas include all the axioms and proof
%rules of FOL with equality as instances (the rules Substitution, Equation
%introduction and Equation elimination allow more general patterns instead of
%terms).
%The second group of proof rules, for reasoning about membership, is introduced
%for technical reasons, namely for the proof of Theorem~\ref{thm:completeness}:
%%We have not used them so far in any of our program verification efforts using
%%matching logic, and our current matching logic prover provides no reasoning
%%support for membership.

\begin{thm}
\label{thm:completeness}
The proof system in Figure~\ref{fig:proof-system} is sound and complete:
$F\models\varphi$ iff $F\vdash\varphi$.
\end{thm}
\begin{proof}
%\grigore{Hm, Functional Substitution seems to not be necessary in its full
%generality.  Investigate this in detail, eventually.}
Propositions~\ref{prop:PL-validity} and \ref{prop:FOL-equality} showed the
soundness of all rules except for Substitution.
Corollary~\ref{cor:FOL-substitution} showed the soundness of the stronger
Term Substitution.
To show the soundness of Functional Substitution, we show
$\overline{\rho}((\forall x\,.\,\varphi) \wedge
  (\exists y\,.\,\varphi'=y)) \subseteq \overline{\rho}(\varphi[\varphi'/x])$
for any model $M$ and valuation $\rho:\Var\ra M$.
Let $s$ be the sort of $\varphi$ and $s'$ be the sort of $\varphi'$.
We have
$\overline{\rho}((\forall x\,.\,\varphi) \wedge (\exists y\,.\,\varphi'=y))
= \bigcap\{\overline{\rho'}(\varphi)\mid
\rho'\!\!\upharpoonright_{\Var\backslash\{x\}} = \rho\!\!\upharpoonright_{\Var\backslash\{x\}}
\} \mathrel\cap \bigcup\{M_s\mid
\rho'\!\!\upharpoonright_{\Var\backslash\{x\}} = \rho\!\!\upharpoonright_{\Var\backslash\{x\}},\ 
\overline{\rho'}(\varphi') = \{\rho'(y)\}\}$.
Since $y\not\in\FV(\varphi')$, it follows that
$\overline{\rho'}(\varphi') =\overline{\rho}(\varphi')$.
Therefore, all we have to show is the following:
if $\overline{\rho}(\varphi') = \{a\}$ for some $a\in M_{s'}$ then
$\bigcap\{\overline{\rho'}(\varphi)\mid
\rho'\!\!\upharpoonright_{\Var\backslash\{x\}} = \rho\!\!\upharpoonright_{\Var\backslash\{x\}}
\} \subseteq \overline{\rho}(\varphi[\varphi'/x])$.
This holds because
$\overline{\rho}(\varphi[\varphi'/x]) = \overline{\rho[a/x]}(\varphi)$.

We now show the completeness.
First, note that Proposition~\ref{prop:mlTopred} and the completeness
of predicate logic imply that
$F\models\varphi$ iff $\PL(F)\PLvdash\PL(\varphi)$.
% iff $\ML(\PL(F))\models\ML(\PL(\varphi))$.
Second, note that $\PL(F)\PLvdash\PL(\varphi)$ implies
$\ML(\PL(F))\vdash\ML(\PL(\varphi))$, because the $\ML$
translation only replaces predicates
$\pi_\sigma(r_1,\ldots,r_n,r)$ with $r \in \sigma(r_1,\ldots,r_n)$
and the proof rules of predicate logic, except for Substitution,
are a subset of the proof rules in Figure~\ref{fig:proof-system},
while the predicate logic Substitution is derivable in matching logic
((1) in Proposition~\ref{prop:substitutions}).
Third, notice that the completeness result holds if we can show
$F \vdash \varphi$ iff $F \vdash\ML(\PL(\varphi))$ for any pattern $\varphi$:
indeed, if that is the case then $F \vdash \ML(\PL(F))$, which together with
$\ML(\PL(F))\vdash\ML(\PL(\varphi))$
implies $F\vdash \ML(\PL(\varphi))$, which further implies $F \vdash \varphi$.
%,
%because we can use the Equation elimination proof rule both
%to derive each pattern in $\ML(\PL(F))$ from its corresponding pattern in $F$,
%and to derive $\varphi$ from $\ML(\PL(\varphi))$ at the end of the proof
%of $\ML(\PL(F))\vdash\ML(\PL(\varphi))$.

Let us now prove that $F \vdash \varphi$ iff $F\vdash\ML(\PL(\varphi))$
for any pattern $\varphi$.
We first show $\vdash r\in\varphi = \ML(\PL_2(\varphi,r))$ by induction on
$\varphi$.
The cases $\varphi \equiv x$, $\varphi \equiv \neg\varphi'$,
$\varphi\equiv \varphi_1 \wedge \varphi_2$, and
$\varphi\equiv \exists y.\varphi'$ are immediate consequences
of the axioms 9-12 in Figure~\ref{fig:proof-system}, using the
induction hypothesis and Equality Elimination (rule 7).
For the case $\varphi\equiv\sigma(\varphi_1,\ldots,\varphi_n)$,
we can first derive
$\vdash
\ML(\PL_2(\varphi,r)) =
  \exists r_1 \cdots \exists r_n\,.\,r_1\in\varphi_1 \wedge \cdots \wedge r_n\in\varphi_n
 \wedge r\in\sigma(r_1,\ldots,r_n)
$
using the induction hypothesis and Equality Elimination, and then
$\vdash
r\in\varphi =
  \exists r_1 \cdots \exists r_n\,.\,r_1\in\varphi_1 \wedge \cdots \wedge r_n\in\varphi_n
 \wedge r\in\sigma(r_1,\ldots,r_n)
$
using axiom 13 in Figure~\ref{fig:proof-system} and conventional FOL reasoning.
Therefore, $\vdash r\in\varphi = \ML(\PL_2(\varphi,r))$.
Our result now follows by proof rules 8 in Figure~\ref{fig:proof-system},
since $\ML(\PL(\varphi)) \equiv \forall r\,.\,\ML(\PL_2(\varphi,r))$.
\end{proof}

As an example, let us informally use the proof system in Figure~\ref{fig:proof-system}
together with the axiom patterns in Section~\ref{sec:map-patterns},
to derive 
$(1 \mapsto 5 \SLstar 2 \mapsto 0 \SLstar 7 \mapsto 9 \SLstar 8 \mapsto 1)
\ra \llist(7,9\cdot 5)$.
For simplicity, like in separation logic, let us assume that the axioms of commutativity,
associativity and idempotence of $\_\SLstar\_$ are axiom {\em schemas}, so we do not need
to explicitly use the substitution rule to instantiate them; in a mechanical verification
setting, their soundness as schemas can be proved separately from the basic axioms.

Recall the following axiom patterns about linked lists with data from
Section~\ref{sec:map-patterns}:
$$
\begin{array}{l@{\hspace*{10ex}}l}
x \mapsto [\epsilon] = \SLemp &
\llist(x,\epsilon) = (\SLemp \wedge x = 0) \\
x \mapsto [a, S] = x \mapsto a \SLstar (x+1) \mapsto [S] &
\llist(x,n\cdot S) = \exists z \,.\,x \mapsto [n,z] \SLstar \llist(z,S) \\
\end{array}
$$
Using the left axioms, axioms for sequences in Section~\ref{sec:collections},
and axioms of maps, by Functional Substitution and Equality Elimination
(Figure~\ref{fig:proof-system}) we derive
$\vdash 1 \mapsto 5 \SLstar 2 \mapsto 0 = 1 \mapsto [5,0]$
and
$\vdash 7 \mapsto 9 \SLstar 8 \mapsto 1 = 7 \mapsto [9,1]$, respectively.
By the first axiom for $\llist$ above, $\vdash \llist(0,\epsilon) = \SLemp$.
Note that Functional Substitution is equivalent to
$\vdash \varphi[\varphi'/y] \wedge (\exists y\,.\,\varphi'=y) \ra 
(\exists x\,.\,\varphi)$ (by propositional reasoning, e.g.,
$A \ra B = \neg B \ra \neg A$), so we get
$\vdash 1 \mapsto [5,0]  \SLstar \llist(0,\epsilon) \ra
(\exists z\,.\,1 \mapsto [5,z]  \SLstar \llist(z,\epsilon))$, which by the
second axiom of $\llist$ above yields
$\vdash 1 \mapsto [5,0]  \SLstar \llist(0,\epsilon) \ra \llist(1,5)$.
Following similar reasoning for the other binding, we can construct the following
(informal) proof derivation:
$$
\begin{array}{rlr}
  & 1 \mapsto 5 \SLstar 2 \mapsto 0 \SLstar 7 \mapsto 9 \SLstar 8 \mapsto 1 \\
= & 1 \mapsto [5,0] \SLstar 7 \mapsto [9,1] \\
= & 1 \mapsto [5,0]  \SLstar \llist(0,\epsilon)\SLstar 7 \mapsto [9,1]
& \textrm{(structural framing---Proposition~\ref{prop:structural-framing})} \\
\ra & (\exists z\,.\,1 \mapsto [5,z]  \SLstar \llist(z,\epsilon))\SLstar 7 \mapsto [9,1] \\
= & \llist(1,5\cdot \epsilon) \SLstar 7 \mapsto [9,1] \\
= & \llist(1,5) \SLstar 7 \mapsto [9,1] \\
\ra & \exists z\,.\,7\mapsto[9,z] \wedge \llist(z,5) \\
= & \llist(7,9\cdot 5)
\end{array}
$$
When applying structural framing
(Proposition~\ref{prop:structural-framing})
above, we assumed the completeness of the
matching logic proof system (Theorem~\ref{thm:completeness}).
It is an insightful exercise to directly prove
Proposition~\ref{prop:structural-framing} with $\vdash$ instead
of $\models$, without using the completeness theorem but only
the proof rules in Figure~\ref{fig:proof-system} (hint: use the membership
rules).

The example proof above was neither difficult nor unexpected, and it followed
almost the same steps as the corresponding separation logic proof.
Indeed, in spite of matching logic's simplicity (recall that its syntax is
even simpler than that of FOL: Definition~\ref{dfn:ML-patterns})
and domain-independence, it has the necessary expressiveness and capability
to carry out proof derivations for particular domains given as matching
logic specifications that are as abstract and intuitive as in logics
specifically crafted for those domains.
Additionally, its patterns are expressive enough to capture complex structural
and logical properties about program configurations, at the same time giving
us the peace of mind that any such properties are derivable with a uniform,
domain-independent proof system.

\section{Additional Related Work}
\label{sec:related-work}

Matching logic builds upon intuitions from and relates to at least five
important logical frameworks:
(1) {\em Relation algebra (RA)} (see, e.g., \cite{tarski1987formalization}),
noticing that our interpretations of symbols as functions to powersets are
equivalent to relations; although our interpretation of symbols captures
better the intended meaning of pattern and matching, and our proof system
is quite different from that of RA, like with FOL we expect a tight
relationship between matching logic and RA, which is left as future work;
(2) {\em Partial FOL} (see, e.g., \cite{DBLP:journals/sLogica/FarmerG00}
for a recent work and a survey), noticing that our interpretations of symbols
into powersets are more general than partial functions
(Section~\ref{sec:equality} shows how we defined definedness);
(3) {\em Separation logics (SL)} (see, e.g.,\cite{sep-logic-csl01}),
which we discussed in Section~\ref{sec:sep-logic};
and (4) Precursors of matching logic in
\cite{rosu-ellison-schulte-2010-amast,rosu-stefanescu-2011-nier-icse,rosu-stefanescu-2012-oopsla,rosu-stefanescu-2012-fm,rosu-stefanescu-ciobaca-moore-2013-lics,stefanescu-ciobaca-mereuta-moore-serbanuta-rosu-2014-rta},
%which we briefly discuss below.
%
%
%Furthermore, earlier works on separation
%logics~\cite{sep-logic-csl01,reynolds-02,OHearn99thelogic} can also be regarded
%as precursors of matching logic.
%%, because separation logic can be framed as a
%%particular matching logic theory in a particular model.
%Finally, since we interpret symbols as arbitrary relations, which include
%partial functions, matching logic also generalizes
%work on {\em partial FOL} (see \cite{DBLP:journals/sLogica/FarmerG00}
%for a more recent work on this topic and a survey).
%In light of the above, ``relational FOL'' could have been
%an appropriate name for our new logic; however,
%we prefer to keep the original name ``matching logic'' of the old
%variant, because we believe that it better captures the intuition that we
%specify and reason about state properties in terms of patterns and their
%matching.
which proposed the pattern idea by
extending FOL with particular ``configuration'' terms (grayed box below is the
only change to FOL):
$$
\begin{array}{rcl}
t_s & ::= & x \in \Var_s \ \ \mid \ \ \sigma(t_1,\ldots,t_n)
  \ \mbox{ with } \sigma\in\Sigma_{s_1 \ldots s_n,s} \\
\varphi & ::= & \pi(x_1,\ldots,x_n)
  \ \mbox{ with }\pi\in\Pi_{s_1 \ldots s_n} \ \ \mid \ \ 
 \neg \varphi \ \ \mid \ \ \varphi \wedge \varphi
 \ \ \mid \ \ \exists x . \varphi \\
& \mid &
\graybox{14ex}{ $t \in T_{\Sigma,{\it Cfg}}(X)$ \vspace*{-3ex}}
\end{array}
$$
where $T_{\Sigma,{\it Cfg}}(X)$ is the set of terms of a special sort
$\it Cfg$ (from ``configurations'') over variables in set $X$.
To avoid terminology conflicts, we here strengthen the proposal in
\cite{rosu-2015-rta}
to call the variant above {\em topmost matching logic} from here on.
Topmost matching logic can trivially be desugared into FOL with equality
by regarding a particular pattern predicate $t \in T_{\Sigma,{\it Cfg}}(X)$
as syntactic sugar for
``(current state/configuration is) equal to $t$'', i.e., $\square=t$.
%, so a specialized proof system was not necessary.
%Then conventional FOL provers can be used for matching logic reasoning.
One major limitation of topmost matching logic, which motivated
the generalization in \cite{rosu-2015-rta} with full details added in this paper,
is that its restriction to patterns of sort {\it Cfg} prevented us to
define local patterns (e.g., the heap list pattern) and perform local
reasoning.
They had to be defined globally, as patterns of sort $\it Cfg$ with
structural frames for everything else except their target cell
(e.g., the heap), which was not only more verbose but also less modular.

The basic idea of regarding terms with variables as sentences/patterns
that are satisfied/matched by ground terms, goes back
to \cite{DBLP:journals/jar/LassezM87}
and it was further studied in
\cite{DBLP:conf/mfcs/LassezMM91,DBLP:conf/rta/Tajine93,DBLP:journals/jsc/Fernandez98,DBLP:journals/tcs/Pichler03,DBLP:conf/lopstr/MeseguerS15}.
Furthermore, terms enriched with Boolean conditions over their variables,
called {\em constrained terms}, were studied in
\cite{DBLP:journals/scp/CholewaEM15}, together with their relation to
narrowing.
These approaches allow certain Boolean algebra operations to be applied to
patterns, and study the expressiveness of such operations w.r.t.\ the
languages of ground terms that they define, in particular conditions
under which negation can be eliminated.
In addition to Boolean algebra operations and conditions on terms with
variables, matching logic also allows quantification over variables,
as well as using the resulting patterns nested inside other
patterns.
The richer syntax of patterns in matching logic is motivated by needs
to specify complex structures with mixed constraints over program
configurations, as shown in Section~\ref{sec:example}.
Also, matching logic allows models with any data, not only term models,
interprets symbols as relations with the axiomatic capability to constrain
them as functions, and organizes the patterns and their
models in a logic that admits a sound and complete proof system.

The idea of regarding terms as patterns is also reminiscent of
{\em pattern calculus}~\cite{Jay:2004:PC:1034774.1034775}, although
note that matching logic's patterns are intended to express and reason about
static properties of data-structures or program configurations, while pattern
calculi are aimed at generally and compactly expressing computations and
dynamic behaviors of systems.
So far we used rewriting to define dynamic semantics; it would
be interesting to explore the combination of pattern calculus and matching
logic for language semantics and reasoning.

\section{Conclusion and Future Work}
\label{sec:conclusion}

Matching logic is a sound and complete FOL variant
that makes no distinction between function and predicate symbols.
Its formulae, called patterns, mix symbols, logical connectives
and quantifiers, and evaluate in models to sets of values, those that
``match'' them, instead of just one value as terms do or a truth value as
predicates do in FOL.
Equality can be defined and several important variants
of FOL fall as special fragments.
Separation logic can be framed as a matching logic theory within the
particular model of partial finite-domain maps, and heap patterns can
be elegantly specified using equations.
Matching logic allows spatial specification and reasoning anywhere in a
program configuration, and for any language, not only in the heap or other
particular and fixed semantic components.
% hardwired within the logic.
%An elimination theorem shows that matching logic can be translated
%to predicate logic with equality, and a sound and complete proof system
%including that of FOL is given.

We made no efforts to minimize the number of rules in our proof system
(Figure~\ref{fig:proof-system}), because our main objective in this paper was
to include the proof system for FOL with equality as part of our proof system,
to indicate that conventional reasoning remains valid and thus automated provers
can be used unchanged.
It is likely, however, that a minimal proof system working directly with
the definedness symbols $\lceil\_\rceil$ can be obtained such that the
equality and membership axioms and rules in Figure~\ref{fig:proof-system}
can be proved as lemmas.
%, but that was not our objective here.

Our completeness result in Section~\ref{sec:deduction} relies heavily on
equality and on membership patterns, whose definitions require the existence
of the definedness symbols $\lceil\_\rceil$.
On the other hand, Proposition~\ref{prop:mlTopred} translates arbitrary
matching logic validity to validity in predicate logic with equality, even
when there are no definedness symbols.
Since predicate logic with equality admits complete deduction, we conjecture
that matching logic must admit an alternative complete proof system which
does not rely on definedness symbols.

We have not discussed any computationally effective fragments of matching
logic or heuristics to automate matching logic deduction.
These are crucial for the development of practical provers and program
verifiers.
The systematic study of such fragments and heuristics is left for future work.
Also, complexity results in the style
of~\cite{Calcagno2001,joel-2014-fossacs,Brotherston:2014:DPS:2603088.2603091,DBLP:conf/cade/IosifRS13}
for separation logic can likely also be obtained for fragments of matching
logic.

Many of the results related to localizing/globalizing reasoning, such as
Propositions~\ref{prop:structural-framing}, \ref{prop:constraint-propagation}, and \ref{prop:symbol-distributivity}, extend to monotone/positive contexts, that
is, to ones without negations on the path to the placeholder.
While non-monotonic contexts do not seem to occur frequently in program
verification efforts, it would nevertheless be worthwhile investigating
techniques for the elimination of negation, likely generalizing those in
\cite{DBLP:conf/mfcs/LassezMM91,DBLP:conf/rta/Tajine93,DBLP:journals/jsc/Fernandez98},
or intuitionistic variants of matching logic where negation is not
allowed at all in patterns.

Finally, the main application of matching logic so far was as a pattern
language for reachability logic~\cite{stefanescu-park-yuwen-li-rosu-2016-oopsla,stefanescu-ciobaca-mereuta-moore-serbanuta-rosu-2014-rta,rosu-stefanescu-ciobaca-moore-2013-lics,rosu-stefanescu-2012-oopsla}, where reachability rules, which
are pairs of patterns $\varphi \Ra \varphi'$, can be used to specify both operational
semantics rules and properties to prove about programs.
Reachability logic has its own (language independent) sound and relatively complete
proof system.
We conjecture that we can capture reachability logic as an instance of matching
logic, too, in a similar vein to how we did it for modal logic in Section~\ref{sec:modal-logic}: add some new symbols with their (axiomatized) semantics and then prove
the proof rules of reachability logic as lemmas/corollaries.
For example, we can extend the world models $M$ in Section~\ref{sec:modal-logic} with
a Kripke transition relation $w\,R\,w'$ by adding a symbol
$\circ\_\in\Sigma_{\World,\World}$ and assuming $w\,R\,w'$ iff
$w \in \circ_M(w')$, then define $\Diamond$ and other CTL or even CTL$^*$ operators
as least fixed points, and finally the reachability rules as sugar.

\textbf{Acknowledgments.}
This is an extended version of an RTA'15 invited paper \cite{rosu-2015-rta}.
The author warmly thanks the RTA'15 program committee for the invitation and
to the anonymous reviewers. 
The author also expresses his deepest thanks to the \K team (\url{http://kframework.org}),
who share the belief that programming languages should have only one semantics,
which should be executable, and formal analysis tools, including fully fledged
deductive program verifiers, should be obtained from such semantics at little or
no extra cost.
I would like to also warmly thank the following colleagues and friends for their
comments and criticisms on previous drafts of this paper:
Nikolaj Bjorner,
Xiaohong Chen,
Claudia-Elena Chiri\c{t}\u{a},
Maribel Fern{\' a}ndez,
Ioana Leu\c{s}tean,
Dorel Lucanu,
Jos{\' e} Meseguer,
Brandon Moore,
Daejun Park,
Cosmin R\u{a}doi,
Traian Florin \c{S}erb\u{a}nu\c{t}\u{a},
and 
Andrei \c{S}tef\u{a}nescu.

\grigore{Make sure you don't forget the material after the references.}

\bibliographystyle{abbrv}
\bibliography{refs}
\vspace{-40 pt}
\end{document}